\documentclass{amsart}
\pdfoutput=1

\usepackage{amssymb}
\usepackage{amsthm}
\usepackage{amsfonts}
\usepackage{amsmath}
\usepackage{bbm}
\usepackage{bm}

\usepackage{pmboxdraw}
\usepackage{verbatim} 
\usepackage{graphicx}
\usepackage{color}
\usepackage[colorlinks=true, citecolor=blue, filecolor=black, linkcolor=black, urlcolor=black]{hyperref}
\usepackage{cite}
\usepackage[normalem]{ulem}
\usepackage{subcaption}
\usepackage{mathtools}
\usepackage{kantlipsum}
\allowdisplaybreaks
\usepackage{enumitem} 
\usepackage{tikz}
\usepackage{pgfplots}
\pgfplotsset{compat=1.18}
\usepackage{todonotes}
\usetikzlibrary{shapes.geometric, arrows.meta, positioning, patterns}

\usetikzlibrary{patterns}

\pgfdeclarepatternformonly{triangular dots}
{\pgfpointorigin} 
{\pgfpoint{4pt}{6.868pt}} 
{\pgfpoint{4pt}{6.868pt}} 
{
\pgfpathcircle{\pgfpoint{1pt}{1.732pt}}{0.7pt}
\pgfpathcircle{\pgfpoint{3pt}{5.196pt}}{0.7pt}
\pgfusepath{fill}
}

\reversemarginpar

\usepackage{xcolor}

\hypersetup{
    colorlinks=true,
    linkcolor=black,
    filecolor=magenta,      
    urlcolor=cyan,
    pdfpagemode=FullScreen,
    }

\newcommand{\RN}[1]{%
	\textup{\uppercase\expandafter{\romannumeral#1}}%
}

\newcommand{\ds}{\displaystyle}

\def\sm{\setminus}

\def\wh{\widehat}
\def\wt{\widetilde}

\DeclareMathOperator{\Res}{Res}

\def\AA{\mathcal{A}}
\def\BB{\mathcal{B}}

\def\EE{\mathcal{E}}
\def\FF{\mathcal{F}}

\def\II{\mathcal{I}}
\def\LL{\mathcal{L}}

\def\ZZ{\mathcal{Z}}

\def\QQ{\mathcal{Q}}
\def\SS{\mathcal{S}}

\def\C{\mathbb{C}}
\def\D{\mathbb{D}}

\def\P{\mathbf{P}}
\def\R{\mathbb{R}}

\def\T{\mathbb{T}}

\def\n{\mathbf{n}}
\def\t{\mathbf{t}}
\def\K{\mathbf{K}}

\newcommand{\re}{\operatorname{Re}}
\newcommand{\im}{\operatorname{Im}}
\newcommand{\ud}{\, \mathrm{d}}
\newcommand{\interior}{\operatorname{Int}}
\newcommand{\exterior}{\operatorname{Ext}}

\theoremstyle{plain}
\newtheorem{thm}{Theorem}[section]

\newtheorem{lem}[thm]{Lemma}

\newtheorem{prop}[thm]{Proposition}

\theoremstyle{remark}
\newtheorem{rem}{Remark}[section]

\numberwithin{equation}{section}
\def\alert#1{\textcolor{red}{#1}}

\usepackage{tcolorbox}
\tcbuselibrary{breakable}
\usepackage{colortbl}\newtcolorbox{defn}[1]{breakable,
colbacktitle=gray!50!white, fonttitle=\bfseries, coltitle=black, title=Definition: {#1}}

\begin{document}

\title[Free energy expansion of determinantal Coulomb gases]{Free energy expansion of determinantal Coulomb gases \\ in the quadratic fields with a point charge}
\author{Sung-Soo Byun}
\address{Department of Mathematical Sciences and Research Institute of Mathematics, Seoul National University, Seoul 08826, Republic of Korea}
\email{sungsoobyun@snu.ac.kr}

\author{Meng Yang}
\address{School of Mathematics and Statistics, HNP-LAMA, Central South University, Changsha 410083, China}
\email{mengyangrmt@csu.edu.cn}

\author{Eui Yoo}
\address{Department of Mathematical Sciences, Seoul National University, Seoul 08826, Republic of Korea}
\email{yysh0227@snu.ac.kr}

\begin{abstract}  
We study a determinantal Coulomb gas in the complex plane associated with the external potential
$$
Q(z)=\frac{1}{1-\tau^2}\bigl(|z|^2-\tau \re z^2\bigr)-2c\log|z-a|,
$$
where $\tau\in[0,1)$, $c\ge0$, and $a\ge0$. In the regimes where the associated droplet is simply or doubly connected, we derive the free energy expansion up to and including the constant term, with all coefficients computed explicitly, thereby extending recent results in the isotropic case $\tau=0$. In particular, we identify the constant term with the Liouville action associated with the droplet. Our result admits a natural interpretation in terms of asymptotic expansions of moments of characteristic polynomials for the elliptic Ginibre ensemble. The proof is based on a deformation framework involving both the singularity location $a$ and the anisotropy parameter $\tau$, relating variations of the free energy to refined asymptotics of planar orthogonal polynomials. The asymptotic analysis relies on the foliation flow method of Hedenmalm and Wennman, providing an alternative to the Riemann--Hilbert approach used in the isotropic setting. The present work suggests a general framework connecting free energy expansions, refined asymptotics of planar orthogonal polynomials, and conformally invariant geometric functionals, with several intermediate results already formulated for general algebraic Hele-Shaw potentials.
\end{abstract}

\maketitle


\section{Introduction}

One of the distinctive features of probabilistic and statistical physics models in two dimensions is the emergence of conformal geometry as an underlying signature. It often governs the large-scale behaviour of these systems and provides a natural framework for describing the emergent structures of interacting particle systems. A prominent example is given by two-dimensional Coulomb gases \cite{Se24,Fo10,BF25}, particularly in the study of their free energy expansion. 

By definition, for a given external potential $Q:\C \to (-\infty, +\infty]$ satisfying suitable assumptions, the determinantal Coulomb gas in the plane (often referred to as a random normal matrix ensemble) is a system of particles $\bm{z}=\{z_j\}_{j=1}^N \subset \C$ whose joint probability density function is defined by 
\begin{equation}\label{eq_RNM}
    \ud \mathbb{P}_N(\bm{z}) = \frac{1}{\ZZ_{N}(Q)}\prod_{1\le j<k\le N}|z_j-z_k|^2 \prod_{j=1}^N e^{-NQ(z_j)}\ud A(z_j),
\end{equation}  
where $\ud A = \ud^2 z / \pi$ is the area measure on the complex plane. Here, the normalising constant $\ZZ_N(Q)$ is referred to as the \textit{partition function}. For this model, the equilibrium measure plays a central role as the object governing the conformal geometric structure of the system at the macroscopic scale. It characterises the large-scale limit of \eqref{eq_RNM}, and, by Frostman's theorem, it takes the form
\begin{align} \label{eq for Frostman}
    \ud\mu_Q = \Delta Q \cdot \mathbbm{1}_{K_Q}\,\ud A, \qquad (\Delta := \partial \bar{\partial}),
\end{align}
where $K_Q$ is a compact subset of the complex plane, known as the \textit{droplet}. The equilibrium measure can be characterised as the unique minimiser of the energy among probability measures 
\begin{align} \label{def of log energy}
    I_Q[\mu]:=\int_{\mathbb{C}^2}\log \frac{1}{|z- w|}\ud \mu(z)\ud \mu(w) + \int_\mathbb{C}Q(z) \ud \mu(z).
\end{align}

A fundamental problem in Coulomb gas theory is to understand the large-$N$ expansion of the free energy $\log \ZZ_N(Q)$, which is believed to take the form
\begin{equation} \label{def of free energy expansion general}
\log \ZZ_N(Q) \sim C_1 N^2 +C_2 N \log N + C_3 N + C_4 \log N+ C_5, 
\end{equation}
as $N \to \infty$. (Here, one assumes that the potential \(Q\) is regular in the sense that it has no logarithmic or jump-type singularities inside the droplet, and that the associated droplet is connected with a smooth boundary; see \cite{By25a} for a more detailed discussion.) Such an expansion plays a central role in understanding the behaviour of Coulomb gases.  
For instance, the leading-order term $C_1$ is essential to determine the law of large numbers for the empirical measure \cite{HM13,Joh98}. Moreover, to capture fluctuations of linear statistics supported in the bulk of the droplet, one needs to identify the expansion up to the term $C_3$ \cite{Se23}. 

For general Coulomb gases, the first three coefficients $C_1$, $C_2$, and $C_3$ were identified in \cite{LS17}; see also \cite{BBNY19,AS21,Se23}. These terms admit a potential-theoretic interpretation and are given by 
\begin{equation}
C_1= -I_Q[\mu_Q], \qquad C_2= \frac12, \qquad C_3= \frac{ \log(2\pi) }{2} -1-\frac12 \int_\C \log(\Delta Q) \ud \mu_Q. 
\end{equation}
Here, the coefficients $C_1, C_2,$ and $C_3$ represent, respectively, the macroscopic mean-field energy, a universal self-energy contribution at the microscopic scale, and a subleading term encoding entropy effects. Note that, in these first three terms, the geometric and topological structure of the equilibrium measure does not play a significant role. In contrast, a key feature of the remaining terms $C_4$ and $C_5$ is that they depend crucially on the underlying geometry and topology of the droplet. 

Over the past decades, several conjectures have been proposed for the terms $C_4$ and $C_5$. The coefficient $C_4$ is a topological term, conjectured in \cite{JMP94} to be given by 
\begin{equation}
C_4= \frac12 -\frac{\chi}{12},
\end{equation}
where $\chi$ denotes the Euler characteristic of the droplet; see also \cite{CFTW15,TF99}. 
On the other hand, the constant term $C_5$, often regarded as the most significant contribution in the physics literature—particularly from the perspective of conformal field theory \cite{ZW06}—is expected to be related to the spectral determinant of the droplet. While its general form is quite involved, in the case where $\Delta Q$ is constant so that the equilibrium measure is uniform on its support, it was conjectured that
\begin{equation}
C_5 = \chi \, \zeta'(-1) + \frac{\log(2\pi)}{2} -\frac12 \log \textup{det}_\zeta ( \Delta_{ \C \setminus K_Q } ),
\end{equation}
where $\zeta$ is the Riemann zeta function and $\textup{det}_\zeta$ denotes the regularised spectral determinant.  
We refer to \cite[Section 5.3]{BF25}, \cite[Section 9.3]{Se24} and introductions of \cite{BKS23,BSY25,ACC26,Rou25,By25a} for more on this conjecture. In the next section, we examine this conjecture in detail, replacing the spectral determinant with the Liouville action in the characterisation of $C_5$, as the latter admits a more transparent description.  

In recent years, the conjectured free energy expansion has been the subject of extensive study. A common starting point is the representation of the partition function as a structured determinant of the form 
\begin{equation} \label{ZN in terms of det}
\ZZ_N(Q)= N! \det \bigg[ \int_{\C} z^{j}\bar{z}^k   e^{-NQ(z)} \, \ud A(z) \bigg]_{j,k=0}^{N-1},
\end{equation}
which follows from the fact that \eqref{eq_RNM} forms a determinantal point process.
An effective way to analyse the structured determinant \eqref{ZN in terms of det} is to work on an orthogonal polynomial basis. More precisely, let $\{P_{j,N}\}_{j=0}^{N-1}$ be the family of monic planar orthogonal polynomials satisfying 
\begin{equation} \label{def of planar OP}
\int_\C P_{j,N}(z) \overline{P_{k,N}(z)} \, e^{-NQ(z)}\,\ud A(z) = h_{j,N} \,\delta_{j,k},\quad j,k=0,1,\ldots, N-1,
\end{equation} 
where $h_{j,N}$ is the squared norm of $P_{j,N}$.
With this choice of basis, it follows that 
\begin{equation} \label{ZN in terms of prod norms}
\ZZ_N(Q) = N! \prod_{j=0}^{N-1} h_{j,N}.
\end{equation}
This representation naturally brings tools from the theory of orthogonal polynomials into the analysis.

In recent studies of free energy expansions based on the representation \eqref{ZN in terms of det}, the main strategies can be broadly classified into three complementary approaches: \textbf{exact solvability}, \textbf{duality}, and \textbf{deformation of the free energy}.  
These approaches are implemented depending on the model and the objectives of the analysis, and together provide a flexible framework for studying the problem. We now describe each of them in more detail; see Table~\ref{Table_approaches} for a summary.

\begin{table}[h!]
 \begin{center}
\renewcommand{\arraystretch}{1.5}
\begin{tabular}{|c|c|}
\hline
\textbf{Approach} & \textbf{Key idea} \\
\hline
Exact solvability & Explicit and semi-explicit evaluations of integrals; cf. \cite{BF25,BKS23,BKSY25,ACC26} \\
\hline
Duality & Ensemble interrelations and alternative observables; cf. \cite{BSY25,BCMS25,Fo25}  \\
\hline
Deformation of the free energy & Variational analysis and the reference free energy; cf. \cite{BSY25,Rou25} \\
\hline
\end{tabular}
\end{center}
    \caption{Summary of complementary approaches}
    \label{Table_approaches}
\end{table}

\noindent \textbf{(A) Exact solvability.} The first approach, based on exact solvability, is particularly effective for many well-studied models in non-Hermitian random matrix theory, especially various extensions of the complex Ginibre ensemble \cite{BF25}. 

It is well known that the eigenvalues of the complex Ginibre matrix follow the law \eqref{eq_RNM} with the Gaussian potential $Q^{\rm g}(z):=|z|^2$. 
Several variants of the Ginibre ensemble include potentials of the form
\begin{equation} \label{def of Q induced elliptic}
Q^{ \rm i }(z):=|z|^2-2c \log|z|,\qquad Q^{ \rm e }(z):= \frac{|z|^2-\tau \re z^2}{1-\tau^2},
\end{equation}
where $c \ge 0$ and $\tau \in [0,1)$. These potentials arise from the induced and elliptic Ginibre ensembles, respectively; see \cite[Sections 2.3 and 2.4]{BF25}. Further examples include the truncated unitary ensemble and the spherical ensemble; see \cite[Sections 2.5 and 2.6]{BF25}.  
For all such exactly solvable models, the partition function can be expressed in terms of classical special functions—more precisely, in terms of the Barnes $G$-function \cite[Section 5.17]{NIST}. The large-$N$ expansion then follows from classical asymptotic analysis of these functions. This mechanism is closely analogous to Selberg integral-type evaluations in one dimension; see \cite{FW08}.

Such an approach has been systematically developed in recent works \cite{BKS23,BKSY25,ACC26} for radially symmetric potentials of the form $Q(z)=Q(|z|)$. In this setting, the matrix on the right-hand side of \eqref{ZN in terms of det} is already diagonal, and the integrals reduce to their radial components. As a result, the partition function can be analysed using the Laplace method combined with Euler--Maclaurin type expansions.  
Nevertheless, additional care is required, since the free energy expansion depends sensitively on the geometry of the droplet (which, in this case, is itself radially symmetric). In \cite{BKS23}, the cases $\chi=0,1$ (corresponding to the disc and the annulus) were treated. This was further extended in \cite{ACC26} to configurations where the droplet consists of concentric annuli, in which case non-trivial oscillatory terms appear as well. Moreover, in \cite{BKSY25}, the spherical geometry was considered, and the conjectural free energy expansion for $\chi=2$ was obtained.

\smallskip 

\noindent \textbf{(B) Duality.} The second approach is based on certain ensemble interrelations, often referred to as dualities in random matrix theory; see \cite{Fo25} for a comprehensive review. This perspective has its origins, in part, in the supersymmetry method; see e.g. \cite{Fyo18}.  
The central idea is that, for certain models, the associated free energy can be reformulated into alternative integral representations that are more amenable to analysis. 

A prominent example in this approach is the potential
\begin{equation} \label{def of Q GinUE with insertion}
Q_a^{\mathrm{i}}(z) = |z|^2 - 2c \log|z-a|,
\end{equation}
where $c, a \ge 0$. This model naturally arises in the study of moments of characteristic polynomials of the complex Ginibre ensemble. The associated droplet exhibits a topological phase transition \cite{BBLM15}: depending on the parameters $c$ and $a$, it may be either simply or doubly connected, with the critical regime characterised by the formation of a double point.

A key feature of this model is a duality relation showing that, up to an explicit multiplicative constant, the corresponding partition function can be identified with the distribution of the extremal eigenvalue of the Laguerre unitary ensemble. Exploiting this connection, the free energy expansion at criticality was derived in \cite[Proposition 2.5]{BSY25}, revealing the emergence of the Tracy--Widom distribution. 

This approach has since been extended in \cite{BCMS25,BFL25} to generalisations of the potential \eqref{def of Q GinUE with insertion} arising in the truncated unitary and induced spherical ensembles. In both cases, the dual model is given by the Jacobi unitary ensemble. Owing to the well-developed asymptotic theory of the associated Hankel determinants \cite{CG21}, one can again derive the free energy expansion to this order of precision. Furthermore, such duality relations have been extended and are now understood to connect these models with observables in integrable probability, such as the last passage time in geometric last passage percolation \cite{BCMS25}, in line with the philosophy of the seminal work \cite{Joh00}. 

\smallskip 

\noindent \textbf{(C) Deformation of the free energy.} The third approach, based on the deformation of the free energy, is one of the most classical and important, yet still actively developing, methods for analysing free energies and conformal invariants in statistical physics. It also constitutes the central framework of the present paper.

From a physical perspective, this approach is related to the Ward identities in conformal field theory, which capture the variations of conformal fields under the insertion of the stress energy tensor; see e.g. \cite[Lecture 5, Appendix 6]{KM13}. In particular, the deformation of the spectral determinant is governed by the Polyakov--Alvarez formula; see e.g. \cite[Proposition 2.3]{Dub09}. From a mathematical viewpoint, this method has been extensively developed, especially in the context of Hermitian random matrix theory; see e.g. \cite{CK15, Joh98, Kra07, DIZ97}. Related ideas also appear in the so-called transport method for general Coulomb gases; see \cite{Se24}.

The core idea is to compute the variation of the free energy along a suitable deformation, and then recover the free energy itself by integration. More precisely, the method consists of the following two steps:
\begin{itemize}
    \item[(i)] \textit{Identify a deformation for which the large-$N$ expansion of the free energy variation can be computed;}
    \smallskip
    \item[(ii)] \textit{Determine an exactly solvable reference free energy serving as the integration constant.}
\end{itemize}
In the present two-dimensional setting, where our goal is to derive the precise asymptotic expansion up to the constant term $C_5$, a prototypical example of this approach is again provided by the potential \eqref{def of Q GinUE with insertion}. In this case, the main idea in \cite{BSY25} is to perform the deformation with respect to the position parameter $a \ge 0$, so that the resulting variation can be identified with the fine asymptotic behaviour of the associated orthogonal polynomials.
The reference free energy is then obtained from a tractable model exploiting radial symmetry. The asymptotics of the corresponding planar orthogonal polynomials are analysed via the Riemann--Hilbert method, which in turn yields the free energy expansion; see \cite[Theorem 2.2]{BSY25}.

Such variational ideas have also been developed further in the recent work \cite{Rou25}. There, instead of the point charge insertion at a single point $a$, corresponding to the term $\log|z-a|$ in \eqref{def of Q GinUE with insertion}, one considers a screened version, namely a globally distributed potential obtained by integrating the point charge. In this framework, the variation induced by moving holes within the droplet---while preserving its topology---is analysed using correlation energy together with a mean-field approximation.

\medskip 

While the above three approaches have been primarily discussed in the context of free energy expansions in the regular setting~\eqref{def of free energy expansion general}, there is also a substantial body of work devoted to related problems in the presence of two-dimensional Fisher--Hartwig or hard-wall type singularities. Recent developments in the Fisher--Hartwig (logarithmic singularity) setting include \cite{DS22,WW19,By25a,DMMS25,BDHK25}, whereas hard-wall type constraints have been studied in \cite{AFLS25,ACCL24,ACCL23,Ch22,Ch23,No25,BP26,BC25,AL25,JV26}. In addition, we note that there exists another formulation of the free energy expansion, often referred to as the geometric Zabrodin--Wiegmann conjecture. In this setting, one considers ensembles on Riemann surfaces without boundary, which has also been extensively studied in the literature; see e.g. \cite{KMMW17,SY25,Bo25} and references therein.

\medskip

In this work, we extend the results of \cite{BSY25} to a class of models with more general quadratic potentials \eqref{potQ} below. This setting is both physically natural and of independent interest in non-Hermitian random matrix theory, and it exhibits a richer geometric structure than that considered in \cite{BSY25,BBLM15,KLY25} for the case $\tau=0$. 

In contrast to the case $\tau = 0$, the duality method is no longer applicable; cf. Remark~\ref{Rem_duality}. We therefore proceed by employing a deformation-based approach, suitably extended to accommodate the present setting. In particular, the deformation with respect to the position parameter used in \cite{BSY25} is no longer sufficient, and we introduce a new class of deformations adapted to the present model. 

From the perspective of asymptotic analysis, the deformation arguments in~\cite{BSY25} were carried out largely at a computational level, without an explicit identification of the underlying conformal geometry. By contrast, the present approach systematically tracks the geometric content of each variation, thereby uncovering the conformal structure governing the expansion. 
Furthermore, instead of relying on Riemann--Hilbert techniques, we make use of the recent general theory of planar orthogonal polynomials developed by Hedenmalm and Wennman~\cite{Hed24,HW21,HW24}, implemented here in a fully quantitative form; see Theorem~\ref{thm_F0F1}. Altogether, the method developed in the present work points toward a general framework relating free energy expansions, refined asymptotics of planar orthogonal polynomials, and conformally invariant geometric functionals.  

The main results are presented in detail in the next section.

\section{Main results and discussions}

In this section, we present our main results. In particular, Theorem~\ref{thm_ZW} provides the free energy expansion for the potential \eqref{potQ}, defining a Coulomb gas in a general quadratic field with a point charge.  
Subsection~\ref{Subsection_main results} introduces the model and the results in detail, followed by remarks on related work. Subsection~\ref{Subsection_main strategy} outlines the proof strategy, including the deformation of the free energy (Proposition~\ref{prop_key}), which we regard as a key result of independent interest.

As mentioned in the previous section, both the results and their proofs are fundamentally driven by conformal geometric computations. We begin by introducing several notions from conformal geometry that will play a central role in their formulation. 

We fix notation used throughout the paper. We write $\wh{\C} = \C \cup \{\infty\}$, $\D = \{ |w| < 1 \}$, $\D_e = \{ |w| > 1 \}$, and $\T = \{ |w| = 1 \}$.  
For a Jordan curve $\Gamma \subset \C$, we denote by $\interior(\Gamma)$ the bounded component and $\exterior(\Gamma)$ the unbounded component of its complement. We also use notation $D_e = \wh{\C} \sm \bar{D}$ for the unbounded component when a Jordan domain $D\subset \C$ is given. 

Let $D \subset \C$ be a Jordan domain. In what follows, we will make use of both the interior and exterior conformal maps associated with $D$. 
The \emph{normalised exterior conformal map} is the unique conformal map
\begin{equation} \label{def of psi exterior}
    \psi: \mathbb{D}_e \to D_e, \qquad \psi(\infty) = \infty, \qquad \psi'(\infty) > 0.
\end{equation}
On the other hand, a \emph{normalised interior conformal map} is a conformal map satisfying
\begin{equation} \label{def of psi interior}
    \psi: \mathbb{D} \to D, \qquad \psi'(0) > 0.
\end{equation}

The Liouville action is a fundamental functional in Liouville field theory and a central object in conformal geometry, where it arises naturally in the study of conformal maps and their geometric properties. In this work, we employ a version of the Liouville action adapted to our setting, defined as follows.

\begin{defn}{Liouville action}
For a normalised exterior conformal map $\psi:\D_e \to \C$, we define the \emph{Liouville action} of $\psi$ by 
\begin{equation}\label{def of Liouville action}
    \SS[\psi] := \int_{\D_e} \Big|\frac{\psi''(w)}{\psi'(w)}\Big|^2\ud A(w) - 4\log |\psi'(\infty)|.
\end{equation}
Similarly, for a normalised interior conformal map $\psi:\D \to \C$, we define 
\begin{equation}\label{def of Liouville action int}
    \SS[\psi] := \int_{\D} \Big|\frac{\psi''(w)}{\psi'(w)}\Big|^2\ud A(w) + 4\log |\psi'(0)|.
\end{equation}

Let $K \subset \C$ be a connected compact set whose boundary consists of finitely many Jordan curves. Let $\Omega_1, \ldots, \Omega_d$ denote the bounded components of $K^c$, and let $\Omega_\infty$ denote the unbounded component.  
We define the \emph{Liouville action} of $K$ by
\begin{equation}\label{def of Liouville action domain}
    \LL[K] := \SS[\psi_\infty]+ \sum_{j=1}^d \SS[\psi_j],
\end{equation}
where $\psi_j$ are normalised interior conformal maps of $\Omega_j$ and $\psi_\infty$ is the normalised exterior conformal map of $\Omega_\infty$. 
\end{defn}

By the Polyakov--Alvarez conformal anomaly formula, which describes the variation of the spectral determinant under conformal changes of the metric, the Liouville action naturally emerges as the governing functional. 
However, whereas the spectral determinant requires delicate regularisation, particularly on unbounded domains, the Liouville action admits a direct formulation in terms of uniformising conformal maps. We therefore state our results in terms of the Liouville action, albeit the spectral determinant is more common in the physics literature.

We also note that the Liouville action considered here is closely related to the universal Liouville action in the framework of Teichmüller geometry \cite{TT03,TT06}.
Its connection to the Löwner energy of Weil--Petersson class loops appears in~\cite{SW24, Wa19b}, which was also identified as the constant term of free energy expansion of determinantal Coulomb gases on Jordan curves~\cite{WZ22, CJ23,Joh22}; see also~\cite{JV26} for results on Coulomb gases on Jordan domains.

\subsection{Main results} \label{Subsection_main results}

In this work, we focus on the potential of the form 
\begin{equation}\label{potQ}
    Q(z) := \frac{1}{1-\tau^2}(|z|^2 -\tau \re z^2)  - 2c \log |z-a|,
\end{equation}
where $\tau \in [0,1)$, $ c \ge 0$ and $a \ge 0$. Note that when $c=0$, the potential reduces to $Q^{\mathrm{e}}$ in~\eqref{def of Q induced elliptic}, whereas when $\tau=0$, it reduces to the potential $Q_a^{\mathrm{i}}$ in~\eqref{def of Q GinUE with insertion}. 

Compared to the potential~\eqref{def of Q GinUE with insertion} studied in \cite{BSY25}, the potential \eqref{potQ} admits two complementary interpretations. From the viewpoint of potential theory and electrostatics \cite{BF25a}, it corresponds to replacing the radially symmetric term $x^2 + y^2$ (with $z = x + iy$) by a more general anisotropic quadratic form $\frac{x^2}{1+\tau} + \frac{y^2}{1-\tau}$. From the perspective of non-Hermitian random matrix theory \cite{BF25}, this potential arises naturally from the elliptic deformation of the complex Ginibre ensemble, where the parameter $\tau$ quantifies the degree of non-Hermiticity, and bridges one- and two-dimensional models.  

It follows from \eqref{eq for Frostman} that the equilibrium measure $\mu_Q$ is of the form
\begin{align}\label{eq:eqmeasure}
    \ud \mu_Q(z)  = \frac{1}{1-\tau^2}\mathbbm{1}_{K_Q}(z) \ud A(z),
\end{align}
where $K_Q$ denotes the associated droplet. The potential \eqref{potQ} belongs to the class of algebraic Hele-Shaw potentials, for which the theory of quadrature domains developed in \cite{LM16} provides an effective framework to analyse the geometry and regularity of the droplet; see also~\cite[Section 2]{BY25} for further background and \cite{MR25} for more recent development.  

In \cite{BY25}, a precise characterisation of the droplet $K_Q$ was obtained. In particular, it was shown in \cite[Theorem 1.1]{BY25} that the droplet exhibits exactly three possible topological phases: 
\begin{itemize}
    \item (Regime I) The droplet is doubly connected;
    \smallskip 
    \item (Regime II) The droplet is simply connected; 
    \smallskip 
    \item (Regime III) The droplet consists of two simply connected components.
\end{itemize}
Moreover, the phase diagram in the parameter space $(a,c,\tau)$, characterising the conditions under which each regime occurs, is determined explicitly. For completeness, we summarise these conditions in Appendix~\ref{Appendix_exact description}. See Figure~\ref{fig_phases} for an illustration of Regimes~I--III in the $(a,\tau)$-plane for a fixed value of $c$.

\begin{figure}[t]
    \centering

\begin{tikzpicture}[
    tick style/.style={-, semithick},
    arrow style/.style={->, very thick, red}
]
    \node (image) at (0,0) {\includegraphics[width=0.55\linewidth]{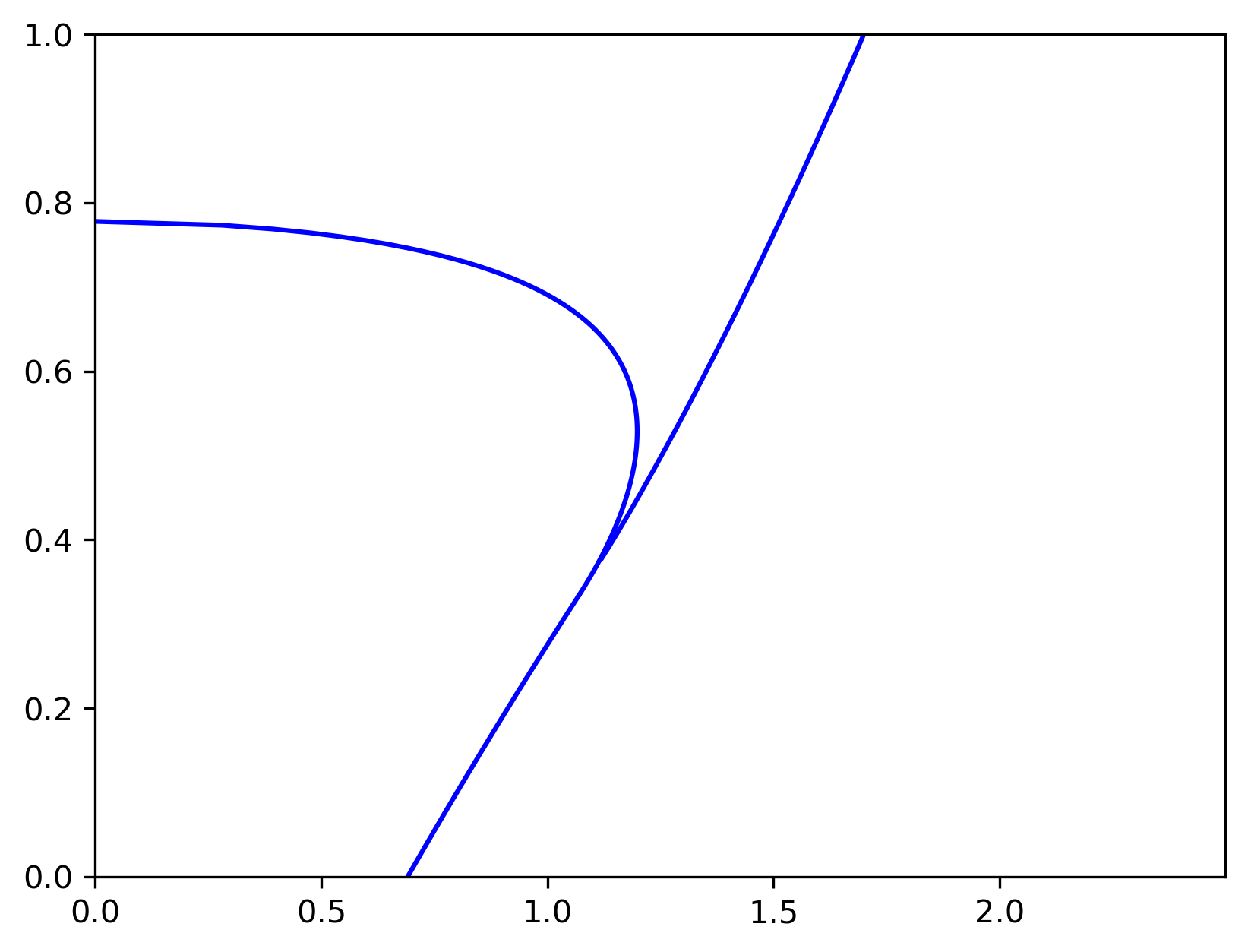}};

    \node at (4.48,-3.25) {$\infty$};

    \node at (-4.8, 0) [rotate = 90] {$\tau$};
    \node at (0.3, -3.7) {$a$};

    \node at (-3.3, 1.3) {\Large \textbf{I}};
    \node at (3.8, 2.8) {\Large \textbf{II}};
    \node at (-3.3, 2.8) {\Large \textbf{III} }; 

    \node at (-0.6, -0.19) {\large (A)};
    \node at (-3.6, -0.15) {\large (B)};
    \node at (-3.55, -2.7) {\large (C)};
    \node at (0.6, -0.19) {\large (D)};
    \node at (4.2, -0.15) {\large (E)};

    \draw [arrow style] (-0.6, -0.55) -- ( -3.93,-0.55);
    \draw [arrow style] (-3.95, -0.55) -- ( -3.95, -2.94);
    \draw [-, very thick, red] (0.6, -0.55) -- ( 3.5, -0.55);
    \draw [dotted, ->, very thick, red] (3.5, -0.55) -- (4.5, -0.55);

\end{tikzpicture}

    \caption{The plot depicts the phase diagram of three regimes---Regime~I (doubly connected), Regime~II (simply connected), and Regime~III (multi-component)---in the \((a,\tau)\)-plane for \(c = 1/7\). The arrows indicate the deformation paths used to compute the variation of the free energy from the generic points (A) and (D): in the doubly connected regime, (A)\(\to\)(B):=$(0,\tau)$ corresponds to variation in \(a\) and (B)\(\to\)(C):=$(0,0)$ to variation in \(\tau\), with (C) as the reference configuration for \(Q^{\mathrm{i}}\); in the simply connected regime, (D)$\to $ (E):=$(\infty,\tau)$ corresponds to variation in \(a\), yielding the reference configuration for \(Q^{\mathrm{e}}\). See also Figure~\ref{fig_deformation} for these deformations. }
    \label{fig_phases}
\end{figure}
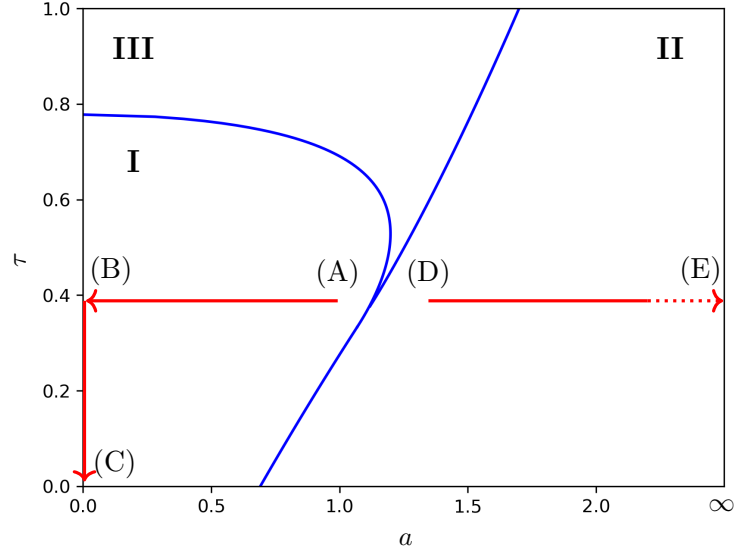

In addition, in Regimes I and II, the explicit shape of the droplet is determined, which we now recall.

\begin{thm}[\textbf{Description of the droplet; cf. Theorem 1.2 in \cite{BY25}}] \label{Thm_droplet} Let $Q$ be given as~\eqref{potQ}. Then we have the following. 
\begin{itemize}
\item[\textup{(i)}] 
In \textup{Regime I}, the droplet is given by
\begin{equation} \label{def of KQ doubly}
 K_Q = \mathsf{E} \setminus \mathsf{D}, 
\end{equation}
where
    \begin{equation} \label{droplet_doubly connected}
        \mathsf{E}:= \bigg\{ (x,y)\in \R^2: \Big(\frac{x}{1+\tau}\Big)^2+\Big(\frac{y}{1-\tau}\Big)^2\le 1+c\bigg\}, \quad \mathsf{D} :=\bigg\{ (x,y)\in \R^2: (x-a)^2 +y^2 < c(1-\tau^2)\bigg\}.
    \end{equation}
\item[\textup{(ii)}] In \textup{Regime II}, the droplet is given by
\begin{equation} \label{def of KQ simply}
    K_Q = \overline{\mathrm{Int}\,\psi(\T)},
\end{equation}
where $\psi:\D_e \to \C$ is the rational map defined by
    \begin{align}\label{eq_simply conformal map}
    \psi(w) = R\Big(w + \frac{\tau}{w} - \frac{\lambda}{w-q}-\frac{\lambda}{q(1-\tau)}\Big). 
\end{align}
Here, $R>0$, $q \in (0,1)$, and $\lambda \in [0, \lambda_{\rm cri})$ (see Appendix~\ref{Appendix_exact description} for the definition of $\lambda_{\rm cri}$), and $(R,q,\lambda)$ is determined as the solution to a system of coupled algebraic equations \eqref{def of coupled equations}. 
\end{itemize}
\end{thm}

Recall that the partition function is defined by 
\begin{equation} \label{def of ZN N-fold intgral}
\ZZ_N(Q):= \int_{ \C^N } \prod_{1\le j<k\le N}|z_j-z_k|^2 \prod_{j=1}^N e^{-NQ(z_j)}\ud A(z_j). 
\end{equation}
Now we are ready to present our main result.

\begin{thm}[\textbf{Free energy expansion}] \label{thm_ZW}
Let $Q$ be given as~\eqref{potQ}. Suppose that the associated droplet $K_Q$ is either simply or doubly connected, i.e. $(a,c,\tau)$ falls within \textup{Regime I} or \textup{Regime II}. Then as $N \to \infty$, we have 
\begin{align}
\begin{split} \label{free energy expasion in main Thm}
\log \ZZ_N(Q) &= -I_Q[\mu_Q]N^2+ \frac12 N \log N + \bigg( \frac{ \log(2\pi) }{2} -1-\frac12 \int_\C \log(\Delta Q) \ud \mu_Q \bigg)N 
\\
&\quad +\frac{6-\chi}{12} \log N + \chi \, \zeta'(-1)+\frac{\log(2\pi)}{2} + \frac{1}{24} \mathcal{L}[K_Q]  -\frac{1}{24\pi}\int_{\partial K_Q} \log (\Delta Q) \kappa \ud s + O( N^{-1} ), 
\end{split}
\end{align}
where $\chi$ is the Euler characteristic of $K_Q$, $\mathcal{L}[K_Q]$ is the Liouville action~\eqref{def of Liouville action domain}, and $\kappa$ is the curvature of $\partial K_Q$.  
\end{thm}

\begin{rem}[Comparison with the $\tau=0$ case in \cite{BSY25}] Note that Theorem~\ref{thm_ZW} recovers~\cite[Theorem 2.2]{BSY25} in the case $\tau=0$.
In contrast to \cite{BSY25}, where the results are presented in an explicit, formula-driven manner, our approach is formulated in terms of conformal-geometric functionals. In particular, it does not rely heavily on explicit expressions for the droplet or related quantities, marking a conceptual shift and rendering the method amenable to extensions, for instance, to algebraic Hele-Shaw potentials. 
\end{rem}

\begin{rem}[Explicit formulas of the coefficients]\label{rem_explicit Liouville}
In \eqref{free energy expasion in main Thm}, all coefficients can be made explicit in terms of the parameters $a,c,\tau$ defining the potential \eqref{potQ}. 
More precisely, the associated energy $I_Q[\mu_Q]$ was computed in \cite[Theorem 1.2]{BY25}; see also Proposition~\ref{Prop_energy explicit} below. Moreover, for the $O(N)$ term, since $\mu_Q$ is a probability measure and $\Delta Q = 1/(1-\tau^2)$, we have
\begin{equation}\label{eq_explicit entropy}
    -\frac{1}{2} \int_{\C} \log(\Delta Q)\, \ud \mu_Q 
    = \frac{1}{2} \log(1-\tau^2).
\end{equation} 
Note also that by definition, 
\begin{equation}
\chi = \begin{cases}
0  &\textup{for Regime I}, 
\smallskip \\
1 & \textup{for Regime II}. 
\end{cases}
\end{equation} 
Furthermore, it is straightforward to compute the Liouville action; cf. \eqref{eq_rational action}. In the doubly connected regime,
\begin{equation} \label{Liouville action explicit doubly}
  \frac{1}{24} \mathcal{L}[K_Q] = \frac{1}{12}\log\Big(\frac{c}{1+c}\Big), 
\end{equation}  
while in the simply connected regime,
    \begin{align}
    \begin{split} \label{Liouville action explicit simply}
    \frac{1}{24} \mathcal{L}[K_Q] &=-\frac{1}{6}\log \Big(R(1-q^2)\Big)-\frac{1}{24}\log (\lambda+ (1-\tau)(1-q)^2) 
    - \frac{1}{24}\log(\lambda+(1-\tau)(1+q)^2)\\
    &\quad +\frac{1}{6}\log\Big(q^2\lambda+ (1-q^2)^2(1-\tau q^2)\Big) -\frac{1}{12}\log \Big((1+\tau q^2)^2\lambda -(1+\tau)(1-q^2)(1-\tau q^2)^2\Big).
        \end{split}
    \end{align} 
Here, the parameters are defined in Theorem~\ref{Thm_droplet} (ii). 
On the other hand, by the Gauss--Bonnet theorem, the last term simplifies to
\begin{equation} \label{evaluation of boundary terms kappa}
    -\frac{1}{24\pi}\int_{\partial K_Q} \log (\Delta Q) \kappa \ud s = -\frac{\chi}{12} \log (\Delta Q) = \begin{cases}
        0 &\text{for Regime I,}\\
        \ds \frac{1}{12}\log (1-\tau^2) &\text{for Regime II.}
    \end{cases}
\end{equation} 
\end{rem}

\begin{rem}[Full order expansion] \label{Rem_full order expansion} The error term $O(N^{-1})$ in \eqref{free energy expasion in main Thm} is indeed optimal. Furthermore, in the doubly connected case, the simple form of the droplet \eqref{def of KQ doubly} allows one to derive a full asymptotic expansion. 
More precisely, in this case, the error term admits the expansion: for any positive integer $k_0$,
\begin{equation} \label{doubly error}
\sum_{k=1}^{k_0} N^{-k}\EE_k+O(N^{-k_0-1}), \qquad     \EE_k = 
        \begin{cases}
            \ds \frac{B_{k+2}}{k(k+1)}\Big(\frac{1}{(1+c)^k}-\frac{1}{c^k}\Big), & k:\text{even},
            \smallskip 
            \\
            \ds \frac{B_{k+1}}{k(k+1)}, &k:\text{odd}.
        \end{cases}
\end{equation} 
Here, $B_k$ denotes the $k$-th Bernoulli number defined by 
\begin{equation} \label{def of Bernoulli number}
\frac{t}{e^t-1} = \sum_{k=0}^\infty B_k \frac{t^k}{k!}.
\end{equation}
\end{rem}

\begin{rem}[Moments of characteristic polynomials]
While we present Theorem~\ref{thm_ZW} in the framework of the free energy expansion of two-dimensional Coulomb gases, it admits a natural interpretation in non-Hermitian random matrix theory in terms of moments of characteristic polynomials. This perspective is of independent interest and has attracted considerable attention due to its various connections, for instance, with Gaussian multiplicative chaos \cite{Lam20}. 
In this context, one distinguishes between regimes where the powers are fixed and those where they scale proportionally with $N$. Recent works include \cite{WW19,By25a,DS22} for the Ginibre ensemble, \cite{BCMS25, BFL25, DMMS25} for the truncated unitary and spherical ensembles, and \cite{BDHK25} for general determinantal two-dimensional Coulomb gas  models with finite moments.

To be more precise, let $X$ be the elliptic Ginibre ensemble with non-Hermiticity parameter $\tau \in [0,1)$; see \cite[Section 2.3]{BF25}. As mentioned earlier, the eigenvalues of $X$ follow the distribution \eqref{eq_RNM} with potential $Q^{\rm e}$ defined in \eqref{def of Q induced elliptic}.  By definition, the moment of the characteristic polynomial of $X$ is related to the partition functions via
\begin{equation} \label{moment in terms of Z ratio}
\mathbb{E} \Big[ \,  \Big|\det (X-a) \Big|^{2cN} \Big]  = \frac{ \ZZ_N(Q) }{ \ZZ_N(Q^{ \rm e }) }. 
\end{equation}
Combining the above, we see that Theorem~\ref{thm_ZW} yields, as an immediate corollary, the large-$N$ expansion of the moments of the characteristic polynomial of the elliptic Ginibre ensemble; see also~\cite[Corollary 1.3]{BY25} for such formulation deriving the leading order term.

As noted above, while it is natural in the free energy expansion framework for the exponent in the left-hand side of \eqref{moment in terms of Z ratio} to be of order $O(N)$, the case of a fixed exponent is also of interest in probabilistic applications. In our setting, this corresponds to replacing $c$ by $c/N$ in the potential \eqref{potQ}, introducing an $N$-dependence. 
Such an expansion was recently obtained in \cite{BDHK25} for general genuinely two-dimensional potentials. The methods developed here, in particular Proposition~\ref{prop_key}, are also applicable in this regime, and we plan to investigate this problem further, especially in the weak non-Hermiticity regime $\tau \to 1$ under an appropriate scaling that interpolates between two- and one-dimensional behaviours.
\end{rem}

\begin{rem}[Multi-component regime] In Theorem~\ref{thm_ZW}, we focus on the case where the droplet is either simply or doubly connected. Nevertheless, as previously mentioned, there is a third regime (Regime~III) in which the droplet consists of two disjoint simply connected components. 

In Hermitian random matrix theory, such a situation where the limiting spectral measure is supported on disconnected sets (the so-called \textit{multi-cut regime}) and the associated free energy expansion have been extensively studied; see e.g. \cite{CFWW25,CGM15,BG24} and references therein. In that setting, the constant term in the free energy expansion exhibits an oscillatory contribution described in terms of the Riemann theta function, reflecting the underlying quasi-periodic behaviour of the particle system.

By contrast, in the two-dimensional Coulomb gas setting, an analogous oscillatory structure in the free energy expansion remains largely open, particularly beyond radially symmetric situations \cite{ACC26}. In fact, even the precise form of the conjectural expansion has not yet been established; we refer to \cite{By25a} for further discussion.
In particular, Regime~III falls outside the scope of the present work. For instance, the main approach of this paper—based on the general theory of planar orthogonal polynomials developed in \cite{HW21,Hed24}—does not apply in this setting. Instead, analysing this regime appears to require the development of a suitable Riemann--Hilbert analysis, which we leave for future work. 

As a related side remark, we note that the inapplicability of the results in \cite{HW21,Hed24} to multi-component ensembles is the main reason why the edge universality for determinantal Coulomb gases established in \cite{HW21} has thus far been confined to connected droplet geometries. Nevertheless, this universality has recently been extended to the multi-component regime in \cite{CW25} by means of an alternative approach that does not rely on orthogonal polynomials. 
\end{rem}

\begin{rem}[Duality relation for the elliptic Ginibre matrix] \label{Rem_duality} In the previous section, we noted that one of the key tools in recent studies of free energy expansions is the use of duality relations. However, the type of duality that arises from supersymmetric methods, particularly those based on Grassmann integration, appears to be difficult to exploit in the present model. 

Let us make this more precise. Recall that \(X\) denotes an elliptic Ginibre matrix of size \(N\) with non-Hermiticity parameter \(\tau \in [0,1]\), and suppose that \(cN\) is an integer. Then the associated duality relation is given by (see \cite[Remark 5.1]{Fo25}) 
\begin{align} \label{duality for eGinUE}
\mathbb{E} \Big[ \,  \Big|\det (X-a) \Big|^{2cN} \Big]  \propto \bigg \langle \det \begin{pmatrix}
- i a I + \sqrt{2\tau} A & C
\smallskip 
\\
C^* & -ia I+\sqrt{2\tau}B 
\end{pmatrix}^N \bigg \rangle_{ \substack{  A,B \in {\rm GUE}_{cN} \\ C \in {\rm GinUE}_{cN} }  }. 
\end{align}
Here \(A\) and \(B\) are independent GUE matrices of size \(cN\), and \(C\) is an independent Ginibre matrix of the same size; the expectation on the right-hand side is taken over all three ensembles. 

Observe that when \(\tau = 0\) (corresponding to the setting studied in \cite{BSY25}), the right-hand side of \eqref{duality for eGinUE} simplifies considerably, as the dependence on \(A\) and \(B\) disappears. Moreover, since \(C C^*\) is distributed according to the Laguerre unitary ensemble, the right-hand side can be reformulated in terms of a probability involving the smallest eigenvalue. This observation played a crucial role in \cite{BSY25}, as well as in the subsequent works \cite{BFL25,BCMS25}, where the left-hand side of \eqref{duality for eGinUE} is replaced by truncated unitary or spherical ensembles, which can likewise be related to extremal eigenvalue statistics of the Jacobi unitary ensemble.

In contrast to these earlier works, the right-hand side of \eqref{duality for eGinUE} in the present setting involves three random matrix ensembles, for which a systematic asymptotic theory is currently unavailable. This constitutes the main reason why deformation techniques provide the most effective approach in our analysis. On the other hand, our results yield the asymptotic behaviour of the right-hand side of \eqref{duality for eGinUE}, which may be of independent interest.
\end{rem}

\begin{rem}[Critical asymptotics and Tracy-Widom distribution] As is evident from the phase diagram in Figure~\ref{fig_phases}, there exists a critical line separating Regimes~I and II. Geometrically, this corresponds to the situation in \eqref{droplet_doubly connected} where the inner disc meets the outer ellipse tangentially, thereby forming a double point; see \cite[Figure~4 (A)]{BY25}. In this critical regime, it is expected that the coefficient of the $\log N$ term coincides with that of the doubly connected droplet.  
Moreover, in analogy with the case $\tau = 0$, one expects the Tracy--Widom distribution to appear in the constant term of the free energy expansion. While such a result for $\tau = 0$ was obtained via a duality argument, as mentioned above, this approach is not applicable in the present setting. We also note that a further indication of the emergence of the Tracy--Widom distribution is given by the universality of the Painlev\'e~II asymptotics for the associated orthogonal polynomial (established in \cite{BBLM15,KLY25} for $\tau = 0$).

The derivation of the corresponding critical asymptotics appears to require a different approach, or at least the use of the Riemann--Hilbert analysis framework. Nevertheless, similarly to \cite[Remark~2.6]{BSY25}, one can formally observe the emergence of an additional $\frac{1}{12}\log N$ term arising from the topological phase transition, which is related to the tail asymptotics of the Tracy--Widom distribution. 
More precisely, in this regime we consider the scaling
\begin{equation}
a= a_{ \rm cri }[{\rm I,II}] +O(N^{-2/3}), 
\end{equation}
where $a_{\rm cri}[{\rm I,II}]$ denotes the critical value separating Regimes~I and~II. This is precisely the scale at which one expects the Painlev\'e~II asymptotics to arise. With this scaling, it follows directly from the explicit formula \eqref{Liouville action explicit simply} that
\begin{equation}
\mathcal{L}[K_Q] = \frac{1}{12}\log N +O(1). 
\end{equation} 
This behaviour captures the critical interpolation in the free energy expansion in \eqref{free energy expasion in main Thm}: as the Euler characteristic transitions from $\chi = 1$ to $\chi = 0$, an additional $\frac{1}{12}\log N$ term emerges. 

We further note that, for the model under consideration, there exists an even more singular critical regime, namely the case where all three regimes~I, II, and~III meet at a single point; see \cite[Remark~1.5]{BY25} for further discussion of this regime and its geometric interpretation. The analysis of such a highly critical case, referred to as the \textit{triple point}, remains widely open. 
One can nevertheless expect that, in this regime, the Painlev\'e II hierarchy arises. In particular, in the free energy expansion, the constant term is expected to involve not the standard Tracy--Widom distribution, but rather its higher-order analogue, namely the Fredholm determinant associated with a higher-order Airy kernel \cite{CCG21,LMS18}. The analysis of this critical regime lies well beyond the scope of the present work, and we leave it for future investigation.
\end{rem}

\subsection{Strategy of the proof} \label{Subsection_main strategy}

In this subsection, we outline the overall strategy of the proof. The main idea is to relate the free energy to the associated planar orthogonal polynomials defined in \eqref{def of planar OP}. For the reader's convenience, we summarise the argument in the following steps.

\begin{itemize}
    \item We relate the variations of the free energy to the first few leading coefficients of the planar orthogonal polynomials; see Proposition~\ref{prop_key}. In doing so, we consider deformations with respect to both the insertion point $a$ and the non-Hermiticity parameter $\tau$.
    
    \smallskip
    
    \item We then derive the asymptotic behaviour of the coefficients of the planar orthogonal polynomials, in particular the first two subleading coefficients (Propositions~\ref{Prop_OP coefficients doubly} and ~\ref{prop_sub_simply}). For this, we make use of the theory and algorithm developed in \cite{HW21,Hed24}. Implementing this framework in a fully explicit manner in the present setting requires substantial work, including a number of conformal geometric computations. (From a methodological viewpoint, this marks a key distinction from \cite{BSY25}, where the coefficients were derived via Riemann--Hilbert analysis combined with a partial Schlesinger transform.)
    
    \smallskip
    
    \item Finally, to integrate the variations, we need to identify appropriate reference partition functions. In particular, in the doubly connected regime, it is essential to combine the deformations with respect to both $a$ and $\tau$. The reference partition functions are then obtained from exact solvability (Lemma~\ref{Lem_reference free energy}). After this, we interpret all resulting terms within the framework of the general conjecture, leading to the form \eqref{free energy expasion in main Thm}.
\end{itemize}

See Figure~\ref{Fig_diagram strategy} for a schematic summary of the overall strategy.

\begin{figure}[ht]
\centering
\begin{tikzpicture}[
  box/.style={
    draw,
    rectangle,
    minimum height=1.4cm,
    text width=4.6cm,
    align=center,
  inner sep=3pt
  },
  boxL/.style={
    box,
    fill=blue!8,
    draw=blue!40
  },
  boxR/.style={
    box,
    fill=green!8,
    draw=green!40!black
  },
  boxC/.style={
    box,
    fill=orange!10,
    draw=orange!60!black
  },
  boxT/.style={
    box,
    fill=red!10,
    draw=red!50!black
  },
  arr/.style={-Latex, thick},
  node distance=1.2cm and 2.5cm,
  every node/.style={font=\small}
]

\node[boxL] (refpf) {Asymptotics of the\\ reference free energy\\
(Lemma~\ref{Lem_reference free energy})};

\node[boxR, below=of refpf] (variation) {Asymptotics of the\\ variation of the free energy\\
(Propositions~\ref{Prop_OP coefficients doubly} and~\ref{prop_sub_simply})};

\path (refpf) -- (variation) coordinate[midway] (midpoint);

\node[boxC, right=3.5cm of midpoint] (deformation)
{Deformation of the free energy\\
(Proposition~\ref{prop_key})};

\node[boxT, right=1.25cm of deformation] (expansion)
{Free energy expansion\\
(Theorem~\ref{thm_ZW})};

\draw[arr] (refpf.east) to[bend left=15] (deformation.north west);
\draw[arr] (variation.east) to[bend right=15] (deformation.south west);
\draw[arr] (deformation.east) -- (expansion.west);

\end{tikzpicture}
\caption{Schematic diagram of the proof strategy for Theorem~\ref{thm_ZW}. }
\label{Fig_diagram strategy}
\end{figure}
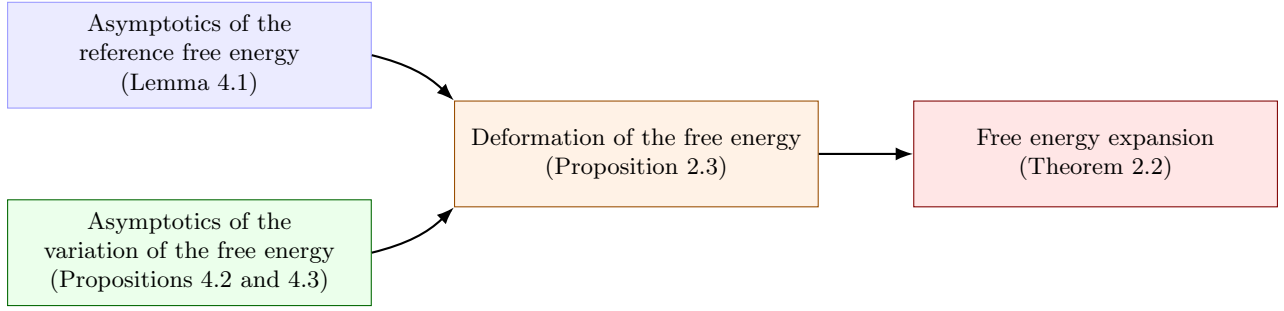

We proceed to describe this strategy in greater detail.
It follows from the well-known representation \eqref{ZN in terms of prod norms} that the free energy can be expressed in terms of the corresponding orthogonal norms. However, the right-hand side of \eqref{ZN in terms of prod norms} involves contributions from polynomials of all degrees ranging from $j=0$ to $j=N-1$. This presents a difficulty, since even though there exist well-developed methods for analysing the asymptotic behaviour of orthogonal polynomials—such as the Riemann--Hilbert approach or the foliation flow framework—these techniques typically provide precise information only for polynomials of large degree, rather than for those of low degree.

On the other hand, the key idea used in \cite{BSY25} is that, by suitably deforming the free energy, its variation can be expressed solely in terms of finitely many orthogonal polynomials of highest degree. (This viewpoint is reminiscent of the philosophy in the classical Christoffel--Darboux formula.) We adopt a similar strategy in the present work; however, in our setting, this idea must be developed at a higher level of refinement and precision.

To be more precise, we relate the variations of the free energy with respect to the parameters to the leading-order coefficients of $P_{n,N}(z)$. To this end, we denote by $\AA_{n,N}$ and $\BB_{n,N}$ the coefficients of $z^{n-1}$ and $z^{n-2}$, respectively, in $P_{n,N}(z)$, i.e.  
\begin{equation}\label{defAB}
    P_{n,N}(z)= z^n + \AA_{n,N}z^{n-1} + \BB_{n,N}z^{n-2} + O(z^{n-3}), \quad z\to \infty.
\end{equation}
The variations of the free energy are then expressed in terms of these coefficients as follows. 

\begin{prop}[\textbf{Deformations of the free energy}]\label{prop_key}
 Let $Q$ be given by \eqref{potQ}. Let $P_{n,N}(z)$ be the monic planar orthogonal polynomial of degree $n$ with respect to the weight $e^{-NQ}\,\ud A$, and let $\AA_{n,N}$ and $\BB_{n,N}$ be defined as in \eqref{defAB}. Then the free energy $\ZZ_N(Q) \equiv \ZZ_N(a,c,\tau)$ satisfies the following. 
\begin{enumerate}
    \item[\textup{(i)}] \textup{(\textbf{Deformation with respect to }$a$)}
    \begin{equation}\label{eq_deform a}
        \frac{\ud }{\ud a}\log \ZZ_N(Q) = \frac{2N}{1+\tau}  \AA_{N,N}. 
    \end{equation}
    \item[\textup{(ii)}] \textup{(\textbf{Deformation with respect to }$\tau$)}
    \begin{equation}\label{eq_deform tau}
         \frac{\ud }{\ud \tau}\log \ZZ_N(Q) =-\frac{N}{1-\tau^2} \bigg[\tau (1+2c)N +\tau  +(\BB_{N,N}+\BB_{N+1,N}- \AA_{N,N}\AA_{N+1,N})-\frac{2\tau a}{1+\tau} \AA_{N,N}\bigg].
    \end{equation}
\end{enumerate}
\end{prop}

See Remark~\ref{Rem_defomration for c=0 case} for the exactly solvable case \(c = 0\), in which Proposition~\ref{prop_key}~(ii) can be verified by a direct computation using properties of the Hermite polynomials. 

We note that, in the case $\tau=0$, the deformation with respect to the insertion point $a$ in \eqref{potQ} was employed in \cite{BSY25} (see also \cite{DMMS25}). More precisely, Proposition~\ref{prop_key} (i) extends \cite[Proposition 3.9]{BSY25}, which serves as a key ingredient in the proof of the main results therein. 
On the other hand, the second type of deformation with respect to the non-Hermiticity parameter $\tau$, as formulated in Proposition~\ref{prop_key}~(ii), appears to be new in the literature, to the best of our knowledge. These two types of deformations are complementary, and it is precisely their combination that allows us to derive our main results; see Figures~\ref{fig_phases} and ~\ref{fig_deformation}.  

The proof of Proposition~\ref{prop_key} is given in Section~\ref{Section_proof main}. In the case $\tau = 0$, Proposition~\ref{prop_key}~(i) was based on the isomonodromic tau-function associated with a determinantal Coulomb gas model, with the underlying idea tracing back to \cite{Be09}. 
In contrast, in the present paper we provide a simple and elementary proof of Proposition~\ref{prop_key}, which does not rely on any of these frameworks. Our approach is based on the observation that variations of the free energy can be interpreted as particular one-point observables of the Coulomb gas model, and then the remaining proof uses only the orthogonality of the polynomials.  
Notably, the proof of \cite[Proposition 3.9]{BSY25} is rather technical and depends on substantial preliminary work, including the construction of the associated Riemann--Hilbert problem in \cite{BBLM15}. In this sense, the elementary proof presented here may be regarded as one of the main contributions of this paper. In addition, our proof extends to a more general setting, including algebraic Hele-Shaw potentials. 

We also note that, in the literature, deformation methods based on varying the strength of the point charge $c \ge 0$ in \eqref{potQ} have been used as well; see \cite{WW19,HW24a}. While this approach is natural, for our purposes it is more convenient to work with deformations with respect to $a$ and $\tau$, which better align with its simple reduction to coefficients of orthogonal polynomials; see Remark~\ref{rem_deform c}. 

\medskip 

With Proposition~\ref{prop_key} at hand, we proceed as follows. In applying these deformations, it is necessary to distinguish between the simply and doubly connected regimes.

\begin{itemize}
    \item In the simply connected case, we use only the deformation with respect to $a$. In particular, we exploit the fact that as $a \to \infty$, the potential \eqref{potQ}, after a suitable renormalisation, reduces to $Q^{\rm e}$. The latter corresponds to an exactly solvable case, for which the free energy admits an explicit expression; see Lemma~\ref{Lem_reference free energy}. Therefore, it suffices to obtain the asymptotics of $\mathcal{A}_{n,N}$, which is provided in Proposition~\ref{prop_sub_simply}.
    
    \smallskip
    
    \item In the doubly connected regime, however, the deformation with respect to $a$ alone is not sufficient, since the case $a=0$ for the potential \eqref{potQ} does not admit a representation in terms of a well-known special function. We therefore additionally employ the deformation with respect to $\tau$ to reach the point $a=\tau=0$, where the potential reduces to $Q^{\rm i}$ in \eqref{def of Q induced elliptic}, which is exactly solvable; see Lemma~\ref{Lem_reference free energy}. In this case, we require the asymptotics of both $\mathcal{A}_{n,N}$ and $\mathcal{B}_{n,N}$, obtained in Proposition~\ref{Prop_OP coefficients doubly}. 
    
   In this two-step deformation, it is essential to first deform with respect to $a$ and then with respect to $\tau$. Indeed, reversing the order may cause the deformation path to pass through Regime~II; see Figure~\ref{fig_phases}. In that case, a transition in the topology of the equilibrium measure may occur, necessitating a careful analysis of critical asymptotics. By contrast, in the order of deformation adopted here, no such issue arises. 
\end{itemize}

\begin{figure}[t]
    \centering
\begin{tikzpicture}[ 
    domain fill/.style={pattern=triangular dots, pattern color=blue!40!gray},
    shape style/.style={draw, very thick, inner sep=0pt},
    ellipse style/.style={shape style, domain fill, ellipse, minimum width=3cm, minimum height=1cm},
    ellipse2 style/.style={shape style, domain fill, ellipse, minimum width = 2.7cm, minimum height = 0.9cm},
    circle style/.style={shape style, domain fill, circle, minimum size=1.73cm},
    circle2 style/.style={shape style, domain fill, circle, minimum size= 1.558cm},
    nucleus style/.style={shape style, circle, minimum size=0.4cm, fill=white}, 
    label node/.style={anchor=south, node distance=0.8cm, font=\large},
    arrow style/.style={->, thick, red},
    arrow blue/.style={->, thick, blue},
    arrow2 style/.style={->, thick},
    scale = 0.9
]
    \node (A) [ellipse style] at (0,0) {};
    \node at (0, 0.6) [label node] {(A)};
    \node at (A.center) [nucleus style, xshift =0.7cm] {}; 

    \node (B) [ellipse style] at (-4.0,0) {};
    \node at (-4.0, 0.6) [label node] {(B)};
    \node at (B.center) [nucleus style] {}; 

    \node (C) [circle style] at (-4.0, -3.5) {};
    \node at (C.north) [label node] {(C)};
    \node at (C.center) [nucleus style] {};

    \node at (4.0, 0.6) [label node] {(D)};
    \begin{scope}[shift={(4.0,0)}]
      \draw [shape style, domain fill] plot [variable=\t, domain=0:360, samples=300, smooth] 
      (
          {(1.5*cos(\t) - 0.1*(cos(\t) - 0.8)/(1.64 - 1.6*cos(\t)))}, 
          {(0.5*sin(\t) + 0.1*sin(\t)/(1.64 - 1.6*cos(\t)))}         
      ) -- cycle;
    \end{scope}

    \node (E) [ellipse2 style] at (8.0, 0){};
    \node at (8.0, 0.6) [label node] {(E)};

    \node [circle style] at (0.0, -3.5) {};
    \node [nucleus style] at (0.6, -3.5) {};

    \node [circle2 style] at (8.0, -3.5) {};

    \begin{scope}[shift={(4.0,-3.5)}]
      \draw [shape style, domain fill] plot [variable=\t, domain=0:360, samples=300, smooth] 
      (
          {0.89*(cos(\t) - 0.1*(cos(\t) - 0.8)/(1.64 - 1.6*cos(\t)))}, 
          {0.89*(sin(\t) + 0.1*sin(\t)/(1.64 - 1.6*cos(\t)))}         
      ) -- cycle;
    \end{scope}

    \node at (-4.0, -5.6) [label node] {$a=0$};
    \node at (8.0, -5.6) [label node] {$a=\infty$};
    \node at (2.0, -5.6) [label node] {$0<a<\infty$};
    \node at (-6.7, -3.8) [label node] {$\tau =0$};
     \node at (-6.7, -0.35) [label node] {$\tau >0$}; 
     
    \draw [arrow style] (-4.0, -1.1) -- (-4.0, -1.5);
    \draw [arrow style] (-1.8, 0.0) -- (-2.2, 0.0);
    \draw [arrow style] (5.8, 0.0) -- (6.2, 0.0);
    \draw [arrow2 style] (-6.3, -5.0) -- (10.0, -5.0);
    \draw [arrow2 style] (-6.0, -5.3) -- ( -6.0, 1.2);
    \draw [arrow blue] (-1.8, -3.5) -- (-2.2, -3.5);
    \draw [arrow blue] (5.8, -3.5) -- (6.2, -3.5);
\end{tikzpicture}
    \caption{Schematic illustration of the deformation procedures. The labels (A)--(E) are consistent with those in Figure~\ref{fig_phases}. 
In the simply connected case, the deformation from the generic configuration (D) to the reference configuration (E) involves only a deformation in $a$. In contrast, in the doubly connected case, the deformation from the generic configuration (A) to the reference configuration (C) requires successive deformations with respect to $a$ and then $\tau$. 
The second row illustrates the deformation path for $\tau = 0$ employed in \cite{BSY25}; see also Figure~5 therein.}
    \label{fig_deformation}
\end{figure}
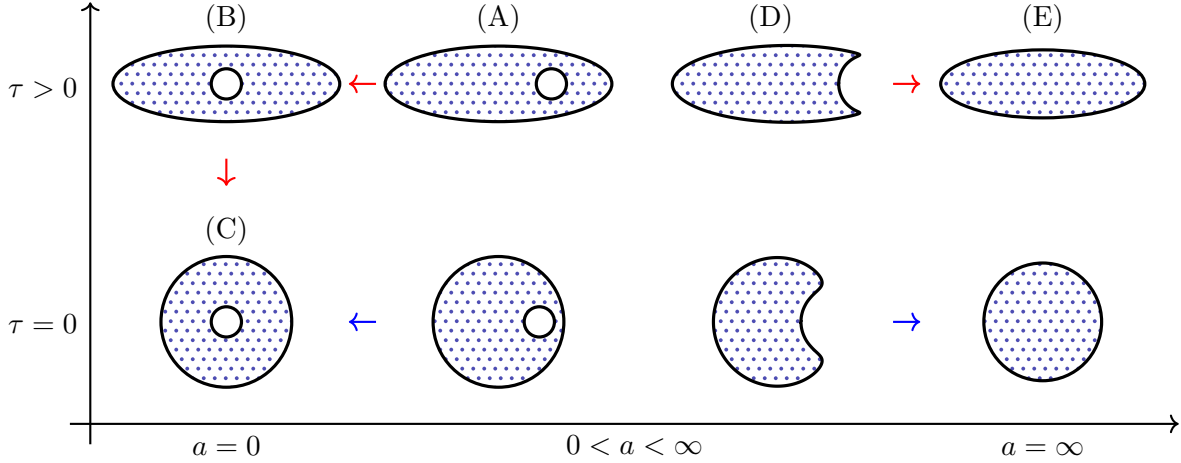
See Figure~\ref{fig_phases} for an illustration of the deformation paths in each case within the phase diagram. See also Figure~\ref{fig_deformation}, where the deformation strategy is depicted, with particular emphasis on the topology of the droplet. 

Combining all of these ideas yields the outline of the proof of our main result, Theorem~\ref{thm_ZW}; see Section~\ref{Section_proof main}.

\begin{rem}[Generality of the method] In contrast to the earlier work~\cite{BSY25} on the model~\eqref{def of Q GinUE with insertion}, whose method is rather specifically tailored to that particular setting, the approach developed in the present paper has the possibility of being extended to a much broader class of potentials~\(Q\).

Indeed, although our main results are stated for the particular potential~\eqref{potQ}, many of the intermediate steps are developed in the more general framework of algebraic Hele-Shaw potentials of the form~\eqref{def of algebraic HeleShaw}. These include the refined asymptotic analysis of planar orthogonal polynomials in Theorem~\ref{thm_F0F1} and Proposition~\ref{prop_multiply connected}, which we regard as contributions of independent interest in the study of more general potentials.

Nevertheless, extending the free energy expansion to a genuinely broad class of potentials requires several additional ingredients that are not yet fully understood. For instance, one needs an appropriate deformation path preserving the topology of the droplet, as well as finer asymptotic expansions beyond those established in Theorem~\ref{thm_F0F1}. Once a suitable deformation framework is available, the strategy developed in the present paper strongly suggests that the first subleading correction term should already determine the $O(1)$ contribution in the free energy expansion. On the other hand, finer asymptotic expansions still appear to be necessary in order to control the corresponding error terms rigorously. 
This is one of the main reasons why, in the present work, we focus on the special potential~\eqref{potQ}, for which the earlier work~\cite{BY25} provides the necessary structural inputs.

At the same time, one of the principal messages of the present paper is that the general strategy developed here---namely, relating the asymptotics of planar orthogonal polynomials to conformally invariant geometric functionals---indeed works in a concrete nontrivial model \eqref{potQ}. We regard the present work as a first step toward developing a more systematic framework for understanding free energy expansions in considerably broader settings.
\end{rem}

\subsection*{Organisation of the paper}

The remainder of this paper is organised as follows. In Section~\ref{section_prelim}, we review several well-known results, including those from \cite{HW21,Hed24}, and present preliminary material on conformal geometric computations. In Section~\ref{Section_proof main}, we give the proof of the main results, following the overall strategy outlined above.  
The key ingredients---Propositions~\ref{Prop_OP coefficients doubly} and~\ref{prop_sub_simply}, which concern the asymptotic behaviour of the variation of the free energy, or equivalently, the leading coefficients of the planar orthogonal polynomials---are established in Sections~\ref{sec simply} and ~\ref{sec doubly}.

\subsection*{Acknowledgements} SSB was supported by the National Research Foundation of Korea grants (RS-2025-00516909).

\section{Preliminaries: conformal geometry, potential theory and orthogonal polynomials} \label{section_prelim}

In this section, we introduce several preliminary results that will be used later. While we mainly review well-known facts, we also provide self-contained proofs of certain folklore formulas adapted to our purposes. These include the variational formulas for the Liouville action (Proposition~\ref{prop_var Liouville}) and the harmonic extension of Hele-Shaw potentials expressed in terms of the Schwarz function (Proposition~\ref{prop_R Hele-Shaw}).

\subsection{Variational formula of the Liouville action} \label{sec conformal Liouville}

Recall that $\psi$ is the exterior conformal map defined in \eqref{def of psi exterior}. Unless stated otherwise, we always assume that the exterior conformal map extends holomorphically to an open neighbourhood of $\overline{\mathbb{D}}_e$. 
Let $\Gamma = \partial D$ denote the boundary of $D$, which we assume throughout to be a smooth, positively oriented Jordan curve. 
Let $\phi = \psi^{-1} : D_e \to \D_e$. 
Then the outward unit normal vector $\n$ and the unit tangent vector $\t$ to $\Gamma$ at the point $z = \psi(w)$ can be expressed in terms of $\psi$ and $\phi$ as 
\begin{equation}\label{def of normal tangent}
    \n(z)   = \frac{\phi(z) \overline{\phi'(z)}}{|\phi'(z)|} = \frac{w \, \psi'(w)}{|\psi'(w)|}, \qquad \t(z) = i\n(z) = \frac{i\phi(z) \overline{\phi'(z)}}{|\phi'(z)|}=  \frac{iw\, \psi'(w)}{|\psi'(w)|},
\end{equation}
respectively. 

Note that $\n$ and $\t$ extend from $\Gamma$ to real-analytic vector fields in a neighbourhood of $D_e$, since $\phi'$ and $\psi'$ do not vanish on $D_e$. 
For functions defined in a neighbourhood of $\Gamma$, the directional derivatives in the normal and tangential  directions are given, respectively, by
\begin{equation}
    \partial_\n  = \n \partial + \bar{\n}\bar{\partial}, \qquad \partial_\t = \t \partial + \bar{\t}\bar{\partial} = i (\n \partial -\bar{\n}\bar{\partial}).
\end{equation}
We recall that the (signed) curvature $\kappa$ of the curve $\Gamma$ is given by 
\begin{equation}\label{def of curvature}
    \kappa(z) = \kappa(\psi(w)) := -i\partial_\t \arg \n = \partial \n -\bar{\n}^2\bar{\partial} \n = \frac{1}{|\psi'(w)|}\Big(1+ \re \frac{w\psi''(w)}{\psi'(w)}\Big);
\end{equation}
see also \eqref{pre Schwar derivative change of variable} for the derivation of the final identity. 
Among the important objects in conformal computations is the \textit{Schwarzian derivative}, defined as
\begin{equation}\label{def of Schwarzian}
    \{\phi, z\} := \frac{\phi'''(z)}{\phi'(z)} -\frac{3}{2}\Big(\frac{\phi''(z)}{\phi'(z)}\Big)^2 = N_\phi'(z) -\frac12 N_\phi(z)^2,  
\end{equation}
where 
\begin{equation} \label{def of pre Sch N_phi}
N_{\phi}(z):= \frac{\phi''(z)}{ \phi'(z) } =  \frac{\ud}{\ud z}\log \phi'(z) 
\end{equation}
is the pre-Schwarzian derivative of $\psi$. 

We now investigate the variational formula for the Liouville action associated with an exterior conformal map. 
Consider a family of exterior conformal maps $\{\psi_t\}_{t \in I}$ governed by a \textit{Löwner--Kufarev equation} 
\begin{equation}\label{eq_LK evolution}
    \dot{\psi}_t(w) = w \, \psi_t'(w) P(t,w), \qquad w\in \bar{\D}_e.
\end{equation}
Here, $\dot{X}$ and $X'$ denote differentiation of $X$ with respect to $t$ and $w$, respectively. 
We assume that $P(t,w)$ is smooth in $t$ and bounded and holomorphic on $\D_e$, and extends analytically to a neighbourhood of $\bar{\D}_e$ with respect to $w$. 
On the unit circle $\T$, the normal velocity of the evolution is given by 
\begin{equation}\label{eq_normal velocity}
    v_{\n}(t,w)= \dot{\psi}_t (w)\cdot \n =\re \Big(\frac{ \dot{\psi}_t(w) \overline{w\psi_t'(w)}}{|\psi_t'(w)|}\Big)=|\psi_t'(w)| \re P(t,w), \quad w\in \T,
\end{equation}
where $\n=\n(\psi_t(w))$ denotes the outward unit normal vector at $\psi(w)\in \partial D$. 

Recall that the Liouville action of an exterior conformal map is given by \eqref{def of Liouville action}. Then we have the following. 

\begin{prop}[\textbf{Variational formula of the Liouville action}]\label{prop_var Liouville}
Let $\{\psi_t\}_{t\in I}$ be a family of exterior conformal maps evolving according to the Löwner--Kufarev equation~\eqref{eq_LK evolution}. 
Then the variation of the Liouville action with respect to the parameter $t$ is given by 
\begin{equation}\label{eq_var Liouville}
    \frac{\ud}{\ud t}\SS[\psi_t] = -\frac{2}{\pi}\int_\T \Big(\re (w^2\{\psi_t, w\})+ |\psi'_t(w)|^2 \kappa(\psi_t(w))^2\Big) \re P(t,w) \ud s(w).
\end{equation} Here, $\ud s$ is the arc length measure on curves.  
\end{prop}
 
The variational formula~\eqref{eq_var Liouville} and its derivation can be found in the physics literature; see ~\cite[Appendix~B]{ZW06}. Related variational formulas for the universal Liouville action, arising in the context of the Löwner energy of Weil--Petersson class loops, also appear in~\cite{SW24}. For the convenience of the reader, we provide here a self-contained and rigorous proof of Proposition~\ref{prop_var Liouville}.

We also note that the special case $\re P(t,w)=1/|\psi_t'(w)|^2$ of the variational formula~\eqref{eq_var Liouville} appears in~\cite{GV06} in the context of Hele-Shaw flows. Our argument may be viewed as an extension of the proof of~\cite[Theorem~7.4.10]{GV06} to the general Löwner--Kufarev setting.

\begin{proof}[Proof of Proposition~\ref{prop_var Liouville}]

As $\psi'(w)$ is a bounded holomorphic function on a neighbourhood of $\bar{\D}_e$, it follows that 
\begin{equation}\label{eq_psi infty}
    \frac{1}{2\pi}\int_\T \log \psi'(w) \ud s(w) = \frac{1}{2\pi i}\int_\T \frac{\log \psi'(w)}{w}\ud w = \log \psi'(\infty).
\end{equation} 
Recall that $X'$ denotes the derivative of $X$ with respect to $w$. Due to Löwner--Kufarev equation~\eqref{eq_LK evolution}, we have 
\begin{equation*}
    \frac{\ud}{\ud t}\log \psi_t'(w) 
    = \Big(1+ w N_{ \psi_t }(w) \Big)P(t,w) + wP'(t,w),
\end{equation*} 
where $N_\psi$ is the pre-Schwarzian derivative~\eqref{eq_pre Schwarzian}.
By using \eqref{eq_psi infty} and $\ud \bar{w}= -i \bar{w}\ud s(w)$, we obtain
\begin{align*}
    \frac{\ud }{\ud t}\SS[\psi_t] &=
    2\re \int_{\D_e} \overline{N_{\psi_t}(w)}\frac{\ud}{\ud t}N_{\psi_t}(w) \ud A(w) - \frac{2}{\pi}\re \int_\T \frac{\ud}{\ud t} \log \psi_t'(w) \ud s(w)\\
    &=  2\re \bigg[\frac{1}{2\pi i }\int_\T \Big(\overline{ N_{ \psi_t }(w) }+\frac{2}{\bar{w}}\Big)\frac{\ud}{\ud t}\log \psi'_t(w) \ud \bar{w} \bigg]  =   -\re \big( I_1 + I_2 \big), 
\end{align*} 
where 
\begin{align}
I_1 &:= \frac{1}{\pi} \int_\T \Big(2+ \overline{ w N_{ \psi_t }(w) }\Big)\Big(1+w N_{ \psi_t }(w)  \Big) P(t,w)\ud s(w) ,
\\
I_2 &:= \frac{1}{\pi} \int_\T \Big(2+ \overline{ w N_{ \psi_t }(w) } \Big) wP'(t,w) \ud s(w). 
\end{align}

Since $P(t,w)$ and $N_{\psi_t}(w)$ are bounded holomorphic functions on a neighbourhood of $\bar{\D}_e$ with asymptotic behaviour
\begin{equation*}
    P(t,w) = O(1), \qquad N_{ \psi_t }(w) = O(1/w^2),  
\end{equation*}
as $w \to \infty$, it follows that
\begin{equation*}
    \frac{1}{\pi }\int_\T w N_{ \psi_t  }(w)P(t,w) \ud s(w)= \frac{1}{\pi i }\int_\T N_{ \psi_t  }(w)P(t,w)\ud w=0.
\end{equation*}
Consequently, we obtain 
\begin{equation}\label{eq_int I1}
    \re I_1 = \frac{1}{\pi}\int_\T \Big(1+\Big|1+ w N_{ \psi_t }(w)  \Big|^2\Big)\re P(t,w)\ud s(w).
\end{equation}

On the other hand, since $\ud P = i w P'\,\ud s$, integration by parts yields
\begin{align*}
    I_2 &= \frac{1}{\pi i}\int_\T \Big(2+ \overline{ w N_{ \psi_t }(w)  }\Big) \ud P(t,w) = -\frac{1}{\pi i}\int_\T P(t,w) \overline{ (w N_{\psi_t  }(w))'  }  \ud \bar{w}  
    \\
    &= \frac{1}{\pi}\int_\T P(t,w)\Big(\overline{w^2\{\psi_t, w\} + \frac{1}{2}\big( w N_{ \psi_t  }(w) \big)^2+ w N_{ \psi_t  }(w) }\Big)\ud s(w),
\end{align*}
where we recall that $\{\psi_t, w\}$ denotes the Schwarzian derivative defined in~\eqref{def of Schwarzian}. 
Note that
\begin{align*}
&\quad  \frac{1}{\pi}\int_\T P(t,w)\Big( w^2\{\psi_t, w\} + \frac{1}{2}\big( w N_{ \psi_t  }(w) \big)^2+ w N_{ \psi_t  }(w)  \Big)\ud s(w)
\\
&= \frac{1}{\pi i}\int_\T P(t,w)\Big( w^2\{\psi_t, w\} + \frac{1}{2}\big( w N_{ \psi_t  }(w) \big)^2+ w N_{ \psi_t  }(w)  \Big) \frac{\ud w}{w} = 0 
\end{align*}
since the integrand in the second expression is holomorphic on a neighbourhood of $\bar{\D}_e$ and satisfies the asymptotic bound $O(1/w^{2})$ as $w \to \infty$. Therefore, we obtain 
\begin{equation}\label{eq_int I2}
    \re I_2 = \frac{2}{\pi} \int_\T \re P(t,w) \re \Big( w^2\{\psi_t, w\} + \frac{1}{2}\big( w N_{ \psi_t  }(w) \big)^2+ w N_{ \psi_t  }(w)  \Big)\ud s(w). 
\end{equation}

Combining \eqref{eq_int I1} and~\eqref{eq_int I2}, together with the identity
\begin{equation*} 
    1 + \Big|1+w N_{ \psi_t }(w)\Big|^2 + \re\big(w N_{ \psi_t }(w)\big)^2+2\re \big( w N_{ \psi_t }(w) \big) = 2\Big(1+ \re  \big( w N_{ \psi_t }(w) \big)   \Big)^2 = 2|\psi'_t(w)|^2 \kappa( \psi_t(w) )^2,
\end{equation*}
where $\kappa$ is the curvature defined in~\eqref{def of curvature}, we obtain the desired conclusion~\eqref{eq_var Liouville}. 
\end{proof}

\subsection{Algebraic Hele-Shaw potentials and quadrature domains}
 
The use of conformal mapping methods in the study of the macroscopic behaviour of two-dimensional Coulomb gases has been extensively explored in the physics literature \cite{WZ00,MWZ00,KKMWZ01}, particularly in connection with integrable hierarchies and string equations. The underlying idea is to consider the conformal map from the droplet onto a reference domain, typically the unit disc, and to characterise this map via a Liouville-type argument; see e.g. \cite[Appendix A]{By24} for further exposition.  

A central concept in this construction, particularly for defining the analytic continuation of the conformal map to the entire complex plane, is the \textit{Schwarz function}. 
To recall its definition, let $\Omega \subsetneq \wh{\C}$ be a domain, not necessarily bounded, such that $\Omega = \interior \overline{\Omega}$ and $\infty \notin \partial \Omega$. By definition, the Schwarz function of $\Omega$ is a continuous function $\mathsf{S}_\Omega : \overline{\Omega} \to \wh{\C}$ which is meromorphic in $\Omega$ and satisfies 
\begin{equation}
    \mathsf{S}_\Omega(z) =\bar{z}, \quad \text{for }z\in\partial \Omega.
\end{equation}
In particular, for an exterior domain $D_e$ and its conformal map $\psi: \D_e\to D_e$ we have
\begin{equation}
    \overline{\mathsf{S}_{D_e}(z)} = \psi(1/ \overline{\psi^{-1}(z)}), \quad z\in \bar{D}_e.
\end{equation}

Using conformal mapping methods, the equilibrium measure associated with external potentials of the form 
\begin{equation}\label{eq_poly potential}
    Q(z) =\frac{1}{t_0} \Big(|z|^2 - 2\re \sum_{j=1}^k \frac{t_j}{j} z^j\Big), \quad t_0>0, \; t_1,\ldots, t_k \in \C,
\end{equation}
was derived in~\cite{EF05}. (For $k \ge 3$, this requires an appropriate localisation to a bounded domain.) Observe that $\Delta Q = 1/t_0$ is constant, which implies that the corresponding equilibrium measure $\ud \mu_W$ has uniform density $1/(\pi t_0)$ on its droplet $K$. Moreover, the coefficients $t_j$ in~\eqref{eq_poly potential} encode the area and the exterior harmonic moments of $K$: for $j=1,\dots, k$, 
\begin{equation}\label{eq_harmonic moment rel}
    t_0 = \frac{1}{\pi}\text{area }K, \qquad t_j =-\int_{K^c} z^{-j}\ud A(z).
\end{equation}
The boundary of the droplet $\partial K$ is identified as an algebraic curve, often referred to as a polynomial curve. Equivalently, the exterior conformal map $\psi : \D_e \to K_e$ is a rational function of the form 
\begin{equation*}
    \psi(w) = R \Big(w + a_0 + \frac{a_1}{w} + \ldots + \frac{a_k}{w^k}\Big). 
\end{equation*}

In the seminal paper~\cite{LM16}, Lee and Makarov extended this class of external potentials to a broader family exhibiting similar properties, namely the \emph{algebraic Hele-Shaw potentials}, defined as 
\begin{equation}\label{def of algebraic HeleShaw}
    Q(z) = \frac{1}{t_0} \Big(|z|^2-  H(z)\Big), \qquad  H(z):\text{real harmonic and}\;\; \partial H(z) : \text{rational}.
\end{equation}
One observes that the polynomial potentials~\eqref{eq_poly potential}, as well as the potentials with insertions~\eqref{def of Q induced elliptic},~\eqref{def of Q GinUE with insertion},~\eqref{potQ}, all belong to this class. 

Here we recall some basic notions from~\cite{LM16}. Given an algebraic Hele-Shaw potential $Q$, a compact set $K \subset \C$ is called a  \emph{local droplet}
\footnote{The definition of local droplets presented here corresponds to \emph{non-maximal} local droplets in~\cite{LM16}. We will always refer to non-maximal local droplets throughout the paper.} 
of $Q$ if there exists an open neighbourhood $U$ of $K$ such that $K$ is the droplet associated with a localisation of $Q$ to $U$: 
\begin{equation} \label{def of localised Q}    Q_{\bar{U}} = Q\cdot \mathbbm{1}_{\bar{U}} + \infty \cdot \mathbbm{1}_{\bar{U}^c}.
\end{equation}
Note that the (global) droplet of $Q$ is, in particular, a local droplet of $Q$. The first part of the Lee--Makarov theorem relates local droplets to domains with a special property, known as \emph{quadrature domains}. 

A domain $\Omega$ is called a quadrature domain if there exists a rational function $r_\Omega$ such that the following \emph{quadrature identity} holds: 
\begin{equation}\label{eq_quadrature identity}
    \int_\Omega f(z) \ud A(z) = \frac{1}{2\pi i}\int_{\partial\Omega} f(z) r_\Omega(z) \ud z,
\end{equation}
for all analytic functions $f$ on $\Omega$ that are integrable and continuous on $\overline{\Omega}$. 
It is evident that if $\Omega$ is a quadrature domain, such rational function $r_\Omega$ is uniquely determined. We call $r_\Omega$ the \emph{quadrature function} of $\Omega$. One of the main results in~\cite{LM16}, relating this property to the equilibrium measure problem, is the following. 

\begin{thm}[\textbf{cf. Theorem 3.3 in}~\cite{LM16}]\label{thm_local droplet}
Let $K$ be a local droplet associated with an algebraic Hele-Shaw potential $Q$ of the form~\eqref{def of algebraic HeleShaw}. Then the complement $K^c$ is a finite union of disjoint quadrature domains $\Omega_1, \ldots, \Omega_q$, whose corresponding quadrature functions $r_1, \ldots, r_q$ satisfy 
\begin{equation} \label{LM sum of quad functions}
    r_1(z) + \ldots +r_q(z) = \partial H(z).
\end{equation} 
\end{thm}

It is well known that a domain $\Omega$ is a quadrature domain if and only if it admits a Schwarz function. In this case, the Schwarz function and the quadrature function are related by 
\begin{equation}\label{eq_Schwarz QD}
    \mathsf{S}_\Omega(z) = r_\Omega(z) + \int_{\Omega^c}\frac{\ud A(\zeta)}{z-\zeta}, \quad z \in \bar{\Omega}.
\end{equation}
Another useful characterisation of local droplets is given as follows. 

\begin{lem}[\textbf{cf. Lemma 3.2 in}~\cite{LM16}]\label{lem_variational equality}
Given an algebraic Hele-Shaw potential $Q$ of the form~\eqref{def of algebraic HeleShaw}, a compact set $K$ is a local droplet if and only if $\mathrm{area}(K) = \pi t_0$ and there exists a constant $C$ such that 
\begin{equation}\label{eq_Euler lagrange}
    C = \int_K \log \frac{1}{|z-\zeta|^2}\ud A(\zeta) + t_0 Q(z),
\end{equation}
for all $z\in K$.
\end{lem}
From the viewpoint of logarithmic potential theory, Lemma~\ref{lem_variational equality} remains valid for droplets associated with general external potentials; see~\cite{ST97}. The constant $C$ is commonly referred to as the \emph{(modified) Robin constant}. If $K$ is not only a local droplet but also the (global) droplet of $Q$, we denote $C = C_Q$ and refer to it as the Robin constant associated with the external potential $Q$.

We also note that a classical result~\cite{AS76} asserts that a simply connected domain is a quadrature domain if and only if its conformal map is rational, in agreement with the polynomial curves obtained in~\cite{EF05}. This characterisation has been used to derive explicit conformal maps of droplets arising in the context of non-Hermitian random matrix theory; see~\cite{BBLM15,By24,BFL25,BY25,ABK21}. 

An alternative, yet equivalent, approach to computing the conformal map is to solve the string equation in terms of exterior moments~\cite{WZ00,EF05}; see also~\eqref{eq_integrability} below. For explicit implementations of this method and its connection to vector equilibrium problems, we refer to~\cite{BS20,KKL25}.

The second part of the Lee--Makarov theorem in \cite{LM16} classifies all possible topologies of the droplet for a given degree of the rational function $\partial H$, under the assumption that the boundary is non-singular. This result also played a crucial role in the authors' previous work~\cite{BY25}, where it was used to determine all possible topological phases corresponding to the external potential~\eqref{potQ}.  
For boundaries with singularities, such as cusps or double points, a classification including the number and types of singularities has been recently obtained in~\cite{MR25}. For further background on quadrature domains, we refer to~\cite{GS05} and the references therein.

For a simply connected (unbounded) quadrature domain, the Liouville action~\eqref{def of Liouville action} admits an explicit representation due to the rationality of the conformal map. Suppose that the normalised exterior conformal map is given by a rational function 
\begin{equation}
    \psi'(w) = \psi'(\infty) \frac{(w-p_1) \ldots (w-p_r)}{(w-q_1)\ldots (w-q_r)}.
\end{equation}
By standard residue computations (see, for instance,~\cite[Section~6]{ZW06}), it is straightforward to verify that
\begin{equation}\label{eq_rational action}
    \SS[\psi] = -4\log \psi'(\infty) - \sum_{j=1}^r \log \psi'(1/\bar{p}_j) + \sum_{j=1}^r \log \psi'(1/\bar{q}_j).
\end{equation}
Using~\eqref{eq_rational action}, we obtain explicit expressions for the Liouville action of connected domains; see Remark~\ref{rem_explicit Liouville}.

\subsection{Determinantal Coulomb gases and planar orthogonal polynomials}\label{sec RNM OP}

As previously mentioned, a remarkable feature of determinantal Coulomb gases is their underlying integrable structure. We briefly review several aspects of this structure, together with the basic properties of the planar orthogonal polynomials~\eqref{def of planar OP}. 

Given an external potential $Q$ satisfying the growth condition
\begin{equation}\label{eq_growth cond}
    \liminf_{|z|\to \infty}\frac{Q(z)}{2\log |z|}> 1,
\end{equation}
we define the Bergman polynomial kernel associated with $L^2(\C, e^{-NQ(z)} \ud A)$ by 
\begin{equation}\label{def of polynomial kernel}
    \K_N(z,z') = \sum_{n=0}^{N-1}\frac{P_{n,N}(z) \overline{P_{n,N}(z')}}{h_{n,N}},  
\end{equation}
where $P_{n,N}$ are the monic planar orthogonal polynomials defined via the orthogonality relation~\eqref{def of planar OP}. Note that $\K_N$ is the reproducing kernel of the space of complex polynomials of degree less than $N$, endowed with the inner product 
\begin{equation}
    \langle p_1, p_2\rangle_{e^{-NQ}\ud A} := \int_\C p_1(z) \overline{p_2(z)} e^{-NQ(z)}\ud A(z).
\end{equation}
The Bergman polynomial kernel encodes the $m$-point correlation functions of the eigenvalues. In particular, the one-point function of the determinantal Coulomb gas \eqref{eq_RNM} can be expressed as 
\begin{equation}\label{def of one point function}
    \rho_N(z) := \frac{1}{N}\int_{\C^{N-1}} \mathbb{P}_N(z,z_2,\ldots, z_N)\prod_{j=2}^N \ud A(z_j) =\frac{1}{N}\K_N(z,z) e^{-NQ(z)}.
\end{equation}

The asymptotic behaviour of planar orthogonal polynomials has been extensively studied for special classes of external potentials using the Riemann--Hilbert approach. An incomplete list of works in this direction includes~\cite{BBLM15,BFKL26,BK12,DMMS25,KKL25,LY17,LY23,BEG18,BGM17}. 
In addition, fine asymptotic expansions (up to an error of order $O(N^{-k_0-1})$) can be obtained via partial Schlesinger transformations; see e.g.~\cite{BL08,LY17,LY23,BSY25,KLY25,BY23}. 

For more general potentials, a major breakthrough was achieved by Hedenmalm and Wennman~\cite{HW21}, who generalised results from special cases within the framework of foliation flow. This approach was later reformulated in~\cite{Hed24} in terms of a soft Riemann--Hilbert problem. 
Here we present the main theorem of~\cite{HW21,Hed24} in a form suited to our purposes.

\begin{thm}[\textbf{cf. Theorem 1.3 in \cite{HW21}, Theorem  1.2 in \cite{Hed24}}]\label{thm_Hedenmalm Wennman}
Let $Q$ be a $C^2$-smooth external potential satisfying the growth condition~\eqref{eq_growth cond}, associated with the connected (possibly multiply connected) droplet $K$. 
Let $\Omega_\infty$ denote the unbounded component of $K^c$, and let $\Gamma = \partial \Omega_\infty$ be the outermost boundary of $K$. 

Let $\phi : \Omega_\infty \to \D_e$ be the inverse of normalised exterior conformal map of $\Omega_\infty$, and let $\QQ$ be the unique bounded holomorphic function in a fixed neighbourhood of $\overline{\Omega}_\infty$ such that $\re \QQ = Q$ on $\Gamma$ and $\im \QQ(\infty) = 0$. We further assume that $Q$ is real-analytic and strictly subharmonic in a neighbourhood of $\Gamma$.

Then, for any positive integer $k_0$, there exist bounded holomorphic functions $F_k$ ($k=0,1,\ldots,k_0$), defined in a fixed neighbourhood of $\Omega_\infty$, such that 
\begin{equation}\label{def of HW expansion}
    P_{N,N}(z) = (\phi'(\infty))^{-N-1}(\phi(z))^N e^{N\QQ(z)/2}\Big(\sum_{k=0}^{k_0}N^{-k}F_k(z) +O(N^{-k_0-1})\Big), 
\end{equation}
as $N \to \infty$, where the error term is uniform for all $z \in \C$ satisfying 
\begin{equation}\label{eq_Omega dist}
    \textup{dist}(z, \Omega_\infty) \le (N^{-1}\log N)^{1/2}.
\end{equation}
\end{thm}

\begin{rem}[Hedenmalm--Wennman theorem for $M=tN$] \label{rem HW for M=tN}
The asymptotic expansion~\eqref{def of HW expansion} also remains valid for the general planar orthogonal polynomial $P_{M,N}(z)$ in the regime $M=tN \to\infty$ where $t>0$ is fixed. Indeed, if the external potential $Q$ satisfies the growth condition
\begin{equation}
    \liminf_{|z|\to \infty}\frac{Q(z)}{2\log |z|}>t 
\end{equation}
for some $t>0$, in place of~\eqref{eq_growth cond}, then Theorem~\ref{thm_Hedenmalm Wennman} applies after the rescaling $Q\mapsto Q/t$. 

More precisely, suppose that the rescaled external potential $Q/t$ satisfies~\eqref{eq_growth cond} together with the assumptions of Theorem~\ref{thm_Hedenmalm Wennman}. Let $K$, $\Omega_\infty$, $\Gamma$, and $\phi$ be defined analogously for the potential $Q/t$. Furthermore, let $\QQ$ denote the unique bounded holomorphic function in a fixed neighbourhood of $\overline{\Omega}_\infty$ satisfying $\re \QQ= Q$ on $\Gamma$ and $\im \QQ(\infty)=0$. Then,  
there exist bounded holomorphic functions $F_k(\cdot; t)$ $(k=0,1,\ldots, k_0)$ such that
\begin{equation}\label{def of HW inhomogeneous}
    P_{M,N}(z) = (\phi'(\infty))^{-M-1}(\phi(z))^M e^{N\QQ(z)/2}\Big(\sum_{k=0}^{k_0}N^{-k}F_k(z;t) +O(N^{-k_0-1})\Big), 
\end{equation}
for $z\in \C$ satisfying~\eqref{eq_Omega dist}
as $M=tN \to \infty$ along integers. 

As we explain below, the coefficients $F_k(z;t)$ are completely determined by the local behaviour of the potential $Q$ near $\Gamma$. In particular, $F_k(z;t)$ depend continuously on $t$, provided that the conformal map $\phi$ varies continuously with respect to $t$. In this situation, the outermost boundary $\Gamma$ deforms continuously, and the droplet $K$ undergoes no topological phase transition.
\end{rem}



\begin{rem}[Hedenmalm--Wennman theorem for multiply connected droplets]
One aspect of the statement in Theorem~\ref{thm_Hedenmalm Wennman} that differs from its original formulation in~\cite{HW21,Hed24} is that it applies not only to simply connected droplets, but also to multiply connected ones, provided that the droplet itself is connected. Indeed, although the authors state their main results for the simply connected case, the proof extends to the multiply connected setting as well (as confirmed by the authors). 

To give an intuitive explanation, we note that the harmonic measure at infinity for $\Omega_\infty$, supported on $\Gamma$, is asymptotically given by $|P_{N,N}|^2 e^{-NQ}$; see, for instance,~\cite{HM13,BBLM15,AHM11}. Formally, this suggests that the asymptotic behaviour of $P_{N,N}$ observed at infinity is governed by data supported in a neighbourhood of $\Gamma$, and is therefore insensitive to the multiple connectedness of the droplet (or the presence of non-trivial inner boundary components); cf. \cite{AC26}. 
\end{rem}

\begin{rem}[Hedenmalm--Wennman theorem for localised potentials] \label{Rem_HW localised potentials}
Although Theorem~\ref{thm_Hedenmalm Wennman} assumes that the external potential is $C^2$-smooth on the entire complex plane, this assumption is only needed to ensure the regularity of $\partial K$ and of the harmonic continuation $\re \QQ$; see~\cite{HM13} and~\cite[Section~2.1]{HW21}. 

Consequently, Theorem~\ref{thm_Hedenmalm Wennman} also applies to localised algebraic Hele-Shaw potentials~\eqref{def of localised Q}, provided that the associated local droplet exists. Indeed, the boundary of a local droplet is a union of algebraic curves~\cite{GS05,LM16}, and the harmonic continuation $\re \QQ$ admits the explicit representation~\eqref{eq_QQ algebraic Hele Shaw}. Furthermore, the localised potential~\eqref{def of localised Q} is real-analytic and strictly subharmonic on a neighbourhood of $K$. See Theorem~\ref{thm_F0F1} below. 
\end{rem}

In addition to the asymptotic expansion~\eqref{def of HW expansion}, a constructive algorithm for deriving the subleading terms $F_k$ is presented in~\cite{HW21,Hed24}. For our purposes, in view of the deformation framework in Proposition~\ref{prop_key}, it is necessary to identify explicitly the first subleading correction term $F_1$. However, implementing the algorithm in~\cite{HW21,Hed24} at a fully explicit level is technically highly demanding. To the best of our knowledge, no such explicit implementation has appeared in the literature. 

Indeed, although the results of~\cite{HW21,Hed24} have recently been applied in the determinantal 2D Coulomb gas literature to a number of important problems, such as macroscopic correlation functions~\cite{AC23}, anomalous fluctuations~\cite{AC26}, smallest gaps~\cite{Ch25}, and number variance~\cite{MMO25}, these applications are largely restricted to the use of leading-order asymptotics in~\eqref{def of HW expansion}. 

In contrast, the main contribution of the present paper as an application of the Hedenmalm--Wennman theory lies in the identification of subleading corrections, which are essential for our proof strategy (and, we believe, optimal for establishing our main results). These contributions can be summarised as follows:
\begin{itemize}
    \item In Theorem~\ref{thm_F0F1} below, we provide a fully explicit description of the first $1/N$ correction term $F_1$ in the general algebraic Hele-Shaw setting. 
    \smallskip
    \item In the special case~\eqref{potQ} of algebraic Hele-Shaw potentials, we characterise its direct relation to zeta-regularised spectral determinants and the Liouville action.
\end{itemize} 


We conclude this subsection by briefly outlining the algorithm for constructing the terms $F_k$ in Theorem~\ref{thm_Hedenmalm Wennman}. We first define 
\begin{equation}\label{def of R}
    R(z) = Q(z)- \re \QQ(z) - 2\log |\phi(z)|,
\end{equation}
in a fixed neighbourhood of $\overline{\Omega}_\infty$. This function is closely related to the logarithmic potential of the equilibrium measure, as well as to the obstacle function associated with $Q$. It will play an important role in the algorithm, and for this purpose it is convenient to have an explicit description. In fact, for algebraic Hele-Shaw potentials, the function $R$ admits an explicit representation in terms of the Schwarz function $\mathsf{S}_{\Omega_\infty}$, which we now formulate. 

\begin{prop}\label{prop_R Hele-Shaw}
Let $K$ be a connected local droplet of area $\pi t_0$ associated with an algebraic Hele-Shaw potential $Q$ of the form \eqref{def of algebraic HeleShaw}. Let $\Omega_\infty$, $\phi$, $\Gamma$, and $\QQ$ be as defined in Theorem~\ref{thm_Hedenmalm Wennman}. Then, for a fixed point $z_0 \in \Gamma$, the function $R$ defined in~\eqref{def of R} is given by 
\begin{equation}\label{def of R Hele-Shaw}
    R(z) = \frac{1}{t_0}\bigg( |z|^2 -|z_0|^2-2\re \int_{z_0}^z\mathsf{S}_{\Omega_\infty}(u) \ud u\bigg),
\end{equation}
which is real-analytic in a neighbourhood of $\overline{\Omega}_\infty$.
\end{prop}

Equivalent formulas to~\eqref{def of R Hele-Shaw} for some special case of potentials have appeared in the literature, see e.g.~\cite[Eq.~(2.24)]{BBLM15} and~\cite[Eq.~(32)]{LR16}, which in turn trace back to~\cite[Section~3]{WZ00} and~\cite[Eq.~(7)]{MWZ00}. 
For completeness, we provide a proof in a more general setting using the quadrature domain framework of~\cite{LM16}.

\begin{proof}[Proof of Proposition~\ref{prop_R Hele-Shaw}]
Note that by \eqref{def of R}, it is enough to show that 
\begin{equation}\label{eq_QQ algebraic Hele Shaw}
    \QQ(z) = Q(z_0)+\frac{2}{t_0} \int_{z_0}^z \big(\mathsf{S}_{\Omega_\infty}(u) -\partial H(u)\big)\ud u  - 2 \log \phi(z)
\end{equation}
for an arbitrarily chosen point $z_0 \in \Gamma = \partial \Omega_\infty$. To establish~\eqref{eq_QQ algebraic Hele Shaw}, it suffices to verify the characterising properties of $\QQ$, namely that $\re \QQ = Q$ on $\Gamma$ and $\im \QQ(\infty) = 0$. The first property follows immediately, since $\re \QQ(z_0) = Q(z_0)$ and, by differentiating the right-hand side of~\eqref{eq_QQ algebraic Hele Shaw}, we obtain
\begin{equation}
    \frac{1}{t_0} \big(\bar{z}- \partial H(z)\big)  = \partial Q(z), \quad z\in \Gamma,
\end{equation}
where we used that $\mathsf{S}_{\Omega_\infty}(z) = \bar{z}$ on $\Gamma = \partial \Omega_\infty$. It therefore remains to show that $\QQ$ is bounded near infinity and $\im \QQ(\infty) =0$.

Since $\partial H$ is rational, there exist constants $k \in \mathbb{N}$, $t_j$, $s_0$, and $s_1$ such that 
\begin{equation}\label{eq_Hele Shaw H}
    H(z)= |z|^2- t_0Q(z) =  2\re \bigg(\sum_{j=1}^k \frac{t_j}{j}z^j + s_0+s_1 \log z + O(\frac{1}{z})\bigg),
\end{equation}
as $z \to \infty$. 
On the other hand, by Theorem~\ref{thm_local droplet}, we have 
\begin{equation}
    K^c = \Omega_1\cup \ldots \cup \Omega_d \cup \Omega_\infty,
\end{equation}
where $\Omega_1, \ldots, \Omega_d$ are disjoint, simply connected, bounded quadrature domains, and $\Omega_\infty$ is the simply connected unbounded quadrature domain. Let $r_1, \ldots, r_d$ and $r_\infty$ denote their respective quadrature functions. 

Note that, by~\eqref{LM sum of quad functions}, the quadrature function $r_\infty$ is uniquely determined by the requirement that it captures the principal part of $\partial H$ at infinity. In particular, it admits the expansion 
\begin{equation}
    r_\infty(z) = \sum_{j=1}^k t_jz^{j-1} +O(\frac{1}{z}), \qquad z \to \infty. 
\end{equation}  
On the other hand, by the quadrature identity~\eqref{eq_quadrature identity} for $\Omega_j$, $j=1,\ldots,d$, we have 
\begin{equation*}
    \frac{1}{\pi}\text{area }\Omega_j= \frac{1}{2\pi i}\int_{\partial \Omega_j} r_j(z) \ud z = - \Res_{z=\infty} r_j. 
\end{equation*}
It follows that $r_j$ has a simple pole at infinity with residue $-\frac{1}{\pi}\mathrm{area}(\Omega_j)$, and hence 
\begin{equation}
    r_j(z) = \frac{\text{area }\Omega_j}{\pi z} + O(\frac{1}{z^2}), \qquad z \to \infty. 
\end{equation} 
Combining the above with~\eqref{eq_Schwarz QD} and $\mathrm{area}(K) = \pi t_0$, we obtain
\begin{align*}
    \mathsf{S}_{\Omega_\infty}(z) -\partial H(z) &= r_\infty(z) + \int_{\Omega_\infty^c} \frac{\ud A(\zeta)}{z-\zeta} - \partial H(z) = \int_{\Omega_\infty^c}\frac{\ud A(\zeta)}{z-\zeta}- \sum_{j=1}^d r_j(z)\\
    &= \frac{\text{area }\Omega_\infty^c}{\pi z} - \sum_{j=1}^d \frac{\text{area }\Omega_j }{\pi z} + O(\frac{1}{z^2}) = \frac{t_0}{\pi z} + O(\frac{1}{z^2}), \quad z\to \infty.
\end{align*} 
Notice that, by definition of the normalised exterior conformal map \eqref{def of psi exterior}, we have
\begin{equation}
\log \phi(z) =  \log \phi'(\infty) + \log z +O(\frac{1}{z}),\qquad z \to \infty. 
\end{equation}
Therefore, $\QQ$ remains bounded and the imaginary part of the right-hand side of~\eqref{eq_QQ algebraic Hele Shaw} vanishes as $z \to \infty$, in accordance with the characterising condition $\im \QQ(\infty) = 0$. This establishes the desired identity~\eqref{eq_QQ algebraic Hele Shaw} and completes the proof. 
\end{proof}

\begin{rem}
As a byproduct of the above proof, we also verify that~\eqref{eq_harmonic moment rel} holds not only for polynomial potentials~\eqref{eq_poly potential}, but also for algebraic Hele-Shaw potentials with simply connected local droplets. In this case, since $K^c = \Omega_\infty$, we readily obtain 
\begin{equation}
    -\int_{K^c} z^{-j} \ud A(z) = \frac{1}{2\pi i}\int_{\partial K} z^{-j}\mathsf{S}_{\Omega_\infty}(z)\ud z = \frac{1}{2\pi i}\int_{\partial K} z^{-j}r_\infty(z)\ud z = t_j. 
\end{equation}
\end{rem}

We proceed to give a brief overview of the algorithm in \cite{Hed24} for constructing the functions $F_k$; a more detailed discussion will be given in Section~\ref{sec simply} with a concrete example. 
We write 
\begin{equation}
\label{def of F}
    F = \sum_{k=0}^{k_0} N^{-k} F_k +O(N^{-k_0-1}),
\end{equation}
where $F_k$ are bounded holomorphic functions in a fixed neighbourhood of $\overline{\Omega}_\infty$. Our goal is to determine the coefficients $F_k$ to the desired order of precision.

In addition, we introduce auxiliary functions $A$ and $B$ via the asymptotic expansions
\begin{equation} \label{def of AB} 
A = \sum_{k=0}^{k_0} N^{-k}A_k +O(N^{-k_0-1}), \qquad B=\sum_{k= 0}^{k_0} N^{-k}B_k+O(N^{-k_0-1}),
\end{equation}
where $A_k$ are bounded holomorphic functions in a fixed neighbourhood of $\overline{\Omega}_\infty$, while $B_k$ are real-analytic functions defined in a neighbourhood of $\Gamma$. These functions $A$ and $B$ arise in the ansatz for the asymptotic expansion of the orthogonal polynomials; see~\cite[Eq.~(1.6.5)]{Hed24}. 
Finally, we impose the normalisation conditions 
\begin{equation}\label{eq_FA asymp infty}
    F(z) = \phi'(\infty) + O(\frac{1}{z}), \qquad A(z) = O(\frac{1}{z}),  
\end{equation}
as $z \to \infty$. 

Since $R(z) = 0$ and $\Delta R = \Delta Q > 0$ on $\Gamma$, there exists a real-analytic function $\wh{R}$ defined in a neighbourhood of $\Gamma$ such that 
\begin{equation}
\wh{R}^2 = R
\end{equation}
and $\wh{R} > 0$ in a neighbourhood outside $\Gamma$. Based on the soft Riemann--Hilbert formulation, the \textit{master equation} governing the asymptotic functions $F$, $A$, and $B$ takes the form
\begin{equation}\label{def of master eq}
    \sqrt{2}A\bar{\partial}\wh{R} + \frac{1}{N} \bar{\partial}B- B\bar{\partial}R =\sqrt{2\pi}\bar{F} 
\end{equation}
on a fixed small neighbourhood of $\Gamma$.\footnote{Note that~\eqref{def of master eq} differs from~\cite[Eq.~(1.7.3)]{Hed24} by certain constant factors of $\sqrt{2}$, due to our choice of planar measure $e^{-NQ}\, \ud A$ rather than $e^{-2NQ}\, \ud A$.}

The master equation~\eqref{def of master eq} can be solved iteratively via the following two steps. First, since $\bar{\partial} R$ vanishes on $\Gamma$, the equation at order $k$ in the $1/N$-expansion of~\eqref{def of master eq} reduces to a Riemann--Hilbert jump problem 
\begin{equation} \label{eq_iteration Ak}
    \sqrt{2}A_k \bar{\partial}\wh{R} +\bar{\partial}B_{k-1} =\sqrt{2\pi}\bar{F}_k,
\end{equation}
where $B_{k-1}$ has already been constructed at the previous stage of the iteration. Since $A$ and $F$ are bounded and holomorphic in a neighbourhood of $\overline{\Omega}_\infty$, the jump problem can be solved using the Szeg\H{o} projection operator $\P_{H^2_{-,0}} : L^2(\T) \to H^2_{-,0}$, defined by 
\begin{equation}\label{def of Szego projection}
    \Big(\P_{H^2_{-,0}} f \circ\phi\Big)(z) = \frac{1}{2\pi}\int_\T \frac{w f(w)}{\phi(z)-w}\ud s(w).
\end{equation}
Here, $H^2_{-,0}$ denotes the codimension-one subspace of the conjugate Hardy space $H^2_{-}$ consisting of functions that vanish at infinity. The elements of $H^2_{-}$ are identified with holomorphic functions on $\D_e$.

Second, the presence of the factor $1/N$ in front of $\bar{\partial} B$ implies that the zeroth-order equation of~\eqref{def of master eq} involves only $A$ and $\overline{F}$. Thus, once $A$ and $F$ have been determined up to order $k$, the coefficient $B_k$ can be obtained by substituting $A_k$, $B_{k-1}$, and $F_k$ into~\eqref{def of master eq}, namely, 
\begin{equation} \label{eq_iteration Bk}
    B_k =\frac{\sqrt{2}A_k \bar{\partial}\wh{R} + \bar{\partial}B_{k-1} - \sqrt{2\pi}F_k}{\bar{\partial }R}.
\end{equation}
Consequently, starting from the zeroth-order equation on $\Gamma$, this procedure yields a constructive scheme for computing $A$, $B$, and $F$ to arbitrary order. 

To solve the master equation~\eqref{def of master eq}, a key step is to compute the derivatives of $R$ in a neighbourhood of $\Gamma$. We note that the Taylor expansion of $R$ near $\Gamma$ completely determines the asymptotic expansion of $F$, as is evident from~\eqref{def of master eq}.  
Moreover, for algebraic Hele-Shaw potentials, the representation~\eqref{def of R Hele-Shaw}, together with~\cite[Eq.~(39)]{LR16}, provides the necessary input. In Section~\ref{sec simply}, we compute $F_1$ based on these observations.

\section{Proof of Theorem~\ref{thm_ZW}} \label{Section_proof main}
 
In this section, we present the proof of Theorem~\ref{thm_ZW}, elaborating on the strategy outlined in Figure~\ref{Fig_diagram strategy}.

Recall from Figure~\ref{fig_deformation} that we consider the models associated with the potentials $Q^{\rm i}$ and $Q^{\rm e}$, defined in~\eqref{def of Q induced elliptic}, and use their free energies as references. 
We begin with the following lemma, which provides exact formulas for the reference free energies, together with their full asymptotic expansions. For this, recall that the Barnes $G$-function is defined recursively by 
\begin{equation} \label{Barnes G def}
G(z+1)=\Gamma(z)G(z),\qquad G(1)=1;  
\end{equation}
see \cite[Section 5.17]{NIST}. 

\begin{lem}[\textbf{Reference free energies}] \label{Lem_reference free energy}
Let $Q^{ \rm i }$ and $Q^{ \rm e }$ be given by \eqref{def of Q induced elliptic}. Then the associated partition functions are evaluated as 
\begin{equation} \label{ZN elliptic explicit Barnes}
\ZZ_N( Q^{ \rm i  }  )= N! \frac{G(N+cN+1)}{ G(cN+1) } N^{ -(c+\frac12)N^2-\frac{1}{2}N }, \qquad  \ZZ_N(Q^{ \rm  e })= (1-\tau^2)^{N/2} \,G(N+2)  N^{ -\frac12 N^2-\frac{1}{2}N }, 
\end{equation}
where $G$ is the Barnes $G$-function. 
Furthermore, they satisfy following asymptotic behaviours as $N \to \infty$:
\begin{align}
\begin{split} \label{Z ind Gin asy}
\log \ZZ_N( Q^{ \rm i  }  ) & = -\Big(  \frac{3}{4}+\frac{3c}{2}+\frac{c^2}{2}\log c -\frac{(c+1)^2}{2}\log(c+1) \Big) \, N^2 +\frac12 N \log N +\Big( \frac{\log(2\pi)}{2}-1 \Big)N
\\
&\quad   +\frac12 \log N+ \frac{\log(2\pi)}{2}+\frac{1}{12} \log \Big( \frac{c}{1+c}\Big) 
\\
&\quad +\sum_{k=1}^\infty \bigg( \frac{B_{2k}}{ 2k(2k-1) } \frac{1}{N^{2k-1}} + \frac{ B_{2k+2} }{ 4k(k+1) } \Big( \frac{1}{(c+1)^{2k}}-\frac{1}{c^{2k}}\Big) \frac{1}{N^{2k}} \bigg) ,
\end{split}
\\
\begin{split} \label{Z elliptic Gin asy}
 \log  \ZZ_N(Q^{ \rm e }) & = -\frac34\, N^2+\frac12 N \log N +\Big( \frac{\log(2\pi)}{2}-1 +\frac12 \log(1-\tau^2)\Big)N
\\
& \quad +\frac{5}{12} \log N+ \frac{\log(2\pi)}{2}+\zeta'(-1) +\sum_{k=1}^\infty \bigg( \frac{B_{2k}}{ 2k(2k-1) } \frac{1}{N^{2k-1}} + \frac{ B_{2k+2} }{ 4k(k+1) } \frac{1}{N^{2k}} \bigg),
\end{split}
\end{align}
where $B_k$ is the Bernoulli number \eqref{def of Bernoulli number}. 
\end{lem}

\begin{proof}  
The formulas associated with $Q^{\rm i}$ can be found in~\cite[Lemma~3.1]{BSY25}. In fact, the expression for $\ZZ_N(Q^{\rm i})$ follows directly from~\eqref{ZN in terms of det}, using that the monomials form an orthogonal polynomial system due to the radial symmetry of $Q^{\rm i}$, and that the resulting radial integrals yield Gamma functions.  
Then~\eqref{Z ind Gin asy} follows from the well-known asymptotic behaviours (see~\cite[Eqs.~(5.11.1), (5.17.5)]{NIST}): as $z \to +\infty,$ 
\begin{align} 
\label{gamma asymp}
 \log \Gamma(z)&= \Big(z-\frac12\Big) \log z - z +\frac12 \log(2\pi) +\sum_{k=1}^\infty \frac{B_{2k}}{ 2k(2k-1)z^{2k-1} }, \\
\begin{split} \label{Barnes G asymp}
\log G(z+1) & =\frac{z^2 \log z}{2} -\frac34 z^2+\frac{ \log(2\pi) z}{2}-\frac{\log z}{12}+\zeta'(-1)  + \sum_{k=1}^\infty \frac{ B_{2k+2} }{ 4k(k+1)  } \frac{1}{z^{2k}},
\end{split}
\end{align}
where $B_k$ is the Bernoulli number~\eqref{def of Bernoulli number}. 

For the formula of $\ZZ_N(Q^{\rm e})$, the key input is that the associated monic planar orthogonal polynomials are given by 
\begin{equation} \label{def of planar Hermite}
    P_{n,N}^{ \rm e }(z) = \big(\frac{\tau}{2N}\big)^{n/2} H_n\big(\sqrt{\tfrac{N}{2\tau}}z\big),
\end{equation}
where $H_n$ denotes the $n$-th Hermite polynomial. More precisely, these polynomials satisfy the orthogonality relation
\begin{equation}
\int_{ \C }  P_{n,N}^{ \rm e }(z) \overline{  P_{m,N}^{ \rm e }(z) }  e^{-NQ^{\rm e}(z)}\ud A(z) = \sqrt{1-\tau^2}\,n! \,\delta_{n,m};
\end{equation}
see e.g. \cite[Lemma 7]{ACV18}. It then follows from~\eqref{ZN in terms of prod norms} that one obtains an exact formula for $\ZZ_N(Q^{\rm e})$, and the asymptotic expansion of the corresponding free energy again follows from~\eqref{Barnes G asymp}. 
\end{proof}

Before proceeding, we illustrate it with an example that can be verified using exact formulas.

\begin{rem}[Exactly solvable example: the elliptic Ginibre ensemble $c=0$] \label{Rem_defomration for c=0 case}
Note that, when $c=0$, the potential $Q$ in~\eqref{potQ} reduces to the elliptic potential $Q^{\rm e}$ defined in~\eqref{def of Q induced elliptic}. As previously mentioned, the polynomials $P^{\rm e}_{n,N} = P_{n,N}\big|_{c=0}$ are given by~\eqref{def of planar Hermite}. Using the explicit coefficients of the Hermite polynomials \cite[Eq.~(18.5.13)]{NIST}, we have 
\begin{equation}\label{eq_Hermite}
    P_{n,N}^{ \rm e }(z) = z^n -\frac{\tau n(n-1)}{2N} z^{n-2}+ O(z^{n-4}), \quad z\to \infty.
\end{equation}
In particular,
\begin{equation}
\mathcal{A}_{n,N}\big|_{c=0} = 0, \qquad 
\mathcal{B}_{n,N}\big|_{c=0} = -\frac{\tau n(n-1)}{2N}.
\end{equation}
On the other hand, by \eqref{ZN elliptic explicit Barnes}, we have 
\begin{equation}
 \frac{\ud }{\ud \tau}\log \ZZ_N(Q) \Big|_{c=0} = -\frac{\tau}{1-\tau^2} N. 
\end{equation}
Combining the above identities, one readily verifies that Proposition~\ref{prop_key} (ii) holds in the case $c=0$.
\end{rem}

We now prove Proposition~\ref{prop_key}. 

\begin{proof}[Proof of Proposition~\ref{prop_key}]
We begin by proving~\eqref{eq_deform a}. Treating $a$ as a complex variable, we observe that 
\begin{equation*}
    Q(z+a) = Q^{\rm e}(z+a) - 2c\log |z|, \qquad \partial_a Q(z+a) = \frac{(\bar{z}+\bar{a}) - \tau (z+a)}{1-\tau^2},
\end{equation*}
where $Q^{\rm e}$ is defined in~\eqref{def of Q induced elliptic}. 
Note that, under the change of variables $z_j \mapsto z_j + a$, it follows from the definition~\eqref{def of ZN N-fold intgral} that 
\begin{align*}
\ZZ_N(a,c,\tau) = \int_{ \C^N }\prod_{j<k}|z_j-z_k|^2 \prod_{j=1}^N |z|^{2cN}e^{-NQ^{\rm e}(z_j+a)}\ud A(z_j). 
\end{align*}
Therefore, we obtain 
\begin{align*}
    \partial_a \log \ZZ_N(a,c,\tau) 
    &= -\frac{N}{\ZZ_N(a,c,\tau)}\int_{\C^N} \Big(\sum_{j=1}^N \frac{(\bar{z}_j+\bar{a}) - \tau (z_j+a)}{1-\tau^2}\Big)\prod_{j<k}|z_j-z_k|^2\prod_{j=1}^N |z|^{2cN}e^{-NQ^{\rm e}(z_j+a)}\ud A(z_j).
\end{align*}
Furthermore, by applying the reverse translation $z_j \mapsto z_j - a$, together with the definition of the one-point function~\eqref{def of one point function}, we have
\begin{align*}
    \partial_a \log \ZZ_N(a,c,\tau) &=
    -\frac{N}{\ZZ_N(a,c,\tau)} \int_{\C^N} \Big(\sum_{j=1}^N \frac{\bar{z}_j - \tau z_j}{1-\tau^2}\Big)\prod_{j<k}|z_j-z_k|^2 \prod_{j=1}^N e^{-NQ(z_j)}\ud A(z_j)\\
    &= -N^2 \int_\C \frac{\bar{z}-\tau z}{1-\tau^2}  \rho_N(z)\ud A(z) = -\frac{N}{1-\tau^2}\int_\C (\bar{z}- \tau z)\K_N(z,z)e^{-NQ(z)}\ud A(z). 
\end{align*}
Applying the same argument to $\bar{\partial}_a$, we obtain that the derivative of the free energy with respect to the real variable $a$ is given by 
\begin{equation} \label{derivative of ZN a in terms of 1pt}
    \frac{\ud }{\ud a} \log \ZZ_N(a,c,\tau)= (\partial_a + \bar{\partial}_a) \log \ZZ_N(a,c,\tau) = -\frac{2N}{1+\tau}
    \int_\C z \, \K_N(z,z)e^{-NQ(z)}\ud A(z).
\end{equation}
Here, we have used the fact that $\int_{\C} z \, \K_N(z,z)e^{-NQ(z)} \ud A(z)$ is real, due to the conjugate invariance $Q(z) = Q(\bar{z})$ of the potential $Q$.

On the other hand, by comparing the leading coefficients \eqref{defAB}, we have
\begin{equation}
    z P_{n,N}(z) = P_{n+1,N}(z) + (\AA_{n,N}- \AA_{n+1,N})P_{n,N}(z) + O(z^{n-1})
\end{equation}
as $z \to \infty$. 
Note that by the orthogonality \eqref{def of planar OP} of $\{P_{j,N}\}$, we have
\begin{equation*}
    \int_\C z|P_{n,N}(z)|^2 e^{-NQ(z)} \ud A(z)= \int_\C \Big((\AA_{n,N}-\AA_{n+1,N})P_{n,N}(z)\Big) \overline{P_{n,N}(z)}e^{-NQ(z)}\ud A(z) = (\AA_{n,N}- \AA_{n+1,N})h_{j,N}.
\end{equation*}
Then, by \eqref{derivative of ZN a in terms of 1pt} and \eqref{def of polynomial kernel}, we obtain 
\begin{align*}
    \frac{\ud}{\ud a}\log \ZZ_N(a,c,\tau) &= -\frac{2N}{1+\tau} \sum_{n=0}^{N-1}  \int_\C \frac{z |P_{n,N}(z)|^2}{h_{n,N}} e^{-NQ(z)} \ud A(z) = -\frac{2N}{1+\tau}\sum_{n=0}^{N-1}  (\AA_{n,N}-\AA_{n+1,N}) = \frac{2N}{1+\tau} \AA_{N,N},
\end{align*}
which leads to the desired identity \eqref{eq_deform a}.  

Next, we prove~\eqref{eq_deform tau}. To this end, it is convenient to introduce the rescaled potential
\begin{equation}
    \wt{Q}(z) := |z|^2 -\tau \re z^2 - 2c\log \Big|z-\frac{a}{\sqrt{1-\tau^2}}\Big|= Q(z\sqrt{1-\tau^2}) + c\log (1-\tau^2).
\end{equation}
We use a tilde to denote the corresponding rescaled objects. In particular, we write $\wt{P}_{n,N}$, $\wt{h}_{n,N}$, $\wt{\K}_N$, and $\wt{\ZZ}_N$ for the associated monic planar orthogonal polynomials, their squared norms, the Bergman polynomial kernel, and the partition function corresponding to the weight $e^{-N\wt{Q}} \ud A$, respectively. 

Note that by the scaling relation, we have 
\begin{equation} \label{scaling relations of wt P wt h}
  \wt{P}_{n,N}(z)  = (1-\tau^2)^{-n/2} P_{n,N}(z\sqrt{1-\tau^2}), \qquad    \wt{h}_{n,N} = (1-\tau^2)^{-1-n-cN}h_{n,N}.
\end{equation} 
In particular, the subleading terms in the expansion of \(\wt{P}_{n,N}(z)\) are directly related to those of \(P_{n,N}(z)\) via 
\begin{equation}
    \wt{\AA}_{n,N} = \frac{\AA_{n,N}}{\sqrt{1-\tau^2}}, \qquad \wt{\BB}_{n,N} = \frac{\BB_{n,N}}{1-\tau^2}.
\end{equation}
Furthermore, by ~\eqref{ZN in terms of prod norms} and \eqref{scaling relations of wt P wt h}, we have
\begin{equation} \label{relation btw ZN wt ZN}
    \log \ZZ_N  = \Big(\frac{1+2c}{2}N^2 + \frac{1}{2}N\Big)\log (1-\tau^2) + \log \wt{\ZZ}_N .
\end{equation}

Upon setting \(\wt{a} = a/\sqrt{1-\tau^2}\), an application of~\eqref{eq_deform a} to \(\wt{\ZZ}_N\) yields 
\begin{equation}\label{eq_tau deform 1}
    \frac{\ud}{\ud \wt{a}}\log \wt{\ZZ}_N = (1-\tau^2) \frac{2N}{1+\tau} \wt{\AA}_{N,N}=\frac{2(1-\tau)N}{\sqrt{1-\tau^2}}  \AA_{N,N}.
\end{equation}
On the other hand, differentiation with respect to \(\tau\) acts on both the coefficient of \(\re z^2\) and \(\wt{a}\). Recalling the conjugate invariance $\wt{Q}(z) = \wt{Q}(\bar{z})$, we obtain 
\begin{align} \label{tau deriva wt ZN}
    \frac{\ud }{\ud \tau}\log \wt{\ZZ}_N &= N  \int_\C z^2\wt{\K}_N(z,z) e^{-N\wt{Q}(z)}\ud A(z) +\frac{\tau a}{(1-\tau^2)^{3/2}} \frac{\ud}{\ud \wt{a}}\log \wt{\ZZ}_N.
\end{align}
To compute the first term on the right-hand side, we compare coefficients and observe that 
\begin{align}
\begin{split}
    z^2 \wt{P}_{n,N}(z) & = \wt{P}_{n+2,N}(z) + (\wt{\AA}_{n,N}-\wt{\AA}_{n+2,N}) \wt{P}_{n+1,N}(z)
    \\
    &\quad +\Big(\wt{\BB}_{n,N}-\wt{\BB}_{n+2,N}-(\wt{\AA}_{n,N}-\wt{\AA}_{n+2,N})\wt{\AA}_{n+1,N}\Big)\wt{P}_{n,N}(z) +O(z^{n-1}).
\end{split}
\end{align}
This gives rise to
\begin{align}\label{eq_tau deform 2}
\begin{split}
   &\quad N  \int_\C z^2 \wt{\K}_N(z,z) e^{-N\wt{Q}(z)}\ud A(z) = N\sum_{n=0}^{N-1}  \int_\C \frac{z^2 |\wt{P}_{n,N}(z)|^2}{\wt{h}_{n,N}} e^{-N\wt{Q}(z)} \ud A(z) 
    \\
    &= N\sum_{n=0}^{N-1} \Big(\wt{\BB}_{n,N}-\wt{\BB}_{n+2,N}-(\wt{\AA}_{n,N}-\wt{\AA}_{n+2,N})\wt{\AA}_{n+1,N}\Big)
    \\
    &= -N \Big(\wt{\BB}_{N,N} + \wt{\BB}_{N+1, N} - \wt{\AA}_{N,N}\wt{\AA}_{N+1,N}\Big) =-\frac{N}{1-\tau^2}\Big(\BB_{N,N}+\BB_{N+1,N}-\AA_{N,N}\AA_{N+1,N}\Big).
\end{split}
\end{align} 
Finally, by combining \eqref{relation btw ZN wt ZN},~\eqref{eq_tau deform 1}, \eqref{tau deriva wt ZN} and~\eqref{eq_tau deform 2} we obtain the desired identity~\eqref{eq_deform tau}.
\end{proof}

It is worth noting that the above proof method applies more generally to a larger class of algebraic Hele-Shaw potentials, such as those of the form~\eqref{eq_poly potential} with respect to deformations by exterior moments $t_j$.

In addition, while we assume $c \geq 0$ in the definition of the potential $Q$ in \eqref{potQ}, we emphasise that Proposition~\ref{prop_key}, being essentially a finite-$N$ identity, remains valid for arbitrary parameter values provided that the partition function is finite. In particular, one may consider the scaling $c=\gamma/N$, where $\gamma>-1$ is an arbitrary real parameter.

\begin{rem}[Deformation with respect to the charge strength $c$]\label{rem_deform c}
Instead of deforming in \(a\) and \(\tau\), one may also consider a deformation in \(c\). This approach has been used extensively in the study of Fisher--Hartwig singularities in both the Hermitian setting~\cite{Kra07, CG21} and the non-Hermitian setting~\cite{WW19, HW24a}.
In our present model, we have 
\begin{equation}\label{eq_deform c}
    \frac{\ud}{\ud c}\log \ZZ_N(a,c,\tau) = N\int_\C \log |z-a| \rho_N(z)\ud A(z).
\end{equation}
Since \(\log |z-a|\) appearing in the integrand on the right-hand side is no longer a polynomial, the integral cannot be reduced to finitely many coefficients of the orthogonal polynomials, in contrast to the argument above. Consequently, the associated deformation formula becomes more intricate; see ~\cite[Lemma~3.2]{WW19} for the case $\tau=0$. On the other hand,~\eqref{eq_deform c} may be interpreted as the logarithmic potential of the finite-\(N\) one-point density \(\rho_N\). From this perspective,~\cite{HW24a} develops a framework for proving the existence of a \(1/N\) expansion of the free energy, albeit deriving explicit formulas or identifying the corresponding terms with conformal geometric functionals remains challenging. While this approach is of independent interest, it is beyond the scope of the present work.
\end{rem}

As illustrated in Figure~\ref{Fig_diagram strategy}, the next main ingredient in the proof of the main result is the asymptotic behaviour of the subleading coefficients \eqref{defAB} of $P_{n,N}(z)$, which, by Proposition~\ref{prop_key}, determine the variation of the free energy. For this analysis, it is necessary to distinguish between Regimes~I and~II. 

We begin with Regime~I. As explained above, in the case of a doubly connected droplet, we make use of deformations with respect to both $a$ and $\tau$; see Figure~\ref{fig_deformation}. Consequently, we require the asymptotic behaviour of the coefficients appearing in both Proposition~\ref{prop_key} (i) and~(ii).

\begin{prop}[\textbf{Variation of the free energy: Regime I}] \label{Prop_OP coefficients doubly}
Let $Q$ be given as~\eqref{potQ}. Suppose that the associated droplet $K_Q$ is doubly connected, i.e. $(a,c,\tau)$ falls within \textup{Regime I}. Then, as $N \to \infty$, the coefficients $\AA_{n,N}$ and $\BB_{n,N}$ defined in~\eqref{defAB} satisfy the asymptotic behaviours
\begin{align}
    \AA_{N,N} &= acN + O(N^{-\infty}),\label{eq_sub_doubly0}\\
    \BB_{N,N} &= \frac{c^2a^2}{2}N^2 + \Big(\frac{ca^2}{2}-\frac{\tau(1+c)^2}{2}\Big)N +\frac{\tau(1+c)}{2} + O(N^{-\infty}), \label{eq_sub_doubly1}\\
    \AA_{N+1, N}&= acN + O(N^{-\infty}), \label{eq_sub_doubly2}\\
    \BB_{N+1, N}&= \frac{c^2a^2}{2}N^2 + \Big(\frac{ca^2}{2}-\frac{\tau(1+c)^2}{2}\Big)N -\frac{\tau(1+c)}{2} + O(N^{-\infty}). \label{eq_sub_doubly3}
\end{align}
Here, $O(N^{-\infty})$ indicates $O(N^{-k_0-1})$ for any positive integer $k_0$ and the error terms are uniform over $(a,c,\tau)$ in compact subsets of \textup{Regime I}.
\end{prop}

Next, we consider the variation of the free energy in Regime~II. 
Unlike the Regime~I case, it suffices here to consider only the deformation with respect to~$a$; see Figure~\ref{fig_deformation}.  
By Proposition~\ref{prop_key}, this reduces the problem to analysing the asymptotic behaviour of~$\mathcal{A}_{N,N}$. 
Recall that $\psi$ is the normalised exterior conformal map of $\Omega_\infty= K_Q^c$.

\begin{prop}[\textbf{Variation of the free energy: Regime II}]\label{prop_sub_simply}
Let $Q$ be given as~\eqref{potQ}. Suppose that the associated droplet $K_Q$ is simply connected, i.e. $(a,c,\tau)$ falls within \textup{Regime II}.  Then, as $N \to \infty$, the coefficient $\AA_{N,N}$ defined in~\eqref{defAB} satisfies the asymptotic behaviour
\begin{equation}\label{eq_sub_simply}
    \AA_{N,N} = -\frac{1+\tau}{2}\Big(\frac{\ud}{\ud a} I_Q[\mu_Q] \Big)N + \frac{1+\tau}{48}\Big(\frac{\ud}{\ud a}\SS[\psi]\Big) N^{-1}+O(N^{-2}). 
\end{equation}
Here, the error term is uniform over $(a,c,\tau)$ in compact subsets of \textup{Regime II}.  
\end{prop}

We emphasise that the error term in Proposition~\ref{prop_sub_simply} is uniform over $a$ in compact subsets of \textup{Regime~II}. 
Here, \textup{Regime~II} is understood to include the limiting case $a=\infty$. Consequently, the same error estimate remains valid in the limit $a \to \infty$.

We are now ready to prove Theorem~\ref{thm_ZW}.

\begin{proof}[Proof of Theorem~\ref{thm_ZW}]
We begin with the doubly connected case. Following the deformation path illustrated in Figure~\ref{fig_phases}, we decompose the free energy $\ZZ_N(Q) \equiv \ZZ_N(a,c,\tau)$ into three parts: 
\begin{align}\label{eq_double decompose}
\begin{split}
    \log \ZZ_N(Q) &= \log \ZZ_N(Q^{\rm i}) + \big(\log \ZZ_N(0,c,\tau)-\log \ZZ_N(0,c,0)\big) + \big(\log \ZZ_N(a,c,\tau)-\log \ZZ_N(0,c,\tau)\big)\\
    &= \log \ZZ_N(Q^{\rm i}) + \int_0^\tau \frac{\ud}{\ud \tau'} \log \ZZ_N(0,c,\tau') \ud \tau' + \int_0^a \frac{\ud}{\ud a'}\log \ZZ_N(a', c,\tau)\ud a'. 
\end{split}
\end{align}
Here, we recall that $Q^{\mathrm{i}}$, defined in
\eqref{def of Q induced elliptic}, corresponds to the potential
\eqref{potQ} with $a=\tau=0$.  
By Propositions~\ref{prop_key} and ~\ref{Prop_OP coefficients doubly}, we obtain 
\begin{align}
    \frac{\ud}{\ud \tau}\log \ZZ_N(0,c,\tau) &= \frac{\tau c^2}{1-\tau^2}N^2 -\frac{\tau }{1-\tau^2}N + O(N^{-\infty}), \label{eq_doubly deform tau}\\
    \frac{\ud}{\ud a}\ZZ_N(a,c,\tau) &= \frac{2ca}{1+\tau} N^2 + O(N^{-\infty}). \label{eq_doubly deform a}
\end{align}
Combining ~\eqref{eq_double decompose}, \eqref{eq_doubly deform tau}, \eqref{eq_doubly deform a} with the asymptotic expansion  ~\eqref{Z ind Gin asy}  of the reference free energy $ \log \ZZ_N(Q^{\rm i})$, for an arbitrary positive integer $k_0$ we obtain 
\begin{align}
\begin{split} \label{ZW expansion explicit Regime 1}
    \log \ZZ_N(Q) 
    &= -\Big( \frac{3}{4}+\frac{3c}{2} + \frac{c^2}{2}\log \Big( c(1-\tau^2) \Big) -\frac{(1+c)^2}{2}\log(1+c) - \frac{c a^2}{1+\tau} \Big) \, N^2 \\
    &\quad +\frac12 N \log N +\Big( \frac{\log(2\pi)}{2}-1 +\frac{1}{2}\log (1-\tau^2)\Big)N  +\frac12 \log N+ \frac{\log(2\pi)}{2}+\frac{1}{12} \log \Big( \frac{c}{1+c}\Big) \\
    &\quad +\sum_{k=1}^{k_0} \bigg( \frac{B_{2k}}{ 2k(2k-1) } \frac{1}{N^{2k-1}} + \frac{ B_{2k+2} }{ 4k(k+1) } \Big( \frac{1}{(c+1)^{2k}}-\frac{1}{c^{2k}}\Big) \frac{1}{N^{2k}} \bigg) + O(N^{-2k_0-1}).
\end{split}
\end{align}
Here, the integration of the error terms is justified by their uniformity over the relevant parameter range. 




It remains to verify that the coefficients appearing in the expansion \eqref{ZW expansion explicit Regime 1} coincide with those in \eqref{free energy expasion in main Thm}.
Observe first that the leading $N^2$ term agrees with $-I_Q[\mu_Q]$ by \eqref{weighted energy doubly connected}. Next, the coefficient of $N$ coincides by virtue of \eqref{eq_explicit entropy}. Note that since the Euler characteristic is $\chi=0$ in the doubly connected case, the logarithmic terms also agree. 

For the constant term, recall that the droplet is given by $K_Q= \mathsf{E}\setminus \mathsf{D}$ where $\mathsf{E}$ and $\mathsf{D}$ are defined in \eqref{droplet_doubly connected}. The associated normalised conformal maps are given by 
\begin{equation}
    \psi_\mathsf{E}(w) =  \sqrt{1+c}\Big(w+\frac{\tau }{w}\Big), \qquad \psi_\mathsf{D}(w) = a+\sqrt{c(1-\tau^2)} w.
\end{equation}
Applying~\eqref{eq_rational action}, the Liouville action of the droplet $K_Q$ is evaluated as 
\[
    \LL[K_Q] = \SS[\psi_\mathsf{E}] +\SS[\psi_\mathsf{D}] 
    = 2\log \Big(\frac{c}{1+c}\Big).
\]
Furthermore, since $\chi=0$, it follows from \eqref{evaluation of boundary terms kappa} that the $O(1)$ term in \eqref{free energy expasion in main Thm} becomes
\[
    \frac{\log (2\pi)}{2} + \chi\zeta'(-1) + \frac{1}{24}\LL[K_Q] -\frac{1}{24\pi} \int_{\partial K_Q} \log (\Delta Q) \kappa \ud s = \frac{\log (2\pi)}{2} +\frac{1}{12}\log \Big(\frac{c}{1+c}\Big),
\]
which agrees with the corresponding term in \eqref{ZW expansion explicit Regime 1}.
This completes the proof for the doubly connected case.

\medskip 

We now turn to the simply connected case. In contrast to the doubly connected setting, it suffices to consider only the deformation with respect to $a$; see Figure~\ref{fig_deformation}. However, additional care is required here, since we eventually take the limit $a \to \infty$ in order to use the ensemble associated with $Q^{\rm e}$ as the reference partition function.

Indeed, the potential $Q$ defined in~\eqref{potQ} diverges as $a \to \infty$. On the other hand, after a suitable renormalisation, the ensemble itself converges to the Coulomb gas associated with the external potential $Q^{\rm e}$ defined in~\eqref{def of Q induced elliptic}. More precisely, we write
\begin{equation} \label{Q in terms of Qe and log}
Q(z) = Q^{ \rm e }(z) -2c \log|z-a|= Q^{ \rm e }(z) -2c \log|1-z/a| -2c \log a. 
\end{equation}
As $a \to \infty$, the logarithmic term on the right-hand side converges to zero locally uniformly in $z$. Hence, modulo the additive constant $2c \log a$, the potential $Q$ converges to $Q^{\rm e}$. This additive constant contributes to the free energy $\log \ZZ_N(Q)$ through the term $(2c \log a) N^2$. Such a renormalisation procedure plays a crucial role in the analysis of the variation of the free energy. 

By \eqref{ZN in terms of det},
the free energy can be written as 
\begin{equation}\label{eq_free energy gNa}
    \log \ZZ_N(Q) = (2c\log a)N^2 +   g_N(a),
\end{equation}
where
\begin{equation}\label{def of gNa}
    g_N(a) := \log N! + \log \det \bigg[\int_\C z^j \bar{z}^k |1-z/a|^{2cN}e^{-NQ^{\rm e}(z)}\ud A(z)\bigg]_{j,k=0}^{N-1}. 
\end{equation}
It follows from Propositions~\ref{prop_key} and~\ref{prop_sub_simply} that as $N \to \infty,$  
\begin{equation}\label{eq_gN a expansion}
    g_N'(a) = -\Big(\frac{\ud}{\ud a}I_Q[\mu_Q]+\frac{2c}{a}\Big)N^2 + \frac{1}{24}\frac{\ud}{\ud a}\SS[\psi] + E_N(a), 
\end{equation}
where the error term $E_N(a)$ satisfies $|E_N(a)|\le \mathsf{f}(a)/N$ for some bounded continuous function $\mathsf{f}$, by the uniformity of the error estimates. Recall here that these error bounds remain valid in the limit $a \to \infty$.

Now, the free energy expansion essentially follows by integrating \eqref{def of gNa} over the interval $[a,\infty)$, and then combining the result with \eqref{eq_free energy gNa} and the reference free energy given in Lemma~\ref{Lem_reference free energy}. However, since the integration is performed over a non-compact domain, one needs to check that the error term $E_N(a)$ is integrable and that its contribution remains of order $O(N^{-1})$ after integration.

For this purpose, we first claim that $g_N(a)$ admits an asymptotic expansion in powers of $1/a$, i.e. for each non-negative integer $L$, there exist coefficients $h_\ell(N)$ such that 
\begin{equation} \label{gN(a) 1/a expansion}
    g_N(a) = \log \ZZ_N(Q^{\rm e}) + \sum_{\ell=1}^L \frac{h_\ell(N)}{a^\ell} + O\Big(\frac{1}{a^{L+1}}\Big), \qquad a\to \infty.
\end{equation}
By the determinantal representation~\eqref{def of gNa}, it suffices to verify that 
\[
    F(a) := \int_\C |1-z/a|^{2cN} f(z)\ud A(z)
\]
admits an asymptotic expansion in powers of $1/a$ for every Schwartz function $f$ on $\C$.  
Let $P_L(u,\bar{u})$ denote the Taylor polynomial of degree $L$ at $u=0$ of the function. Then there exist constants $\delta>0$ and $C>0$ such that
\[
    \Big||1-u|^{2cN} - P_L(u,\bar{u}) \Big| \le C |u|^{L+1}, \quad |u| \le \delta.
\]
It follows that 
\begin{align*}
    \bigg|F(a) - \int_\C P_L(z/a, \bar{z}/a) f(z) \ud A(z) \bigg| &\le \int_{|z|\le \delta a} \Big| |1-z/a|^{2cN}-P_L(z/a, \bar{z}/a)\Big||f(z)|\ud A(z) \\
    &\le \frac{C}{a^{L+1}}\int_\C |z|^{L+1} |f(z)| \ud A(z)  = O\Big(\frac{1}{a^{L+1}}\Big),
\end{align*}
where the final integral is finite since $f$ is a Schwartz function. 
Note that 
\begin{align*}
   \bigg|\int_{|z|>\delta a} P_L(z/a, \bar{z}/a) f(z) \ud A(z)\bigg|&\le C' \int_{|z|>\delta a} |z/a|^{L} |f(z)| \ud A(z) 
   \\
    &\le C' \Big(\sup_{z\in \C} |z|^{L+3}|f(z)| \Big) \int_{\delta a}^\infty \frac{2\ud r}{a^{L}r^2} = O\Big(\frac{1}{a^{L+1}}\Big) 
\end{align*}
for some constant $C'>0$. Also, we have 
\begin{align*}
     \bigg|\int_{|z|>\delta a}  |1-z/a|^{2cN} f(z) \ud A(z)\bigg|&\le \Big(1+\frac{1}{\delta}\Big)^{2cN} \int_{|z|>\delta a} |z/a|^{2cN}|f(z)|\ud A(z) \\
    &\le \Big(1+\frac{1}{\delta}\Big)^{2cN} \Big(\sup_{z\in \C} |z|^{2cN+L+3}|f(z)|\Big) \int_{|z|>\delta a} \frac{2\ud r}{a^{2cN}r^{ L+ 2}}  = O\Big(\frac{1}{a^{L+1}}\Big). 
\end{align*} 
Combining these estimates, it follows that 
\[
    \bigg|F(a) - \int_\C P_L(z/a, \bar{z}/a) f(z) \ud A(z) \bigg| = O\Big(\frac{1}{a^{L+1}}\Big), \qquad a\to \infty.
\]
Therefore, it follows that $F(a)$ admits an asymptotic expansion in powers of $1/a$. 
In particular, since $P_0(u,\bar{u})= 1$, we have 
\[
    F(a) = \int_\C f(z) \ud A(z) + O\Big(\frac{1}{a}\Big), \qquad a\to \infty.
\]
Then, by \eqref{def of gNa}, we obtain the desired expansion \eqref{gN(a) 1/a expansion}.  

Differentiating \eqref{gN(a) 1/a expansion} with respect to $a$, we have 
\begin{equation}
    g_N'(a) = -\frac{h_1(N)}{a^2} + O\Big(\frac{1}{a^3}\Big), \qquad a\to\infty.
\end{equation}
Comparing with~\eqref{eq_gN a expansion}, notice that
\begin{equation} 
    \frac{\ud}{\ud a}I_Q[\mu_Q] = -\frac{2c}{a} + O\Big(\frac{1}{a^2}\Big), \qquad \frac{\ud }{\ud a}\SS[\psi] = O\Big(\frac{1}{a^2}\Big),\qquad a\to \infty 
\end{equation}
as $a \to \infty$ and that 
\begin{equation}
E_N(a) = O\Big(\frac{1}{a^2N}\Big) \qquad a,N \to \infty. 
\end{equation}
Since these estimates justify the term-by-term integration of the asymptotic expansion of $g_N'(a)$, integrating \eqref{eq_gN a expansion} over $[a,\infty)$ together with \eqref{def of gNa} yields 
\begin{align} \label{ZW expansion simply v0}
     \log \ZZ_N(Q) 
    &=\log \ZZ_N(Q^{\rm e})-\Big(I_Q[\mu_Q] - \lim_{a\to \infty} (I_Q[\mu_Q] +2c \log a)\Big)N^2 + \frac{1}{24}( \SS[\psi]-\lim_{a\to \infty} \SS[\psi])  + O(N^{-1}),
\end{align}
as $N \to \infty$. 

Observe that by \eqref{def of log energy} and \eqref{Q in terms of Qe and log}, 
\begin{align*}
I_Q[\mu_Q] + 2c\log a = \int_{\mathbb{C}^2}\log \frac{1}{|z- w|}\ud \mu_Q(z)\ud \mu_Q(w) + \int_\mathbb{C} \Big( Q^{ \rm e }(z) -2c \log|1-z/a| \Big)  \ud \mu_Q(z) .
\end{align*}
Then, by the continuity of the logarithmic energy, we obtain 
\begin{align} \label{energy a infty limit} 
\lim_{ a \to \infty } ( I_Q[\mu_Q] + 2c\log a ) =  I_{Q^{\rm e}}[\mu_{Q^{\rm e}}]=\frac{3}{4}. 
\end{align}
Similarly, continuity of the Liouville action yields 
\begin{equation} \label{Liouville action a infty limit}
    \lim_{a\to \infty} \SS[\psi] = \SS[\psi_\mathsf{E}|_{c=0}] = -2\log (1-\tau^2)
\end{equation} 
since the droplet $K_Q$ converges to the ellipse $\mathsf{E}$ with $c=0$ in \eqref{droplet_doubly connected} along the deformation $a\to\infty$. 
Combining \eqref{ZW expansion simply v0}, \eqref{energy a infty limit}, and \eqref{Liouville action a infty limit} with the reference free energy expansion given in Lemma~\ref{Lem_reference free energy}, together with the identity $\LL[K_Q]=\SS[\psi]$ and \eqref{evaluation of boundary terms kappa}, we obtain the asymptotic expansion \eqref{free energy expasion in main Thm} with $\chi=1$.
\end{proof}

\section{Simply connected case}\label{sec simply}

In this section, we derive the asymptotic expansion of $P_{N,N}$ up to $1/N$ correction when the droplet $K$ is simply connected and prove Proposition~\ref{prop_sub_simply}. 

Throughout much of this section, we work in the general setting of algebraic Hele-Shaw potentials. In Subsection~\ref{Subsec_HW fine asymp}, we derive the first subleading correction term in the Hedenmalm--Wennman expansion of Theorem~\ref{thm_Hedenmalm Wennman}. In Subsection~\ref{Subsec_2D Toda}, we discuss the integrability relations for the tau-function arising in the dispersionless $2$D Toda hierarchy, which may be interpreted as the limiting free energy. Up to this point, our discussion remains within the general algebraic Hele-Shaw framework. Finally, in Subsection~\ref{Subsec_coefficient simply asymp}, we specialise to the particular potential~\eqref{potQ} and prove Proposition~\ref{prop_sub_simply}.

\subsection{Fine asymptotic behaviour of planar orthogonal polynomials} \label{Subsec_HW fine asymp}

Recall that Theorem~\ref{thm_Hedenmalm Wennman}, due to Hedenmalm and Wennman~\cite{HW21,Hed24}, provides the asymptotic expansion~\eqref{def of HW expansion} for planar orthogonal polynomials. To prove Proposition~\ref{prop_sub_simply}, we require an explicit characterisation of the first correction term $F_1$ appearing in~\eqref{def of HW expansion}.  
In the following theorem, we derive this subleading term for the general algebraic Hele-Shaw potentials with simply connected local droplets; cf.~\eqref{def of localised Q}. We also note that Theorem~\ref{thm_Hedenmalm Wennman} remains valid for localised droplets; see Remark~\ref{Rem_HW localised potentials}.

\begin{thm}[\textbf{\textup{Subleading term in the Hedenmalm--Wennman expansion; Hele-Shaw potentials}}]\label{thm_F0F1}
Let $K$ be a simply connected local droplet associated with an algebraic Hele-Shaw potential $Q$ of the form~\eqref{def of algebraic HeleShaw}. Let $U$ be an open neighbourhood of $K$ such that $K$ is the droplet associated with the localised potential $Q_{\bar U}$ defined in~\eqref{def of localised Q}. Let $\psi:\D_{\mathrm e}\to K^c$ be the normalised exterior conformal map onto the complement of $K$, and denote by $\phi:K^c\to \D_{\mathrm e}$ its inverse. 

Then for $P_{n,N}$ the monic planar orthogonal polynomials with respect to the measure $e^{-NQ} \mathbbm{1}_{\bar{U}}\ud A$, the asymptotic expansion~\eqref{def of HW expansion} holds with 
\begin{equation} \label{eq_F0}
F_0(z) = \sqrt{\phi'(\infty) \phi'(z)}
\end{equation}
and
\begin{equation}\label{eq_F1/F0}
    F_1(z)= \frac{F_0(z)}{24\pi(\Delta Q) }\int_\T \frac{w}{\phi(z)- w} \Big(3\re (w^2\{\psi ,w\}) - |\psi'(w)|^2\kappa^2(\psi(w)) +6\frac{w\psi''(w)}{\psi'(w)} +3\frac{|\psi''(w)|^2}{|\psi'(w)|^2}\Big)\frac{\ud s(w)}{|\psi'(w)|^2}.
\end{equation}
Here, $\kappa$ denotes the signed curvature~\eqref{def of curvature} of $\Gamma$, and $\{\psi,w\}$ denotes the Schwarzian derivative~\eqref{def of Schwarzian} of $\psi$. 

In particular, the coefficient of the leading $1/z$ term in the expansion of $F_1/F_0$ at infinity is given by 
\begin{equation}\label{eq_F1 1/z}
    \frac{F_1(z)}{F_0(z)}= \frac{1}{(\Delta Q) \phi'(\infty)} \frac{1}{24\pi } \bigg(\int_\T \Big(\re \big(w^2 \{\psi, w\}\big)+|\psi'(w)|^2 \kappa^2(\psi(w))\Big) \frac{w}{|\psi'(w)|^2}\ud s(w)\bigg)\frac{1}{z} + O(\frac{1}{z^2}),  
\end{equation} as $z \to \infty.$
\end{thm}

\begin{rem}
We expect that Theorem~\ref{thm_F0F1} will also be useful in other related contexts. For instance, in deriving subleading corrections to the edge scaling limits of determinantal Coulomb gas models, the first correction typically follows from the leading-order asymptotics in~\eqref{def of HW expansion} (see e.g.~\cite{Ch25,MMO25,Am25,ACC24}), whereas the second subleading term requires input of the type provided by Theorem~\ref{thm_F0F1}; see~\cite{LR16} for the elliptic Ginibre ensemble.  
Moreover, our results establish a connection between spectral determinants and planar orthogonal polynomials through their refined asymptotic behaviour. We believe that this connection is of independent interest and merits further investigation in a more general setting, which we plan to address in future work. 
\end{rem}


The proof of Theorem~\ref{thm_F0F1} involves several geometric quantities, including the curvature~\eqref{def of curvature}, as well as the Schwarzian~\eqref{def of Schwarzian} and pre-Schwarzian~\eqref{def of pre Sch N_phi} derivatives. Many of these quantities are expressed simultaneously in terms of the exterior conformal map together with the normal and tangent vector fields, a feature that plays an important role not only in the ensuing integral computations but also in tracking the underlying conformal-geometric structure. To facilitate these computations, we first collect several useful identities relating these objects. Although the derivations are fairly elementary, we include them for completeness. In the context of random matrix theory, closely related computations appeared in~\cite{LR16}. Throughout, we adopt the notation and framework introduced in Section~\ref{sec conformal Liouville}. 


For $z$ in a neighbourhood of $\overline{K^c}$, the pre-Schwarzian derivatives~\eqref{def of pre Sch N_phi} may be expressed in terms of the normal vector \eqref{def of normal tangent} as
\begin{equation}\label{eq_pre Schwarzian}
    \frac{\partial \n}{\n} = \frac{\phi'}{\phi} -\frac{1}{2}\frac{\phi''}{\phi'} , \qquad \frac{\bar{\partial} \n}{\n} = \frac{1}{2} \overline{\Big(\frac{\phi''}{\phi'}\Big)}.
\end{equation}
Furthermore, the derivative of the pre-Schwarzian admits two different representations:
\begin{align*}
    \partial\Big(\frac{\partial \n}{\n}\Big) &= \frac{\phi\phi'' - (\phi')^2}{\phi^2} -\frac{1}{2}\Big(\frac{\phi''}{\phi'}\Big)'=-\Big(\frac{\phi'}{\phi}-\frac{1}{2}\frac{\phi''}{\phi'}\Big)^2 + \frac{1}{4}\Big(\frac{\phi''}{\phi'}\Big)^2 - \frac{1}{2}\Big(\frac{\phi''}{\phi'}\Big)', \\  
    \partial\Big(\frac{\partial \n}{\n}\Big)&= \frac{\partial^2 \n}{\n} - \Big(\frac{\partial \n}{\n}\Big)^2= \frac{\partial^2 \n}{\n}-\Big(\frac{\phi'}{\phi}-\frac{1}{2}\frac{\phi''}{\phi'}\Big)^2.
\end{align*}
Comparing the two expressions yields the identity 
\begin{equation}\label{eq_Schwarzian conformal1}
    \{\phi,z\} = \Big(\frac{\phi''}{\phi'}\Big)'(z)-\frac{1}{2} \Big(\frac{\phi''}{\phi'}\Big)^2(z) = -2\frac{\partial^2\n(z)}{\n(z)},
\end{equation}
valid for $z$ in a neighbourhood of $\overline{K^c}$. Similarly, using the second identity in~\eqref{eq_pre Schwarzian}, we obtain 
\begin{equation}\label{eq_Schwarzian conformal2}
    \overline{\{\phi,z\}} = 
    \overline{\Big(\frac{\phi''}{\phi'}\Big)'}(z)-\frac{1}{2} \overline{\Big(\frac{\phi''}{\phi'}\Big)^2}(z)
    = 2\frac{\bar{\partial}^2\n(z)}{\n(z)}-4\frac{(\bar{\partial}\n)^2(z)}{\n^2(z)}. 
\end{equation} 

Next, we derive an identity for the tangential derivative of the curvature $\kappa$: 
\begin{equation}\label{eq_curvature der}
    \partial_\t \kappa(z) = \im (\n^2(z) \{\phi,z\}) ,\quad z\in \Gamma. 
\end{equation}
To verify this identity, observe first that since $|\n(z)|=1$ on $\Gamma$, it follows from~\eqref{eq_Schwarzian conformal1} and~\eqref{eq_Schwarzian conformal2} that
\begin{align*}
    \n^2(z) \{\phi, z\} -\overline{\n^2(z)\{\phi,z\}}&= -2\n(z) (\partial^2\n)(z) - 2\bar{\n}^3(z) (\bar{\partial}^2\n)(z) + 4\bar{\n}^4(z)(\bar{\partial }\n)^2(z).
\end{align*}
Furthermore,
\begin{equation*}
    0 = -i\partial_\t |\n|^2= \n^2 (\partial\bar{\n}) +\partial \n - \bar{\n}^2 (\bar{\partial}\n) - \bar{\partial}\bar{\n}. 
\end{equation*}
Since $\partial (\bar{\partial}\n/\n) =0$ by the second identity in~\eqref{eq_pre Schwarzian}, we obtain 
\begin{equation*}
    \partial \bar{\partial}\n = \partial\Big(\n \cdot \frac{\bar{\partial}\n}{\n}\Big) =  \frac{\partial \n \bar{\partial}\n}{\n}.
\end{equation*}
Combining all of the above, we have 
\begin{align*}
    i\partial_\t \kappa&= -(\n \partial - \bar{\n}\bar{\partial})(\partial \n - \bar{\n}^2\bar{\partial}\n ) = -\n (\partial^2 \n) + 2\bar{\n}\bar{\partial}\partial \n - \bar{\n}^3 (\bar{\partial}^2\n)+2(\partial\bar{\n}-\bar{\n}^2\bar{\partial} \bar{\n})(\bar{\partial}\n) \\
    &= -\n (\partial^2 \n) + 2\bar{\n}\bar{\partial}\partial \n - \bar{\n}^3(\bar{\partial}^2\n) -2 (\bar{\n}^2 \partial \n - \bar{\n}^4\bar{\partial}\n)(\bar{\partial}\n)= -\n (\partial^2 \n) - \bar{\n}^3(\bar{\partial}^2\n) + 2\bar{\n}^4(\bar{\partial} \n)^2,
\end{align*}
which gives the desired identity~\eqref{eq_curvature der}.

All of the identities above can be rewritten in terms of the exterior conformal map $\psi=\phi^{-1}$ by using
\begin{equation} \label{pre Schwar derivative change of variable}
    \phi'(\psi(w)) = \frac{1}{\psi'(w)}, \qquad \phi''(\psi(w)) = -\frac{\psi''(w)}{(\psi'(w))^2}.
\end{equation}
In particular, the curvature is given by 
\begin{align} 
    \kappa(\psi(w)) = \frac{1}{|\psi'(w)|}\Big(1+\re \frac{w\psi''(w)}{\psi'(w)}\Big).
\end{align}
Moreover, by the composition formula  
\begin{equation}
    \{f\circ g, z\} = \{f, g(z)\}(g'(z))^2 + \{g,z\} 
\end{equation}
for the Schwarzian derivative, we have
\begin{equation} \label{Schwar derivative change of variable}
    \{\phi, \psi(w)\} = -\frac{\{\psi, w\}}{\psi'(w)^2}.
\end{equation}

We are now ready to prove Theorem~\ref{thm_F0F1}. 

\begin{proof}[Proof of Theorem~\ref{thm_F0F1}]
Recall that $\n$ and $\t$ denote the normal and tangent vectors of $\Gamma$, respectively, given by~\eqref{def of normal tangent} on a fixed neighbourhood of $\overline{K^c}$. Recall also that the function $R$ is defined by~\eqref{def of R}. 

As explained in Section~\ref{sec RNM OP}, the analysis requires solving the master equation~\eqref{def of master eq}, for which one needs the Taylor expansion of $R$ along $\Gamma$. Using Proposition~\ref{prop_R Hele-Shaw}, this expansion can be expressed in terms of derivatives of the Schwarz function $\mathsf{S}=\mathsf{S}_{K^c}$.
 
Fix a point $z_0 \in \Gamma$. Note that, after lengthy but straightforward computations, 
\begin{align}\label{eq_R Taylor}
\begin{split}
    \frac{R(z)}{\Delta Q} &= |z-z_0|^2 + \frac{\bar{\n}^2(z_0)}{2}(z-z_0)^2 + \frac{\n(z_0)^2}{2}(\overline{z-z_0})^2 -\frac{\kappa(z_0) \bar{\n}^3(z_0)}{3}(z-z_0)^3 -\frac{\kappa(z_0) \n^3(z_0)}{3}(\overline{z-z_0})^3\\
    &\quad +\Big(\frac{\kappa^2(z_0)}{4}+\frac{i\partial_\t\kappa(z_0)}{12}\Big) \bar{\n}^4(z_0)(z-z_0)^4+ \Big(\frac{\kappa^2(z_0)}{4}-\frac{i\partial_\t\kappa(z_0)}{12}\Big) \n^4(z_0)(\overline{z-z_0})^4 + O(|z-z_0|^5)
\end{split}
\end{align}
as $z \to z_0$. 
More precisely, it was shown in \cite[Eq.~(39)]{LR16} (where the function $R(z)/\Delta Q$ is identified with $\Omega(z)$ in the notation of \cite{LR16}) that
\begin{align}
\begin{split}
\label{eq_R Taylor LR}
\frac{R(z)}{\Delta Q}&= 2X^2 -\frac{2\kappa}{3}(X^3-3XY^2) +\frac{\kappa^2}{2}(X^4-6X^2Y^2+Y^4) -\frac{2\partial_{ \t } \kappa}{3}(X^3Y-XY^3) +O(|z-z_0|^5) ,
\end{split}
\end{align} 
where $X,Y\in\mathbb R$ denote the normal and tangential coordinates, respectively, i.e. 
\[
z-z_0 = X\n(z_0)+Y\t(z_0) = \n(z_0)( X+iY ),  \qquad u:=\bar{\n}(z_0)(z-z_0)=X+iY. 
\] 
Substituting the expressions for $X$ and $Y$ in terms of the above complex coordinates into \eqref{eq_R Taylor LR}, since 
$$
X^3-3XY^2 = \re u^3, \qquad X^4-6X^2Y^2+Y^4= \re u^4, \qquad X^3Y-XY^3= \frac14 \im u^4, 
$$
one readily obtains \eqref{eq_R Taylor}.

It follows from \eqref{eq_R Taylor} that the Taylor expansion of its square root is given by 
\begin{align}\label{eq_Rhat Taylor}
\begin{split}
    \sqrt{\frac{2}{\Delta Q}}\wh{R}(z) &= \bar{\n}(z_0)(z-z_0) + \n(z_0)(\overline{z-z_0}) +\frac{\kappa(z_0)}{3} \Big(|z-z_0|^2 - \bar{\n}^2(z_0)(z-z_0)^2 - \n^2(z_0)(\overline{z-z_0})^2\Big)\\
    &\quad +\Big(\frac{7\kappa^2(z_0)}{36}+\frac{i\partial_\t\kappa(z_0)}{12}\Big) \bar{\n}^3(z_0)(z-z_0)^3
    -\Big(\frac{\kappa^2(z_0)}{12}+\frac{i\partial_\t\kappa(z_0)}{12}\Big)\bar{\n}(z_0)(z-z_0)^2(\overline{z-z_0})\\
    &\quad+\Big(\frac{7\kappa^2(z_0)}{36}-\frac{i\partial_\t\kappa(z_0)}{12}\Big) \n^3(z_0)(\overline{z-z_0})^3 -\Big(\frac{\kappa^2(z_0)}{12}-\frac{i\partial_\t\kappa(z_0)}{12}\Big)\n(z_0)(z-z_0)(\overline{z-z_0})^2\\
    &\quad+ O(|z-z_0|^4),
\end{split}
\end{align}
as $ z \to z_0$.

Note that by \eqref{eq_iteration Ak}, the zeroth-order equation of~\eqref{def of master eq} on $\Gamma$ is given by 
\begin{equation}
    \sqrt{2}A_0\bar{\partial}\wh{R} = \sqrt{2\pi} \bar{F}_0,\quad \text{on }\Gamma.
\end{equation}
Applying $\bar{\partial}$ to~\eqref{eq_Rhat Taylor}, we obtain
\begin{equation}\label{eq_Rhat on Gamma}
    \bar{\partial}\wh{R} = \sqrt{\frac{\Delta Q}{2}}  \n=\sqrt{\frac{\Delta Q}{2}} \frac{\phi \sqrt{\bar{\phi'}}}{\sqrt{\phi'}} , \quad \text{on }\Gamma,
\end{equation}
where we used \eqref{def of normal tangent}.
Hence, we have
\begin{equation}\label{eq_zeroth reduced}
    \sqrt{\frac{\Delta Q}{2\pi }}\frac{A_0 \phi}{\sqrt{\phi'}} = \overline{\Big(\frac{F_0}{\sqrt{\phi'}}\Big)}\quad \text{on }\Gamma.
\end{equation}
Since $A_0$ is a bounded holomorphic function on $K^c$ satisfying $A_0(z)= O(1/z)$ as $z\to \infty$, the left-hand side of~\eqref{eq_zeroth reduced} extends to a bounded holomorphic function on a neighbourhood of $\overline{K^c}$. On the other hand, since $F_0$ is a bounded holomorphic function on $K^c$ satisfying $F_0(\infty)=\phi'(\infty),$ the right-hand side of~\eqref{eq_zeroth reduced} extends to a bounded anti-holomorphic function with limit $\sqrt{\phi'(\infty)}$ at infinity. Therefore, both sides must be constant on $K^c$, yielding
\begin{equation}\label{eq_F0A0}
    F_0(z) =\sqrt{\phi'(\infty)\phi'(z)}, \qquad A_0(z) = \sqrt{\frac{2\pi}{\Delta Q}} \frac{\sqrt{\phi'(\infty)\phi'(z)}}{\phi(z)}
\end{equation}
on a neighbourhood of $\overline{K^c}$.

Performing the first iteration of~\eqref{def of master eq}, it follows from~\eqref{eq_iteration Bk} that
\begin{align}
    B_0 &= \frac{\sqrt{2}A_0 \bar{\partial}\wh{R} - \sqrt{2\pi}\bar{F}_0}{\bar{\partial}R}, 
\end{align}
which holds on a neighbourhood of $\Gamma$. 
Rearranging this identity gives
\[
    \sqrt{2}\,A_0\bar{\partial}\wh{R}
    =
    \sqrt{2\pi}\,\bar{F}_0+B_0\bar{\partial}R.
\]
Substituting~\eqref{eq_F0A0}, we obtain 
\begin{equation}\label{eq_WR}
    \sqrt{\frac{2}{\Delta Q}} \bar{\partial}\wh{R}(z) = \n(z) \Big(1+W_R(z)\frac{\bar{\partial}R(z)}{\Delta Q}\Big), \qquad W_R(z):= \frac{\Delta Q}{\sqrt{2\pi}}\Big(\frac{B_0}{\bar{F}_0}\Big)(z).
\end{equation}

Next, we explicitly determine $W_R$ and $\bar{\partial}W_R$ on $\Gamma$ using the Taylor expansions~\eqref{eq_R Taylor} and~\eqref{eq_Rhat Taylor}. 

Comparing the coefficients of $(z-z_0)$ in~\eqref{eq_WR} using the expansions~\eqref{eq_R Taylor} and~\eqref{eq_Rhat Taylor}, we obtain
\begin{equation*}
    \frac{\kappa(z_0)}{3} = \n(z_0)W_R(z_0)+ \partial\n(z_0).
\end{equation*}
It then follows from~\eqref{def of curvature} that
\begin{equation}  \label{eq_WR(z0) eval}
    W_R(z_0) = -\frac{2}{3}\bar{\n}(z_0)\partial\n(z_0) -\frac{1}{3}\bar{\n}^3 (z_0)\bar{\partial }\n(z_0).
\end{equation}
Similarly, by comparing  the coefficients of $(\overline{z-z_0})^2$ in~\eqref{eq_WR}, we have 
\begin{equation*}
    \Big(\frac{7\kappa^2(z_0)}{12}-\frac{i\partial_\t\kappa(z_0)}{4} \Big)\n^3(z_0) = \frac{1}{2}\bar{\partial}^2\n(z_0) +\n^2(z_0) \bar{\partial}\n(z_0) W_R(z_0) -\n^4(z_0)\kappa(z_0)W_R(z_0)  +\n^3(z_0) \bar{\partial}W_R(z_0),
\end{equation*}
which gives 
\begin{equation*} 
  \bar{\partial}W_R(z_0) =   \frac{7\kappa^2(z_0)}{12}-\frac{i\partial_\t\kappa(z_0)}{4}    - \frac{1}{2} \bar{\n}^3(z_0) \bar{\partial}^2\n(z_0) - \Big( \bar{\n} (z_0) \bar{\partial}\n(z_0) -  \n(z_0)\kappa(z_0) \Big) W_R(z_0) . 
\end{equation*}
Note that by ~\eqref{eq_Schwarzian conformal2} and \eqref{eq_curvature der}, we have 
\begin{align*}
-\frac{i\partial_\t\kappa(z_0)}{4}    - \frac{1}{2} \bar{\n}^3(z_0) \bar{\partial}^2\n(z_0) &=  -\frac{i}{4}\im (\n^2(z_0)\{\phi, z_0\}) -\frac{1}{4}\bar{\n}^2(z_0) \overline{\{\phi, z_0\}} - \bar{\n}^4(z_0) (\bar{\partial}\n)^2(z_0) 
\\
&= -\frac{1}{4}\re \big(\n^2(z_0)\{\phi,z_0\}\big)   - \bar{\n}^4(z_0) (\bar{\partial}\n)^2(z_0) . 
\end{align*}
Combining this with \eqref{def of curvature} and \eqref{eq_WR(z0) eval}, we obtain  
\begin{equation} \label{eq_WR(z0) derivative eval}
 \bar{\partial}W_R(z_0) = \frac{7\kappa^2(z_0)}{12}-\frac{1}{4}\re \big(\n^2(z_0)\{\phi,z_0\}\big) -\frac{2}{3}(\partial \n)^2(z_0) + \bar{\n}^2(z_0) \partial \n(z_0) \bar{\partial}\n(z_0) -\frac{1}{3}\bar{\n}^4(z_0)(\bar{\partial}\n)^2(z_0). 
\end{equation}

We now derive the formula for $F_1$. Applying~\eqref{eq_iteration Ak} with $k=1$ and using~\eqref{eq_Rhat on Gamma}, we obtain
\begin{equation}
    \sqrt{\Delta Q} A_1  \n+ \bar{\partial}B_0 =\sqrt{2\pi} \bar{F}_1, \quad \text{on }\Gamma. 
\end{equation}
Dividing by $\sqrt{2\pi}\,\bar{F}_0$ and using~\eqref{def of normal tangent}, it follows that 
\begin{equation}\label{eq_first reduced}
    \Big(\frac{\bar{F}_1}{\bar{F}_0}\Big)(z)= \frac{1}{\sqrt{2\pi}}\Big(\frac{\bar{\partial}B_0}{\bar{F}_0}\Big)(z) + \sqrt{\frac{\Delta Q}{2\pi}} \frac{A_1(z)\phi(z)}{\sqrt{\phi'(\infty) \phi'(z)}}, \quad z\in \Gamma.
\end{equation}
Since $A_1$ is a bounded holomorphic function on a neighbourhood of $\overline{K^c}$ satisfying $A_1(z)=O(1/z)$ as $z\to\infty$, the second term on the right-hand side of~\eqref{eq_first reduced} extends to a bounded holomorphic function on a neighbourhood of $\overline{K^c}$. On the other hand, $\overline{F_1/F_0}$ is a bounded anti-holomorphic function on the same domain satisfying $\overline{F_1/F_0}=O(1/\bar{z})$ as $z\to\infty$. Therefore, after pulling back to the unit circle via the exterior conformal map, the function $F_1/F_0$ is obtained by applying the Szeg\H{o} projection $\P_{H^2_{-,0}}$ given by \eqref{def of Szego projection}. Then by \eqref{eq_WR}, we have
\begin{equation}
    \Big(\frac{F_1}{F_0}\Big)(z) = \frac{1}{\sqrt{2\pi}}\P_{H^2_{-,0}}\Big[\Big(\frac{\partial\bar{ B_0}}{F_0}\Big)\circ \psi\Big](\phi(z))  = \frac{1}{\Delta Q}\P_{H^2_{-,0}}\Big[\Big(\partial \bar{W}_R +\bar{W}_R \frac{\partial F_0}{F_0} \Big)\circ \psi\Big](\phi(z)). 
\end{equation} 

Now, by \eqref{def of Szego projection}, to show \eqref{eq_F1/F0}, it remains to show that 
\begin{equation} \label{eq_input of the Szego projection}
\Big(\partial \bar{W}_R +\bar{W}_R \frac{\partial F_0}{F_0} \Big)\circ \psi(w)= \frac{1}{12|\psi'(w)|^2}\Big(3\re \big(w^2\{\psi, w\}\big)- |\psi'(w)|^2 \kappa^2(\psi(w))  +6\frac{w\psi''(w)}{\psi'(w)} +3\frac{|\psi''(w)|^2}{|\psi'(w)|^2}\Big)  \Big). 
\end{equation} 
Note that by \eqref{eq_F0} and \eqref{eq_pre Schwarzian}, we have 
\begin{align*}
\Big(\overline{\partial \bar{W}_R +\bar{W}_R \frac{\partial F_0}{F_0}} \Big)\circ \psi(w) &= \bar{\partial}W_R(\psi(w))  + \frac{1}{2}W_R(\psi(w)) \overline{\Big(\frac{\phi''}{\phi'}\Big)}(\psi(w))
\\
&=  \bar{\partial}W_R(\psi(w))  + \frac{1}{2}W_R(\psi(w)) \bar{\n}(\psi(w)) \bar{\partial} \n(\psi(w)). 
\end{align*}
Then it follows from \eqref{def of curvature}, \eqref{eq_WR(z0) eval} and \eqref{eq_WR(z0) derivative eval} that 
\begin{align*}
 \Big(\overline{\partial \bar{W}_R +\bar{W}_R \frac{\partial F_0}{F_0}} \Big)\circ \psi(w)  = -\frac{\kappa^2(\psi(w))}{12}-\frac{1}{4}\re\big(\n^2(\psi(w)) \{\phi, \psi(w)\}\big)  - \bar{\n}^2(\psi(w))(\partial \n)(\psi(w))(\bar{\partial}\n)(\psi(w)) .
\end{align*}
Furthermore, by combining 
$$
\n^2(\psi(w)) = w^2 \frac{ \psi'(w) }{ \overline{ \psi'(w) } }
$$
with \eqref{eq_pre Schwarzian}, \eqref{pre Schwar derivative change of variable}, \eqref{Schwar derivative change of variable}, we obtain 
\begin{align*}
 \Big(\overline{\partial \bar{W}_R +\bar{W}_R \frac{\partial F_0}{F_0}} \Big)\circ \psi(w)  =  \frac{1}{12|\psi'(w)|^2}\Big(3\re \big(w^2\{\psi, w\}\big)- |\psi'(w)|^2 \kappa^2(\psi(w)) + 3\overline{\Big(\frac{\psi''}{\psi'}\Big)}\Big(\frac{2}{w}+\frac{\psi''}{\psi'}\Big)\Big). 
\end{align*} 
Taking complex conjugates, we obtain the desired identity~\eqref{eq_input of the Szego projection}. This completes the proof of~\eqref{eq_F1/F0}.
 

\medskip 

Next, we prove \eqref{eq_F1 1/z}. By \eqref{eq_F1/F0}, it follows immediately that the coefficient of $z^{-1}$ in the Laurent expansion of $F_1/F_0$ at infinity is given by 
\begin{equation}\label{eq_F1 1/z begin}
    \frac{1}{(\Delta Q)\phi'(\infty)}\frac{1}{24\pi}\int_\T \Big(3\re (w^2\{\psi ,w\}) - |\psi'(w)|^2\kappa^2(\psi(w)) +6\frac{w\psi''(w)}{\psi'(w)} +3\frac{|\psi''(w)|^2}{|\psi'(w)|^2}\Big)\frac{w}{|\psi'(w)|^2}\ud s(w) .
\end{equation}
We derive \eqref{eq_F1 1/z} from \eqref{eq_F1 1/z begin} by applying integration by parts. Note that for a holomorphic function $f$ on $\T$,
\begin{equation*}
    \partial_\t f(w)= \Big(iw\partial -i\bar{w}\bar{\partial}\Big)f(w).
\end{equation*}
Thus, integration by parts on $\T$ yields 
\begin{align}\label{eq_integration by parts 1}
\begin{split}
    \int_\T \frac{w^2\psi''(w)}{\psi'(w)} \frac{\ud s(w)}{|\psi'(w)|^2} &= i\int_\T \frac{w}{\overline{\psi'(w)}} \partial_\t \Big(\frac{1}{\psi'(w)}\Big) \ud s(w)=-i\int_\T \frac{1}{\psi'(w)} \partial_\t \Big(\frac{w}{\overline{\psi'(w)}}\Big)\ud s(w)\\
    &= \int_\T\Big(1+ \overline{\Big(\frac{w\psi''(w)}{\psi'(w)}\Big)}\Big) \frac{w\ud s(w)}{|\psi'(w)|^2},
\end{split}
\end{align} 
which gives rise to 
\begin{align}
  \int_\T\Big(6\frac{w\psi''(w)}{\psi'(w)} + 3\frac{|\psi''(w)|^2}{|\psi'(w)|^2}\Big)\frac{w\ud s(w)}{|\psi'(w)|^2} = 3\int_\T \Big|1 + \frac{w\psi''(w)}{\psi'(w)}\Big|^2 \frac{w\ud s(w)}{|\psi'(w)|^2}. 
\end{align}
Using this identity, we can rewrite \eqref{eq_F1 1/z begin} as 
\begin{equation}\label{eq_F1 1/z mid}
    \frac{1}{(\Delta Q)\phi'(\infty)}\frac{1}{24\pi}\int_\T \Big(3\re \big(w^2\{\psi ,w\}\big) -|\psi'(w)|^2 \kappa^2(\psi(w)) +3\Big|1+\frac{w\psi''(w)}{\psi'(w)}\Big|^2\Big)\frac{w}{|\psi'(w)|^2}\ud s(w).
\end{equation}

Let 
\begin{equation} \label{def of eq F1 1/z mid integral}
    I:=\int_\T\Big(\re \big(w^2\{\psi, w\}\big)-|\psi'(w)|^2 \kappa^2(\psi(w)) +\frac{3}{2}\Big|1+\frac{w\psi''(w)}{\psi'(w)}\Big|^2\Big) \frac{w}{|\psi'(w)|^2}\ud s(w). 
\end{equation}
In order to show that \eqref{eq_F1 1/z mid} coincides with \eqref{eq_F1 1/z}, it suffices to verify that $I=0$. 
For notational convenience, let $f(w) = w\psi''(w)/\psi'(w)$. Then, by \eqref{def of curvature} and \eqref{def of Schwarzian}, we have 
\begin{align*}
    &\quad \re \big(w^2\{\psi, w\}\big)-|\psi'(w)|^2 \kappa^2(\psi(w)) +\frac{3}{2}\Big|1+\frac{w\psi''(w)}{\psi'(w)}\Big|^2\\
    &=\Big(\re (wf') -\re f-\frac{1}{2} \re f^2\Big) -\Big(1+2\re f +(\re f)^2\Big) +\frac{3}{2}\Big(1+2\re f +|f|^2\Big)  \\
    &= \frac{1}{2} + \re (wf')+2 (\im f)^2 = \frac{1}{2}  + \partial_\t (\im f)+ 2(\im f)^2,
\end{align*}
where we have used
\begin{equation}
    \partial_\t (\im f) = i(w\partial -\bar{w}\bar{\partial}) \frac{f-\bar{f}}{2i} = \frac{1}{2}(wf' +\bar{w}\bar{f'}) =\re (wf').
\end{equation}
Observe that
\begin{align}
    \partial_\t\Big(\frac{ w}{|\psi'(w)|^2}\Big) &= \frac{ iw}{|\psi'(w)|^2}-\frac{iw }{|\psi'(w)|^2}\Big(\frac{w\psi''(w)}{\psi'(w)}- \overline{\Big(\frac{w\psi''(w)}{\psi'(w)}\Big)}\Big) = \frac{i w +2w(\im f)}{|\psi'(w)|^2}.
\end{align}
Combining all of the above, we obtain 
\begin{align*}
    I&= \int_\T \Big(\frac{1}{2}+2(\im f)^2\Big) \frac{ w}{|\psi'(w)|^2}\ud s(w) + \int_\T \partial_\t (\im f) \frac{ w}{|\psi'(w)|^2} \ud s(w)\\
    &= \int_\T \Big(\frac{1}{2}+2(\im f)^2\Big) \frac{w}{|\psi'(w)|^2}\ud s(w) - \int_\T (\im f)  \frac{iw +2w(\im f)}{|\psi'(w)|^2} \ud s(w)\\
    &= \int_\T \Big(\frac{1}{2} -i (\im f)\Big)\frac{w}{|\psi'(w)|^2}\ud s(w) = \frac{1}{2}\int_\T \partial_\t \Big(\frac{w}{|\psi'(w)|^2}\Big)\ud s(w) = 0. 
\end{align*}
This completes the proof. 
\end{proof}

\subsection{Dispersionless 2D Toda hierarchy and logarithmic energy} \label{Subsec_2D Toda}

From now on, we restrict our attention to the case where the droplet $K$ is simply connected. By virtue of the deformation formula~\eqref{eq_deform a}, it remains to relate the asymptotic behaviour of $\AA_{N,N}$ to the variation of the coefficients appearing in the free energy expansion~\eqref{free energy expasion in main Thm}. 

The underlying viewpoint of our approach is that the exterior conformal map associated with the droplet may be interpreted through the dispersionless $2$D Toda hierarchy~\cite{WZ00, KKMWZ01, Te09, GTV14}, an integrable hierarchy describing deformations of planar domains by infinitely many time parameters. More precisely, the coefficients of the potential play the role of the generalised times of the hierarchy, while the free energy corresponds to the logarithm of the associated tau-function. In this framework, variations of the free energy are naturally encoded by deformations of the conformal map and of the harmonic moments of the droplet. 

Together with the exterior moments~\eqref{eq_harmonic moment rel}, we also introduce the interior harmonic moments
\begin{equation}\label{eq_interior harmonic moment}
   M_j = \int_K  z^j \ud A(z), \quad j\ge 1.
\end{equation}
From the viewpoint of the dispersionless $2$D Toda hierarchy, the relevant integrability relations take the form 
\begin{equation}\label{eq_integrability}
    \frac{\partial}{\partial t_j} \log \tau = 
    \begin{cases}
        M_j/j, & j\ge 1,
        \smallskip 
        \\
        -C_Q, & j=0,
    \end{cases}\qquad
    \text{and } \quad \frac{\partial t_\ell}{\partial t_j} = \delta_{j,\ell},\quad \frac{\bar{\partial }t_\ell}{\partial \bar{t}_j}=0, \quad j,\ell \ge0.
\end{equation}
Here, $\tau$ denotes the tau-function of the dispersionless 2D Toda hierarchy, and $C_Q$ is the Robin constant appearing in the Euler--Lagrange relation~\eqref{eq_Euler lagrange}. (As mentioned earlier, we use a slight abuse of notation here: the symbol $\tau$ denotes the tau-function depending on the hierarchy time variables, and should not be confused with the anisotropy parameter appearing in~\eqref{potQ}. Since these objects arise in entirely different contexts, this notation should not lead to ambiguity.) 

In the present context, the tau-function under consideration should be understood as the ``large-$N$'' or dispersionless tau-function, which differs from the corresponding ``finite-$N$'' tau-function. A heuristic explanation, originating from~\cite{MWZ00} (see also~\cite{Be09}), is that for finite $N$, the tau-function $\tau_N$ is essentially the partition function $\ZZ_N$ up to multiplicative constants. This perspective played a crucial role in~\cite{BSY25}. Passing to the dispersionless limit $N\to\infty$, the leading contribution of $\log \tau_N = \log \ZZ_N$ gives rise to $\log \tau$ satisfying \eqref{eq_integrability}. 

In other words, in view of the leading-order term in the free energy expansion~\eqref{def of free energy expansion general}, the dispersionless tau-function associated with determinantal Coulomb gases is naturally identified with the weighted logarithmic energy $I_Q[\mu_Q]$ defined in~\eqref{def of log energy}. For completeness, we now verify this correspondence. 

\begin{lem} \label{Lem_tau function energy}
Let $K$ be a simply connected local droplet associated with an algebraic Hele-Shaw potential $Q$ of the form~\eqref{def of algebraic HeleShaw}, and let $t_j$ denote the exterior harmonic moments defined in~\eqref{eq_harmonic moment rel}. Assume that, under sufficiently small perturbations of the parameters $t_j$ in the potential $Q$, the corresponding local droplets remain simply connected and their boundaries $\partial K$ vary smoothly within the class of smooth Jordan curves. Then the tau-function of the dispersionless $2$D Toda hierarchy satisfying the integrability relations~\eqref{eq_integrability} is given by 
\begin{equation}\label{eq_tau function}
    \log \tau  = -\int_{K^2}\log \frac{1}{|z-\zeta|}\ud A(z)\ud A(\zeta) -t_0 \int_K Q(z)\ud A(z).
\end{equation}
In particular, if $K$ is the droplet of $Q$ then
\begin{equation}
    \log \tau = -t_0^2 I_Q[\mu_Q].
\end{equation}
\end{lem}

\begin{proof}
By the assumption that the droplets vary smoothly under sufficiently small perturbations of the exterior moments, there exists a point $z_0 \in \interior K$ that remains inside the local droplet throughout such perturbations. Moreover, the potential $Q$ can be written in the form 
\begin{equation}
    Q(z) = \frac{1}{t_0}\big(|z|^2 -H(z)\big) -\frac{1}{t_0}\big(|z_0|^2 -H(z_0)\big) + Q(z_0).
\end{equation}
Therefore, $Q(z_0)$ is invariant under perturbation by exterior harmonic moments $t_j$ for $j\ge 1$. Since $z_0 \in K$, it follows from ~\eqref{eq_Euler lagrange} that 
\begin{equation}\label{eq_Robins const}
    C_Q = \int_K \log \frac{1}{|z_0-\zeta|^2} \ud A(\zeta) + t_0 Q(z_0 ).
\end{equation}
We emphasise that the Robin constant $C_Q$ also varies under perturbations of the exterior moments. 

Now let $\log \tau$ be defined as in~\eqref{eq_tau function}. One then observes that 
\begin{align*}
    \log \tau 
    &= \int_{K^2}\log\frac{1}{ |z-\zeta|}\ud A(z)\ud A(\zeta) - \int_K \bigg(2\int_K \log \frac{1}{|z-\zeta|}\ud A(\zeta) + t_0Q(z)\bigg)\ud A(z)
    \\
    &=\int_{K^2}\log\frac{1}{ |z-\zeta|}\ud A(z)\ud A(\zeta) - t_0C_Q.
\end{align*}
Then, by \eqref{eq_Robins const}, we have 
\begin{align} \label{eq_tau function 2}
\log \tau= \log \widetilde{\tau} -t_0^2Q(z_0), \qquad \log \widetilde{\tau}:= - \int_{K^2}\log \Big|\frac{1}{z-z_0}-\frac{1}{\zeta-z_0}\Big|\ud A(z)\ud A(\zeta).  
\end{align}
Note that $\log \widetilde{\tau}$ coincides with the original tau-function associated with exterior conformal maps in the theory of the dispersionless $2$D Toda hierarchy~\cite{KKMWZ01, Te09, GTV14}. Since the quantity $t_0Q(z_0)$ remains invariant under perturbations of the moments $t_j$ for $j\ge 1$, the tau-function~\eqref{eq_tau function} satisfies all of the integrability relations~\eqref{eq_integrability}, except possibly the variation with respect to the area parameter $t_0$. To complete the proof, it therefore remains to verify directly that 
\begin{equation}
    \frac{\ud}{\ud t_0}\log \tau = -C_Q.
\end{equation}
Let $v_\n$ denote the outward normal velocity associated with the perturbation with respect to the area parameter $t_0$, as defined in~\eqref{eq_normal velocity}. This deformation corresponds precisely to the Hele-Shaw flow with source at infinity~\cite{LM16, HM13}. It follows that 
\begin{align*}
    \frac{\ud}{\ud t_0} \log \tau &= -2\int_{\partial K} \bigg(\int_K \log \frac{1}{|z-\zeta|}\ud A(\zeta)\bigg)v_\n(z)\ud s(z) -\int_{\partial K} t_0 Q(z) v_\n(z) \ud s(z)\\
    &= - \int_{\partial K} \bigg(\int_K \log \frac{1}{|z-\zeta|^2}\ud A(\zeta) + t_0 Q(z)\bigg)v_\n(z)\ud s(z)\\
    &= -C_Q \int_{\partial K} v_\n(z) \ud s(z) = -C_Q \frac{\ud t_0}{\ud t_0} = -C_Q.
\end{align*}
This completes the proof. 
\end{proof}

\subsection{Proof of Proposition~\ref{prop_sub_simply}} \label{Subsec_coefficient simply asymp}

We now specialise to the particular potential $Q$ defined in~\eqref{potQ}. With the preparatory results established above, we are ready to derive the refined asymptotic behaviour of the coefficient $\AA_{N,N}$ defined in~\eqref{def of AB}.

\begin{proof}[Proof of Proposition~\ref{prop_sub_simply}] Recall that the droplet $K$ is simply connected. We introduce the function  
\begin{equation}
    g(z) := \frac{\QQ(z)}{2} + \log \frac{\phi(z)}{\phi'(\infty)}, \qquad z\in \overline{K^c},
\end{equation}
where $\QQ$ is the function appearing in Theorem~\ref{thm_Hedenmalm Wennman}. Then it follows from \eqref{def of HW expansion} that  
\begin{equation} \label{eq_HW expansion in terms of g}
    P_{N,N}(z) =(\phi'(\infty))^{-1}F_0(z)e^{Ng(z)} \Big(1+\frac{F_1(z)}{F_0(z)}\frac{1}{N} + O(\frac{1}{N^2})\Big), 
\end{equation}
as $N \to \infty.$

Note that by \eqref{eq_F0}, we have 
\begin{equation} \label{asymp of F0 z to infty}
    F_0(z) =\sqrt{\phi'(\infty)(\phi'(z))} = \phi'(\infty)\Big(1+ O(\frac{1}{z^2})\Big), \qquad z\to \infty.
\end{equation}
To derive the asymptotic behaviour of $g(z)$ as $z\to\infty$, let $\mathsf{S}=\mathsf{S}_{K^c}$ denote the Schwarz function associated with the exterior domain $K^c$. Since $K$ is a simply connected droplet, Theorem~\ref{thm_local droplet} gives $r_{K^c}=\partial H$. Hence, by~\eqref{eq_Schwarz QD}, we obtain 
\begin{equation}
    \mathsf{S}(z) - \partial H(z) = \int_K\frac{\ud A(\zeta)}{z-\zeta}.
\end{equation}
Therefore, it follows from \eqref{eq_QQ algebraic Hele Shaw} that 
\begin{equation} \label{asymp of g z to infty}
    g(z) = \frac{1}{1-\tau^2}\int_K \log (z-\zeta)\ud A(\zeta)= \log z - \frac{1}{1-\tau^2}\sum_{j=1}^\infty \frac{M_j}{j}\frac{1}{z^j},
\end{equation}
where $M_j$ are interior harmonic moments~\eqref{eq_interior harmonic moment}.  

Then by combining \eqref{def of AB}, \eqref{eq_HW expansion in terms of g}, \eqref{asymp of F0 z to infty} and \eqref{asymp of g z to infty}, we obtain
\begin{equation} \label{eq_AAN asymp v0}
    \AA_{N,N} = -\frac{M_1}{1-\tau^2} N +  \big[ \frac{1}{z} \big] \Big(\frac{F_1}{F_0} \Big)(z)\frac{1}{N} + O(N^{-2}),
\end{equation}
where $[1/z]f(z)$ denotes the coefficient of $1/z$ of a function $f(z)$ as $ z \to \infty$. It remains to identify the coefficients appearing in~\eqref{eq_sub_simply} and~\eqref{eq_AAN asymp v0}, i.e. 
\begin{equation} \label{desired eq for ANN prop}
 \frac{1+\tau}{2}\frac{\ud}{\ud a}I_Q[\mu_Q] = \frac{M_1}{1-\tau^2}, \qquad  \frac{1+\tau}{48} \frac{\ud}{\ud a}\SS[\psi] =\big[\frac{1}{z}\big]\Big(\frac{F_1}{F_0}\Big)(z). 
\end{equation}

We begin with the coefficient of the linear term in $N$.
For this, it is convenient to consider a translation of external potential,
\begin{equation}
    \mathsf{Q}(z; t_1) = \frac{1}{1-\tau^2}\Big(|z|^2 -2\re \Big(t_1z + \frac{\tau }{2}z^2\Big)\Big)- 2c\log |z|, \qquad t_1=  -(1-\tau)a. 
\end{equation}
Note that
\begin{equation}
    Q(z+a) = \mathsf{Q}(z; t_1)+ \frac{a^2}{1+\tau}.
\end{equation}
Observe that the droplet $\mathsf{K}$ associated with the potential $\mathsf{Q}(z;t_1)$ is simply the translated domain $K-a$. Moreover, it follows directly from~\eqref{def of log energy} that
\begin{equation}
    I_\mathsf{Q}[\mu_\mathsf{Q}] = I_Q[\mu_Q]-\frac{a^2}{1+\tau}.
\end{equation}

Identifying $-t_0^2 I_{\mathsf Q}[\mu_{\mathsf Q}]$ with the tau-function of the dispersionless $2$D Toda hierarchy associated with $\mathsf{K}$, by Lemma~\ref{Lem_tau function energy}, the integrability relations~\eqref{eq_integrability} imply that 
\begin{equation}\label{eq_diff t1}
    \frac{\partial}{\partial t_1}I_\mathsf{Q}[\mu_\mathsf{Q}] = -\frac{1}{t_0^2}\int_{\mathsf{K}}z  \ud A(z).
\end{equation}
Since $t_0 =1-\tau^2$ in our present case (cf. \eqref{eq:eqmeasure}), it follows from ~\eqref{eq_diff t1} that 
\begin{align*}
    \frac{\ud}{\ud a}I_Q[\mu_Q] &=
    \frac{2a}{1+\tau}+\frac{\ud }{\ud a }I_\mathsf{Q}[\mu_\mathsf{Q}]
    =\frac{2a}{1+\tau}-2(1-\tau) \re \frac{\partial}{\partial t_1}\Big|_{t_1= -(1-\tau)a}I_\mathsf{Q}[\mu_\mathsf{Q}] \\
    &= \frac{2a}{1+\tau} +\frac{2(1-\tau)}{(1-\tau^2)^2} \re \int_{K-a}z \ud A(z) 
    = \frac{2(1-\tau)}{(1-\tau^2)^2} \re M_1.
\end{align*}
Moreover, since $M_1$ is real by symmetry with respect to the $x$-axis, we obtain the first identity in~\eqref{desired eq for ANN prop}.

Next, we prove the second identity in~\eqref{desired eq for ANN prop}. Let $\mathsf{K}(t_1)$ denote the droplet associated with the potential $\mathsf{Q}(\cdot;t_1)$, and let $\varphi_{t_1}$ be the corresponding exterior conformal map. Since $t_1=-(1-\tau)a$, we have 
\begin{equation}
    \psi(w) = \varphi_{t_1}(w)+a. 
\end{equation}
To establish the desired identity, we analyse the Löwner--Kufarev evolution of the domains $\mathsf{K}(t_1)$, namely the evolution for which the associated conformal map satisfies the Löwner--Kufarev equation~\eqref{eq_LK evolution}. 
We write the Löwner--Kufarev equation as
\begin{equation}
    \frac{\partial \varphi_{t_1}}{\partial \re t_1} = w\varphi_{t_1}'(w) P(t_1, w).
\end{equation}
Note that the second and third identities in the integrability relations~\eqref{eq_integrability} imply that 
\begin{align}
    -\frac{\partial}{\partial \re t_1}\int_{\mathsf{K}(t_1)^c} z^{-\ell}\ud A(z)&=\frac{1}{\pi}\int_{\partial \mathsf{K}(t_1)} z^{-\ell}v_\n(z) \ud s(w) =
    \begin{cases}
        0 &\textup{if } \ell=0, \;\text{or } \ell\ge2,  
        \smallskip 
        \\ 
        1 &\textup{if } \ell=1.
    \end{cases}
\end{align}
Moreover, by using~\eqref{eq_normal velocity} we have
\begin{equation*}
    \frac{1}{\pi}\int_{\T} w^{-\ell}|\varphi'(w)|^2 \re P(t_1, w) \ud s(w) = \frac{1}{\pi}\int_{\partial \mathsf{K}(t_1)} (\varphi^{-1}(z))^{-\ell} v_\n(z) \ud s(z) = 
    \begin{cases}
        0 &\textup{if } \ell = 0\;\text{or }\ell\ge 2,
        \smallskip 
        \\
        \ds \varphi'(\infty) &\textup{if } \ell=1.
    \end{cases}
\end{equation*}
Since $|\varphi'(w)|^2 \re P(t_1,w)$ is real and real-analytic on the unit circle $\T$, we obtain 
\begin{equation}
    |\varphi'(w)|^2 \re P(t_1,w) = \frac{\varphi'(\infty)}{2} (w+\frac{1}{w}) = \varphi'(\infty)\re w, \quad w \in \T.
\end{equation}
Note that $\SS[\psi]=\SS[\varphi]$, since the Liouville action is translation invariant by definition; see~\eqref{def of Liouville action}. Therefore, by using Proposition~\ref{prop_var Liouville}, we have 
\begin{align*}
\begin{split}
    \frac{\ud }{\ud a}\SS[\psi] &= -(1-\tau) \frac{\partial}{\partial \re t_1}\bigg|_{t_1= -(1-\tau)a} \SS[\varphi]\\
    &= \frac{2(1-\tau)\varphi'(\infty)}{\pi}  \int_\T \Big(\re \big(w^2\{\varphi, w\}\big)+|\varphi'(w)|^2 \kappa^2(\varphi(w))\Big) \frac{\re w}{|\varphi'(w)|^2} \ud s(w). 
\end{split}
\end{align*}
Finally,
\begin{equation*}
    \varphi'(w) = \psi'(w), \qquad \varphi'(\infty) = \psi'(\infty) = \frac{1}{\phi'(\infty)} 
\end{equation*}
and $\Delta Q = (1-\tau^2)^{-1}$ gives
\begin{equation}\label{eq_Liouville deform a}
    \frac{1+\tau}{48}\frac{\ud}{\ud a}\SS[\psi] = \frac{1}{24\pi (\Delta Q) \phi'(\infty)}\int_\T\Big(\re \big(w^2\{\psi, w\}\big)+|\psi'(w)|^2 \kappa^2(\psi(w))\Big) \frac{\re w}{|\psi'(w)|^2} \ud s(w).
\end{equation}
Hence the second identity of~\eqref{desired eq for ANN prop} follows from~\eqref{eq_F1 1/z} and~\eqref{eq_Liouville deform a} as $\AA_{N,N}$ is real due to the conjugate invariance $Q(z) = Q(\bar{z})$.
\end{proof}

\begin{rem}
The key idea of the proof is to interpret the deformation with respect to the parameter $a$ as a deformation with respect to the generalised time $t_1$ in the integrable hierarchy. This is achieved by introducing a translated determinantal Coulomb gas model, in the same spirit as in the proof of~\eqref{eq_deform a}. 

A crucial point is that the integrability relations provide an explicit expression for $\re P(t_1,w)$ in the Löwner--Kufarev equation with respect to $\re t_1$. Indeed, if one attempts instead to analyse directly the Löwner--Kufarev evolution of $\psi$ with respect to the parameter $a$, the resulting expressions become difficult to relate to the Liouville action, and the computation strongly depends on the explicit form of $\psi$. In contrast, the present argument is entirely independent of any explicit formula for $\psi$. This flexibility allows the argument to extend naturally to more general settings in which analogous integrable-hierarchy structures remain available. We plan to return to this direction elsewhere. 
\end{rem}

\section{Doubly connected case}\label{sec doubly}

In this section, we derive the asymptotic behaviour of the orthogonal polynomials $P_{n,N}$ in the doubly connected regime and thereby prove Proposition~\ref{Prop_OP coefficients doubly}. In contrast to the simply connected case, the doubly connected regime requires analysing deformations not only with respect to the parameter $a$, but also with respect to $\tau$; see Figures~\ref{fig_phases} and~\ref{fig_deformation}. 
Note that the deformation formula with respect to $\tau$ in~\eqref{eq_deform tau} involves the second subleading coefficient $\BB_{n,N}$. Consequently, one must compute both $F_1$ and $F_2$ in the asymptotic expansion \eqref{def of F} of $F$. 

Although $F_2$ can in principle be computed algorithmically using the iterative scheme developed in Section~\ref{sec RNM OP} and implemented in Theorem~\ref{thm_F0F1} for the computation of $F_1$, we instead present a simpler argument based on the following a priori geometric observation: throughout the entire doubly connected regime (namely Regime I in Theorem~\ref{Thm_droplet}), the outer boundary component of the droplet $K_Q$ remains fixed and is explicitly given by an ellipse; see~\eqref{def of KQ doubly}. This special geometric structure allows one to derive the free energy variations in a substantially simpler manner; see also~\cite{Rou25} for a recent related study.

For the convenience of the reader, we briefly outline the main idea of this section. In Subsection~\ref{sec RNM OP}, we observed that the algorithmic computation of the asymptotic expansion of $F$ depends only on the behaviour of the external potential in a neighbourhood of the \textit{outermost} boundary component of the droplet. Consequently, two families of planar orthogonal polynomials associated with different external potentials share the same fine asymptotic structure---in particular, the same $O(N^{-1})$ correction terms---provided that the corresponding droplets have identical outer boundaries and that the external potentials agree near this boundary component.

In the present setting, the planar orthogonal polynomials associated with the potential $Q$ in~\eqref{potQ} share the same fine asymptotic behaviour as the Hermite polynomials, which arise as the planar orthogonal polynomials associated with the elliptic Ginibre potential $Q^{\rm e}$ defined in~\eqref{def of Q induced elliptic}. Hence, the explicit coefficient formulas for the Hermite polynomials given in~\eqref{eq_Hermite} determine the fine asymptotic expansion of $P_{n,N}$. In fact, this argument yields not only the expansion up to order $O(N^{-2})$, but also asymptotic expansions to arbitrary precision; see Remark~\ref{Rem_full order expansion}.

In Subsection~\ref{Subsec_OP with multiply connected droplet}, we consider general algebraic Hele-Shaw potentials associated with multiply connected droplets. There, Proposition~\ref{prop_multiply connected} shows that the asymptotic behaviour of the corresponding planar orthogonal polynomials is essentially governed by the geometry of the outermost boundary component of the droplet. Intuitively, this means that the analysis can be reduced to the simply connected case corresponding to the outer boundary via a ``hole-filling'' procedure. 
In Subsection~\ref{Subsec_coefficient doubly asymp}, we specialise to the model~\eqref{potQ} and prove Proposition~\ref{Prop_OP coefficients doubly} using the structural observation provided by Proposition~\ref{prop_multiply connected}.

\subsection{Planar orthogonal polynomials with multiply connected droplets} \label{Subsec_OP with multiply connected droplet}


Throughout this subsection, we work under the following setting and assumptions.

Let $K$ be a multiply connected local droplet associated with an algebraic Hele-Shaw potential $W$ of the form~\eqref{def of algebraic HeleShaw}, having $d$ holes and total area $\pi t_0$. Then, by Theorem~\ref{thm_local droplet}, we have 
\begin{equation}
    K^c = \Omega_1 \cup \Omega_2 \cup \ldots \cup \Omega_d \cup \Omega_\infty,
\end{equation}
where $\Omega_1,\ldots,\Omega_d$ are bounded quadrature domains, while $\Omega_\infty$ is an unbounded quadrature domain.

We assume that all poles of the quadrature functions associated with the bounded domains $\Omega_j$ are simple. Since the quadrature functions $r_\ell$ corresponding to the components $\Omega_\ell$ sum to $\partial H$ by Theorem~\ref{thm_local droplet}, we may write
\begin{equation}
    \partial H (z) = r_\infty(z) + \sum_{j=1}^m \frac{c_jt_0}{z-a_j},
\end{equation}
where $r_\infty$ is a rational function whose poles lie in $\Omega_\infty$, and each point $a_j$ belongs to some bounded component $\Omega_\ell$ with $\ell\in\{1,\ldots,d\}$.

A typical class of external potentials in this setting is given by
\begin{equation}\label{eq_multiply Q}
    Q(z) = Q^0(z) - 2\sum_{j=1}^m c_j \log |z-a_j|, \qquad c_1,\ldots, c_m>0,
\end{equation}
where
\begin{equation}\label{eq_multiply Q0}
    Q^0(z) := \frac{1}{t_0} \Big(|z|^2- H^0(z)\Big), \qquad \partial H^0(z) = r_\infty(z).
\end{equation}
Clearly, our particular potential~\eqref{potQ} belongs to this class. 

In this general setting, we first establish the following lemma. 

\begin{lem}
Under the assumptions above, the domain $K^0 := \Omega_\infty^c$ is a simply connected local droplet associated with the potential $Q^0/(1+c)$, where $c=c_1+\ldots + c_m. $ 
\end{lem}

\begin{proof}
By virtue of Lemma~\ref{lem_variational equality}, it suffices to verify that
\begin{equation} \label{area of K0}
    \textup{area}(K^0)=\pi(1+c)t_0
\end{equation} 
and that
\begin{equation}\label{eq_Euler lagrange2}
    \int_{K^0} \log \frac{1}{|z-\zeta|^2}\ud A(\zeta)+t_0Q^0(z)= C, \quad z\in K^0,
\end{equation}
for some constant $C$.

We begin by computing the area of $K^0$. For each bounded quadrature domain $\Omega_\ell$, the quadrature identity~\eqref{eq_quadrature identity} yields 
\begin{equation}
   \frac{1}{\pi} \text{area }\Omega_\ell = \int_{\Omega_\ell} \ud A(z) = \frac{1}{2\pi i}\int_{\partial \Omega_\ell} r_\ell(z)\ud z.
\end{equation}
Note that by \eqref{LM sum of quad functions}, we have 
\begin{equation*}
    r_1(z) + \ldots + r_d(z) = \partial H(z) -r_\infty(z) = \sum_{j=1}^m \frac{c_jt_0}{z-a_j},
\end{equation*}
Therefore, we obtain 
\begin{align}
    \frac{1}{\pi}\textrm{area }K^0 = \frac{1}{\pi} \textrm{area }K + \sum_{\ell=1}^d\frac{1}{\pi} \text{area }\Omega_\ell
    = t_0 + \sum_{j=1}^m \Res_{z=a_j}(r_1+\ldots +r_d) = (1+c)t_0,
\end{align}
which gives rise to \eqref{area of K0}. 

Next, we show \eqref{eq_Euler lagrange2}. To prove this, it suffices to show that the derivative with respect to $z$ of the left-hand side of~\eqref{eq_Euler lagrange2} vanishes on $\interior K^0$. Equivalently, it is enough to verify that 
\begin{equation}
    \bar{z} - \partial H^0(z) = \int_{K^0}\frac{\ud A(\zeta)}{z-\zeta},
    \qquad z\in \interior K^0,
\end{equation}
By Green's formula, this is equivalent to 
\begin{equation} \label{H derivative Cauchy}
 - \partial H^0(z) =   \frac{1}{2\pi i}\int_{\partial K^0}\frac{\bar{\zeta} \ud \zeta}{z-\zeta}, \qquad z\in \interior K^0,
\end{equation}

By~\eqref{eq_Schwarz QD}, the Schwarz function $\mathsf{S}_{\Omega_\infty}$ satisfies 
\begin{equation}
    \mathsf{S}_{\Omega_\infty}(z) = r_\infty(z) + \int_{K^0}\frac{\ud A(\zeta)}{z-\zeta}, \quad z\in \bar{\Omega}_\infty.
\end{equation}
Thus $\mathsf{S}_{\Omega_\infty}- r_\infty$ is a bounded holomorphic function on $\Omega_\infty$. Consequently, \eqref{H derivative Cauchy} follows from
\begin{align*}
    \frac{1}{2\pi i}\int_{\partial K^0}\frac{\bar{\zeta} \ud \zeta}{z-\zeta} &= \frac{1}{2\pi i}\int_{\partial K^0} \frac{\mathsf{S}_{\Omega_\infty}(\zeta)}{z-\zeta}\ud \zeta = -r_{\infty}(z) + \frac{1}{2\pi i}\int_{\partial K^0}\frac{\mathsf{S}_{\Omega_\infty}(\zeta)- r_\infty(\zeta)}{z-\zeta}\ud \zeta \\
    &=-r_\infty(z)+ \lim_{r\to \infty} \frac{1}{2\pi  i}\int_{|\zeta|=r} \frac{\mathsf{S}_{\Omega_\infty}(\zeta)-r_\infty(\zeta)}{z-\zeta}\ud \zeta  = -r_\infty(z)= -\partial H^0(z).
\end{align*}
Here we used that $r_\infty$ is a rational function holomorphic on $(K^0)^c= \Omega_\infty$. 
This completes the proof of~\eqref{eq_Euler lagrange2}.
\end{proof}

We now relate the planar orthogonal polynomials associated with the potentials~\eqref{eq_multiply Q} and~\eqref{eq_multiply Q0}. To this end, let $U$ and $U^0$ be open neighbourhoods of $K$ and $K^0$, respectively, such that $K$ and $K^0$ are the droplets corresponding to the localised potentials $Q_{\overline U}$ and $Q_{\overline{U^0}}$. 

We denote by $P_{n,N}$ and $P_{n,N}^0$ the planar orthogonal polynomials associated with the measures
\begin{equation}
     e^{-NQ}\mathbbm{1}_{\bar{U}}\ud A \quad \text{and }\quad  e^{-NQ^0}\mathbbm{1}_{\bar{U}^0}\ud A,
\end{equation}
respectively. Under the assumption that the droplet is multiply connected, these two families of orthogonal polynomials are related by the following proposition.

\begin{prop}[\textbf{Hole-filling reduction for orthogonal polynomials}] \label{prop_multiply connected}
Assume that $c_1N,\ldots,c_mN$ are integers. Under the assumptions above, we have
\begin{equation}
    \Big(\prod_{j=1}^m(z-a_j)^{c_j N}\Big)P_{N,N}(z) = P^0_{(1+c)N,N}(z) (1 +O(N^{-\infty})), \quad N\to \infty,
\end{equation}
where $c=c_1+\ldots + c_m$ and the error term is uniform for all $z\in\C$ satisfying 
\begin{equation}\label{eq_multiply error range}
    \textup{dist}(z, (K^0)^c) \le (N^{-1}\log N)^{1/2}.
\end{equation}
\end{prop}

\begin{proof}
Let $\phi:(K^0)^c\to\D_e$ denote the inverse of the normalised exterior conformal map associated with $K^0$. Then, by Theorem~\ref{thm_Hedenmalm Wennman} and Remark~\ref{rem HW for M=tN}, as $N \to \infty$, we have 
\begin{align}
\begin{split} \label{HW multiply and outer}
    P_{N, N}(z) &= (\phi'(\infty))^{-N-1} (\phi(z))^{N} e^{N\QQ(z)/2} F(z),\\
    P_{(1+c)N, N}^0(z) &= (\phi'(\infty))^{-(1+c)N-1} (\phi(z))^{(1+c) N} e^{N\QQ^0(z)/2} F^0(z).
\end{split}
\end{align}
Here, $\QQ,F$ and $\QQ^0,F^0$ denote the corresponding functions associated with the potentials $Q$ and $Q^0$, respectively. 

We observe that $F$ and $F^0$ agree up to arbitrary precision as $N \to \infty$, in the sense that
\begin{equation} \label{F F0 asymp full}
    F(z) =F^0(z)(1+O(N^{-\infty})), \quad N \to \infty,
\end{equation}
uniformly for all $z\in\C$ satisfying~\eqref{eq_multiply error range}. By the discussion following Theorem~\ref{thm_Hedenmalm Wennman}, it suffices to verify that $R=R^0$ in a neighbourhood of $\Gamma=\partial\Omega_\infty$. This is immediate, since both functions are given by the same expression~\eqref{def of R Hele-Shaw}. On the other hand, it follows directly from~\eqref{eq_multiply Q} and~\eqref{eq_multiply Q0} that
\begin{equation}
    \QQ(z) = \QQ^0(z) -2 \sum_{j=1}^m c_j\log(z-a_j) + 2c \log \big(\phi(z)/\phi'(\infty)\big).
\end{equation} 
Therefore,
\begin{equation} \label{exp N QQ0 Q}
    e^{N(\QQ^0(z)- \QQ(z))/2}(\phi(z)/\phi'(\infty))^{cN} = \prod_{j=1}^m (z-a_j)^{c_jN}, \qquad N \to \infty.
\end{equation} 
Combining~\eqref{HW multiply and outer}, \eqref{F F0 asymp full}, and~\eqref{exp N QQ0 Q}, we obtain the proposition.
\end{proof}

Potentials of the form~\eqref{eq_multiply Q} correspond to determinantal Coulomb gas models with point-charge insertions into an algebraic Hele-Shaw potential. 
This also naturally arises in the products of characteristic polynomials; see e.g. \cite{AV03}. 
The asymptotic behaviour of the associated orthogonal polynomials has been studied extensively from the Riemann--Hilbert perspective; see, for instance,~\cite{BBLM15, BFKL26, KLY25}. In the context of multiple insertions, Ginibre ensembles with two point charges were recently analysed in~\cite{KKL25}. 

Proposition~\ref{prop_multiply connected} can in fact be extended beyond the setting of algebraic Hele-Shaw potentials. In several works, it has been shown that droplets associated with certain non-Hele-Shaw potentials likewise have complements given by unions of disjoint quadrature domains; see, for example,~\cite{BFL25, BFKL26, BDSW18, CK22}, where spherical ensembles with point-charge insertions were studied. In many such situations, the multiply connected regime leads to rather complicated motherbody structures, which substantially complicates the corresponding Riemann--Hilbert analysis. 

On the other hand, Proposition~\ref{prop_multiply connected} suggests that one can still extract the asymptotic behaviour of the orthogonal polynomials, at least near infinity, by reducing the problem to the outer boundary component. For instance, this perspective applies to the post-critical regime studied in~\cite{BFKL26}. We emphasise, however, that in such settings the scaling factor $1+c$ must be modified, since the equilibrium measure typically has non-uniform density.


\subsection{Proof of Proposition~\ref{Prop_OP coefficients doubly}} \label{Subsec_coefficient doubly asymp}

Returning to our model, recall that the external potential is given by~\eqref{potQ}. By~\eqref{Q in terms of Qe and log}, the potential $Q$ is of the form~\eqref{eq_multiply Q} with $Q^0=Q^{\rm e}$, where $Q^{\rm e}$ denotes the elliptic Ginibre potential defined in~\eqref{def of Q induced elliptic}. Hence, Proposition~\ref{prop_multiply connected} provides a direct relation between the orthogonal polynomials $P_{n,N}$ associated with $Q$ and the planar Hermite polynomials $P_{n,N}^{\rm e}$ defined in~\eqref{def of planar Hermite}, corresponding to the potential $Q^{\rm e}$. Based on this observation, we now prove Proposition~\ref{Prop_OP coefficients doubly}.

\begin{proof}[Proof of Proposition~\ref{Prop_OP coefficients doubly}]
Note that if $cN$ were an integer, Proposition~\ref{Prop_OP coefficients doubly} would follow immediately from Proposition~\ref{prop_multiply connected} together with the explicit formula~\eqref{eq_Hermite} for the Hermite polynomials. However, since $cN$ is not necessarily an integer in our setting, we slightly modify the proof of Proposition~\ref{prop_multiply connected}.

Recall from~\eqref{def of KQ doubly} that, in the doubly connected regime, the droplet is given by $K_Q=\mathsf{E}\sm \mathsf{D}$. Accordingly, the inverse of the normalised exterior conformal map associated with $K^0=\mathsf{E}$ is given by the inverse Joukowsky transform
\[
    \phi_{1+c}(z) := \frac{z+ \sqrt{z^2-4\tau(1+c)}}{2\sqrt{1+c}}, \quad z \in \mathsf{E}^c,
\]
where we have made the dependence on the parameter $c$ explicit. 
We denote by $F^{\rm e}$ and $\QQ^{\rm e}$ the functions appearing in Theorem~\ref{thm_Hedenmalm Wennman} associated with the elliptic potential $Q^{\rm e}$. Note that 
\begin{equation} \label{QQ e Q e relation}
  \QQ(z) = \QQ^{\rm e}(z) -2c \log (z-a) + 2c \log \frac{\phi_{1+c}(z)}{\phi_{1+c}'(\infty)}. 
\end{equation}

Due to the explicit form of $P^{\rm e}_{n,N}$ in~\eqref{eq_Hermite}, the asymptotic expansion~\eqref{def of HW inhomogeneous} yields the following whenever $t>0$ is rational and $tN$ is an integer: 
\begin{align*}
    (\phi_t'(\infty))^{-tN-1}(\phi_t(z))^{tN} e^{N\QQ^{\rm e}(z)/2} F^{\rm e}(z;t) & =z^{tN} - \frac{\tau t(tN-1)}{2}z^{tN-2} + O(z^{tN-4}), \quad z\to\infty.
\end{align*}
Since $F^{\rm e}(z;t)$ depends continuously on the parameter $t$ (cf. Remark~\ref{rem HW for M=tN}), the asymptotic expansion
\begin{equation}
    F^{\rm e}(z;t) =  \phi_t'(\infty)\exp\Big(-N\Big(\frac{\QQ^{\rm e}(z)}{2}+ t\log \frac{\phi_t(z)}{\phi'_t(\infty) z}\Big)\Big) \Big(1- \frac{\tau t (tN-1)}{2} z^{-2} + O(z^{-4})\Big), \quad z\to \infty,
\end{equation}
holds not only for rational values of $t$, but for all real $t>0$. 
Then as in the proof of Proposition~\ref{prop_multiply connected}, we obtain
\begin{align*}
  P_{N,N}(z)&=  (\phi_{1+c}'(\infty))^{-N-1}(\phi_{1+c}(z))^N e^{N\QQ(z)/2} F(z)\\
    &= (\phi_{1+c}'(\infty))^{-N-1}(\phi_{1+c}(z))^N e^{N\QQ(z)/2} F^{\rm e}(z; 1+c) (1+O(N^{-\infty})). 
\end{align*}
Combining this with \eqref{QQ e Q e relation}, it follows that 
\begin{align}
\begin{split}\label{eq double connected asymptotic}
    P_{N,N}(z)   &=\exp\Big( \frac{N}{2} \Big(\QQ(z)- \QQ^{\rm e}(z)-2c \log \frac{\phi_{1+c}(z)}{\phi_{1+c}'(\infty)}-2(1+c)\log z\Big)\Big) \\
    &\quad \times \Big(1- \frac{\tau (1+c)((1+c)N-1)}{2} z^{-2} + O(z^{-4})\Big)(1+O(N^{-\infty}))\\
    &=\big(1-\frac{a}{z}\big)^{-cN}\Big(z^N- \Big(\frac{\tau(1+c)^2}{2}N -\frac{\tau(1+c)}{2}\Big) z^{N-2} + O(z^{N-4})\Big)(1+O(N^{-\infty})),
\end{split}
\end{align}
as $z\to\infty$. The asymptotic behaviours~\eqref{eq_sub_doubly0} and~\eqref{eq_sub_doubly1} now follow immediately from~\eqref{eq double connected asymptotic}.

To derive the asymptotic behaviour of $\AA_{N+1,N}$ and $\BB_{N+1,N}$, we make use of the following scaling property of the planar orthogonal polynomials:
\begin{equation}\label{eq_P scaling}
    P_{n,N}(z;a,c,\tau) = \big(\frac{n}{N}\big)^{n/2} P_{n,n}\big(\sqrt{\frac{N}{n}} z; \sqrt{\frac{N}{n}}a, \frac{N}{n}c, \tau\big).
\end{equation}
This identity follows directly from the orthogonality relation~\eqref{def of planar OP} together with the potential~\eqref{potQ} by a simple change of variables; see e.g. \cite{LY17}. 
Setting $n=N+1$ in~\eqref{eq_P scaling}, we have
\begin{align*}
\AA_{N+1, N}(a,c,\tau) &= \sqrt{\frac{N+1}{N}}\AA_{N+1, N+1}\big(\sqrt{\frac{N}{N+1}}a, \frac{N}{N+1}c, \tau\big), \\  \BB_{N+1, N}(a,c,\tau) &= \frac{N+1}{N}\BB_{N+1, N+1}\big(\sqrt{\frac{N}{N+1}}a, \frac{N}{N+1}c, \tau\big). 
\end{align*}
Now the asymptotic behaviours~\eqref{eq_sub_doubly2} and~\eqref{eq_sub_doubly3} follow from~\eqref{eq_sub_doubly0} and~\eqref{eq_sub_doubly1} by straightforward computations. This completes the proof.
\end{proof}

\appendix

\section{Explicit description of the droplet} \label{Appendix_exact description}

Throughout the proof of the main results, we do not rely on the explicit shape of the droplet associated with the potential $Q$ in~\eqref{potQ}, nor on explicit formulas for the associated logarithmic energy. Nevertheless, for completeness, we record these explicit expressions here, which were analysed in detail in the recent work~\cite{BY25}.

Recall that the potential~\eqref{potQ} depends on the three real parameters $\tau\in[0,1)$, $c\ge0$, and $a\ge0$. As discussed in Theorem~\ref{Thm_droplet}, the associated droplet exhibits three distinct topological phases, and the corresponding parameter regions are given as follows; see Figure~\ref{fig_phases}.  
  
\begin{itemize}
    \item (Regime I) The droplet is doubly connected if and only if the parameters $(a,c,\tau)$ satisfy
    \begin{equation}\label{eq_doubly connected range 1}
    a \le \min \Big \{  2\sqrt{\frac{2\tau(1+\tau)}{3+\tau^2}},2 \sqrt{ \frac{     \tau(1-\tau-2c\tau)   }{  1-\tau  } } \Big \}  \quad \text{and} \quad  0\leq c\leq \frac{1-\tau}{2\tau},
    \end{equation}
    or 
    \begin{equation}\label{eq_doubly connected range 2}
  2\sqrt{\frac{2\tau(1+\tau)}{3+\tau^2}} \le   a \le (1+\tau)\sqrt{1+c}-\sqrt{c(1-\tau^2)} 
  \quad \text{and} \quad 0\leq c\leq \frac{(1-\tau)^3}{2\tau(3+\tau^2)}. 
    \end{equation}
  \item (Regime II) The droplet is simply connected if and only if the parameters $(a,c,\tau)$ satisfy the following conditions. For a fixed $\tau$, the parameters $c$ and $a$ can be parametrised by two auxiliary parameters $q$ and $\lambda$ according to
\begin{align}
c& \equiv c(q,\lambda) =   \frac{\lambda}{q^2} \frac{(1-q^2)^2 (1-\tau q^2) +  q^2\lambda  }{(1-q^2)^2(1-\tau^2 + 2\tau \lambda ) - \lambda^2 }, \label{eq_simply connected c}
\\
a & \equiv a(q,\lambda) =  \sqrt{ \frac{1+\tau}{1-\tau} }
 \frac{  (1-\tau)(1-q^2)(  1+\tau q^2 ) -( 1-\tau q^2) \lambda   }{q\sqrt{ (1-q^2)^2 (1-\tau^2 + 2\tau \lambda) - \lambda^2 } }.\label{eq_simply connected p}
\end{align}
Here, the parameters $q$ and $\lambda$ lie in the range
\begin{equation}
  q\in(0,1), \qquad \lambda \in [0, \lambda_{ \rm cri}), 
\end{equation}
where $\lambda_{\rm cri}$ denotes the unique zero of the function $H(q,\cdot)$ given by 
\begin{align}\label{def of H(a,kappa)}
\begin{split} 
    H(q,\lambda) &:= \frac{1-\tau}{q} \Big(1+\tau q^2 -\frac{1-\tau q^2}{1-\tau}\frac{\lambda}{1-q^2} \Big)\Big(w_*-\frac{1}{w_*}\Big)
    \\
    &\quad - 2\Big(\frac{1-\tau q^2}{q^2}\lambda + \frac{\lambda^2}{(1-q^2)^2}\Big)\log\frac{|qw_*-1|}{|w_*-q|} - 2\Big(1-\tau^2 +\frac{1+\tau q^2}{q^2}\lambda\Big)\log |w_*|,
\end{split}
\end{align}
where $w_*$ is given by
\begin{equation} \label{def of zstar}
w_* := \frac1{2q} \bigg( q^2+ 1+\frac{\lambda}{1-\tau} + \sqrt{ \Big( q^2+ 1+\frac{\lambda}{1-\tau} \Big)^2-4q^2 }\,\bigg). 
\end{equation}
\smallskip 
\item (Regime III) The droplet consists of two simply connected components whenever the parameters $(a,c,\tau)$ lie outside both Regime I and Regime II.
\end{itemize}

We also note that, in Theorem~\ref{Thm_droplet}(ii), the parameters $R$, $\lambda$, and $q$ are characterised by the following coupled equations:
\begin{equation}
\begin{cases} \label{def of coupled equations}
\displaystyle 1 = \frac{R^2 }{1-\tau^2} \Big( 1-\tau^2 +  2\tau\lambda -\frac{\lambda^2}{ (1-q^2)^2 } \Big), 
\smallskip 
\\
\displaystyle   c = \frac{R^2 \lambda }{1-\tau^2}\Big(\frac{1-\tau q^2}{q^2} + \frac{\lambda}{(1-q^2)^2} \Big) ,
\smallskip 
\\
\displaystyle  a  = \frac{R}{q}\Big(1+\tau q^2-\frac{1-\tau q^2}{1-\tau}\frac{\lambda}{1-q^2}\Big).
\end{cases}
\end{equation} 
Finally, we record the explicit formulas for the logarithmic energy, which are given as follows.

\begin{prop}[\textbf{Logarithmic energy of the equilibrium measure; cf. Theorem 1.2 in \cite{BY25}}] \label{Prop_energy explicit} We have the following. 
\begin{itemize}
\item[\textup{(i)}] Suppose that $(a,c,\tau)$ falls within \textup{Regime I}. Then 
    \begin{equation} \label{weighted energy doubly connected}
    I_Q[\mu_Q]=\frac{3}{4}+\frac{3c}{2} + \frac{c^2}{2}\log \Big( c(1-\tau^2) \Big) -\frac{(1+c)^2}{2}\log(1+c) - \frac{c a^2}{1+\tau}.
    \end{equation}  
\item[\textup{(ii)}] Suppose that $(a,c,\tau)$ falls within \textup{Regime II}. Then   
    \begin{align}
    \begin{split} \label{weighted energy simply connected}
     I_Q[\mu_Q] &= \frac{3}{4} + \frac{3c}{2} -\frac{ca^2}{1+\tau} +\frac{R^3\lambda a (2-3q^2-3\tau q^2 + 2\tau q^4) }{2(1-\tau^2)^2q^3} \Big( 1-\tau  -\frac{ 2-3q^2+3\tau q^2-2\tau q^4}{  2-3q^2-3\tau q^2 + 2\tau q^4 } \frac{\lambda}{1-q^2}\Big)
    \\
    &\quad +2c(1+c)\log q + c^2 \log \Big( \frac{ c (1-\tau^2)(1-q^2) }{ R \lambda } \Big) - (1+c)^2 \log R. 
    \end{split}
    \end{align}
\end{itemize}
\end{prop}



\begin{thebibliography}{999}


\bibitem{AS76} D. Aharonov and H.~S. Shapiro, \emph{Domains on which analytic functions satisfy quadrature identities}, J. Anal. Math. \textbf{30} (1976), 39--73.

\bibitem{ABK21} G. Akemann, S.-S. Byun and N.-G. Kang, \emph{A non-Hermitian generalisation of the Marchenko–Pastur distribution: from the circular law to multi-criticality}, Ann. Henri Poincaré \textbf{22} (2021), 1035--1068.

\bibitem{ACV18} G. Akemann, M. Cikovic and M. Venker, \emph{Universality at weak and strong non-Hermiticity beyond the elliptic Ginibre ensemble}, Comm. Math. Phys. \textbf{362} (2018), 1111--1141.

\bibitem{AV03} G. Akemann and G. Vernizzi, \emph{Characteristic polynomials of complex random matrix models}, Nuclear Phys. B \textbf{660} (2003), 532--556.

\bibitem{AFLS25} M. Allard, P. J. Forrester, S. Lahiry and B.-J. Shen, \emph{Partition function of 2D Coulomb gases with radially symmetric potentials and a hard wall}, arXiv:2506.14738.

\bibitem{AL25} M. Allard and S. Lahiry, \emph{Birth of a gap: Critical phenomena in 2D Coulomb gas}, arXiv:2509.24529.


\bibitem{Am25} Y. Ameur, \emph{A formula for the edge density $\sqrt{n}$-correction for two-dimensional Coulomb systems}, arXiv:2510.16945.

\bibitem{ACC24} Y. Ameur, C. Charlier and J. Cronvall, \emph{Random normal matrices: Eigenvalue correlations near a hard wall}, J. Stat. Phys. \textbf{191} (2024), 98.

\bibitem{ACC26} Y. Ameur, C. Charlier and J. Cronvall, \emph{Free energy and fluctuations in the random normal matrix model with spectral gaps}, Constr. Approx. \textbf{63} (2026), 279--335.

\bibitem{ACCL23} Y. Ameur, C. Charlier, J. Cronvall and J. Lenells, \emph{Exponential moments for disc counting statistics at the hard edge of random normal matrices}, J. Spectr. Theory \textbf{13} (2023), 841--902.

\bibitem{ACCL24} Y. Ameur, C. Charlier, J. Cronvall and J. Lenells, \emph{Disc counting statistics near hard edges of random normal matrices: the multi-component regime}, Adv. Math. \textbf{441} (2024), 109549.

\bibitem{AC23} Y. Ameur and J. Cronvall, \emph{Szeg\H{o} type asymptotics for the reproducing kernel in spaces of full-plane weighted polynomials}, Comm. Math. Phys. \textbf{398} (2023), 1291--1348.

\bibitem{AC26} Y. Ameur and J. Cronvall, \emph{On fluctuations of Coulomb systems and universality of the Heine distribution}, J. Funct. Anal. \textbf{290} (2026), 111301.

\bibitem{AHM11} Y. Ameur, H. Hedenmalm and N. Makarov, \emph{Fluctuations of eigenvalues of random normal matrices}, Duke Math. J. \textbf{159} (2011), 31--81.

\bibitem{AS21} S. Armstrong and S. Serfaty, \emph{Local laws and rigidity for Coulomb gases at any temperature}, Ann. Probab. \textbf{49} (2021), 46--121.

\bibitem{BBLM15} F. Balogh, M. Bertola, S.-Y. Lee and K. D. T.-R. McLaughlin, \emph{Strong asymptotics of the orthogonal polynomials with respect to a measure supported on the plane}, Comm. Pure Appl. Math. \textbf{68} (2015), 112--172.


\bibitem{BGM17} F. Balogh, T. Grava and D. Merzi, \emph{Orthogonal polynomials for a class of measures with discrete rotational symmetries in the complex plane}, Constr. Approx. \textbf{46} (2017), 109--169.


\bibitem{BBNY19} R. Bauerschmidt, P. Bourgade, M. Nikula and H.-T. Yau, \emph{The two-dimensional Coulomb plasma: quasi-free approximation and central limit theorem}, Adv. Theor. Math. Phys. \textbf{23} (2019), 841--1002.



\bibitem{Be09} M. Bertola, \emph{Moment determinants as isomonodromic tau functions}, Nonlinearity \textbf{22} (2009), 29--50.

\bibitem{BEG18} M. Bertola, J. G. Elias Rebelo and T. Grava, \emph{Painlevé IV critical asymptotics for orthogonal polynomials in the complex plane}, SIGMA Symmetry Integrability Geom. Methods Appl. \textbf{14} (2018), Paper No. 091, 34pp.

\bibitem{BL08} M. Bertola and S.-Y. Lee, \emph{First colonization of a spectral outpost in random matrix theory}, Constr. Approx. \textbf{30} (2008), 225--263.

\bibitem{BK12} P. M. Bleher and A. B. J. Kuijlaars, \emph{Orthogonal polynomials in the normal matrix model with a cubic potential}, Adv. Math. \textbf{230} (2012), 1272--1321.

\bibitem{BS20} P. M. Bleher and G. L. F. Silva, \emph{The mother body phase transition in the normal matrix model}, Mem. Amer. Math. Soc. \textbf{265} (2020), 1289, v+144 pp.

\bibitem{BG24} G. Borot and A. Guionnet, \emph{Asymptotic expansion of beta matrix models in the multi-cut regime}, Forum Math. Sigma \textbf{12} (2024), 1--93.

\bibitem{BDHK25} P. Bourgade, G. Dubach, L. Hartung and A. Keles, \emph{Fisher-Hartwig asymptotics for non-Hermitian random matrices}, arXiv:2512.09123.

\bibitem{Bo25} L. Bourgoin, \emph{Free energy of the Coulomb gas in the determinantal case on Riemann surfaces}, arXiv:2508.20598.

\bibitem{BDSW18} J. S. Brauchart, P. D. Dragnev, E. B. Saff and R. S.~Womersley, \emph{Logarithmic and Riesz equilibrium for multiple sources on the sphere: the exceptional case}, Contemporary Computational Mathematics - A Celebration of the 80th Birthday of Ian Sloan (J. Dick, F. Kuo and H. Wo\'zniakowski, eds.), Springer, Cham, 2018, 179--203.

\bibitem{By24} S.-S. Byun, \emph{Planar equilibrium measure problem in the quadratic fields with a point charge}, Comput. Methods Funct. Theory \textbf{24} (2024), 303--332.

\bibitem{By25a} S.-S. Byun, \emph{Anomalous free energy expansions of planar Coulomb gases: multi-component and conformal singularity}, arXiv:2508.00316.

\bibitem{BC25} S.-S. Byun and C. Charlier, \emph{On the characteristic polynomial of the eigenvalue moduli of random normal matrices}, Constr. Approx. \textbf{62} (2025), 471--521.

\bibitem{BCMS25} S.-S. Byun, C. Charlier, P. Moreillon and N. Simm, \emph{Precise large deviations in geometric last passage percolation}, arXiv:2510.17470.

\bibitem{BF25} S.-S.~Byun and P. J.~Forrester, \emph{Progress on the study of the Ginibre ensembles}, KIAS Springer Ser. Math. \textbf{3} Springer, 2025, 221pp.

\bibitem{BF25a} S.-S.~Byun and P. J.~Forrester, \emph{Electrostatic computations for statistical mechanics and random matrix applications}, MATRIX Book Ser. (to appear), arXiv:2510.14334.

\bibitem{BFKL26} S.-S. Byun, P. J. Forrester, A. B. J. Kuijlaars and S. Lahiry, \emph{Orthogonal polynomials in the spherical ensemble with two insertions}, SIAM J. Math. Anal. (to appear), arXiv:2503.15732.

\bibitem{BFL25} S.-S. Byun, P. J. Forrester and S. Lahiry, \emph{Properties of the one-component Coulomb gas on a sphere with two macroscopic external charges}, Pure Appl. Funct. Anal. (to appear), arXiv:2501.05061.

\bibitem{BKS23} S.-S. Byun, N.-G. Kang and S.-M. Seo, \emph{Partition functions of determinantal and Pfaffian Coulomb gases with radially symmetric potentials}, Comm. Math. Phys. \textbf{401} (2023), 1627--1663.

\bibitem{BKSY25} S.-S. Byun, N.-G. Kang, S.-M. Seo and M. Yang, \emph{Free energy of spherical Coulomb gases with point charges}, J. Lond. Math. Soc. (2) \textbf{112} (2025), e70294.


\bibitem{BP26} S.-S. Byun and S. Park, \emph{Large gap probabilities of complex and symplectic spherical ensembles with point charges}, J. Funct. Anal. \textbf{290} (2026), 111260.

\bibitem{BSY25} S.-S. Byun, S.-M. Seo and M. Yang, \emph{Free energy expansions of a conditional GinUE and large deviations of the smallest eigenvalue of the LUE}, Comm. Pure Appl. Math. \textbf{78} (2025), 2247--2304.

\bibitem{BY23} S.-S. Byun and M. Yang, \emph{Determinantal Coulomb gas ensembles with a class of discrete rotational symmetric potentials}, SIAM J. Math. Anal. \textbf{55} (2023), 6867--6897.

\bibitem{BY25} S.-S. Byun and E. Yoo, \emph{Three topological phases of the elliptic Ginibre ensembles with a point charge}, arXiv:2502.02948.

\bibitem{CCG21} M. Cafasso, T. Claeys and M. Girotti, \emph{Fredholm determinant solutions of the Painlevé II hierarchy and gap probabilities of determinantal point processes}, Int. Math. Res. Not. \textbf{2021} (2021), 2437--2478.

\bibitem{CFTW15} T. Can, P. J. Forrester, G. Téllez and P. Wiegmann, \emph{Exact and asymptotic features of the edge density profile for the one component plasma in two dimensions}, J. Stat. Phys. \textbf{158} (2015), 1147--1180.

\bibitem{Ch22} C. Charlier, \emph{Asymptotics of determinants with a rotation-invariant weight and discontinuities along circles}, Adv. Math. \textbf{408} (2022), 108600.

\bibitem{Ch23} C. Charlier, \emph{Large gap asymptotics on annuli in the random normal matrix model}, Math. Ann. \textbf{388} (2024), 3529--3587.


\bibitem{Ch25} C. Charlier, \emph{Smallest gaps of the two-dimensional Coulomb gas}, arXiv:2507.23502.

\bibitem{CFWW25} C. Charlier, B. Fahs, C. Webb and M. D. Wong, \emph{Asymptotics of Hankel determinants with a multi-cut regular potential and Fisher-Hartwig singularities}, Mem. Amer. Math. Soc. \textbf{310} (2025), 1567, v+138pp.

\bibitem{CG21} C. Charlier and R. Gharakhloo, \emph{Asymptotics of Hankel determinants with a Laguerre-type or Jacobi-type potential and Fisher–Hartwig singularities} Adv. Math. \textbf{383} (2021), 107672.

\bibitem{CGM15} T. Claeys, T. Grava and K. D. T.-R. McLaughlin, \emph{Asymptotics for the partition function in two-cut random matrix models}, Comm. Math. Phys. \textbf{339} (2015), 513--587.

\bibitem{CK15} T. Claeys and I. Krasovsky, \emph{Toeplitz determinants with merging singularities}, Duke Math. J. \textbf{164} (2015), 2897--2987.

\bibitem{CJ23} K. Courteaut and K. Johansson, \emph{Partition function for the 2d Coulomb gas on a Jordan curve}, Ann. Fenn. Math. \textbf{50} (2025), 109–-144.


\bibitem{CK22} J. G. Criado del Rey and A. B. J. Kuijlaars, \emph{A vector equilibrium problem for symmetrically located point charges on a sphere}, Constr. Approx. \textbf{55} (2022), 775--827.

\bibitem{CW25} J. Cronvall and A. Wennman, \emph{A direct approach to soft and hard edge universality for random normal matrices}, arXiv:2511.18628.

\bibitem{DMMS25} A. Deaño, K. D. T.-R. McLaughlin, L. Molag and N. Simm, \emph{Asymptotics for a class of planar orthogonal polynomials and truncated unitary matrices}, arXiv:2505.12633.

\bibitem{DS22} A. Deaño and N. Simm, \emph{Characteristic polynomials of complex random matrices and Painlevé transcendents}, Int. Math. Res. Not. \textbf{2022} (2022), 210--264.


\bibitem{DIZ97} P. Deift, A. R. Its and X. Zhou, \emph{A Riemann-Hilbert approach to asymptotic problems arising in the theory of random matrix models, and also in the theory of integrable statistical mechanics}, Ann. of Math. (2) \textbf{146} (1997), 149--235.



\bibitem{Dub09} J. Dubédat, \emph{SLE and the free field: partition functions and couplings}, J. Amer. Math. Soc. \textbf{22} (2009), 995--1054.


\bibitem{EF05} P. Elbau and G. Felder, \emph{Density of eigenvalues of random normal matrices}, Comm. Math. Phys. \textbf{259} (2005), 433--450.

\bibitem{Fo10} P. J. Forrester, \emph{Log-gases and random matrices}, Princeton University Press, Princeton, 2010.

\bibitem{Fo25} P. J. Forrester, \emph{Dualities in random matrix theory}, arXiv:2501.07144.


\bibitem{FW08} P. J. Forrester and S. O. Warnaar, \emph{The importance of the Selberg integral}, Bull. Amer. Math. Soc. (N.S.) \textbf{45} (2008), 489--534.


\bibitem{Fyo18} Y. V. Fyodorov, \emph{On statistics of bi-orthogonal eigenvectors in real and complex Ginibre ensembles: combining partial Schur decomposition with supersymmetry}, Comm. Math. Phys. \textbf{363} (2018), 579--603.



\bibitem{GS05} B. Gustafsson and H. S. Shapiro, \emph{What is a Quadrature Domain?}, Quadrature Domains and Their Applications (P. Ebenfelt, B. Gustafsson, D. Khavinson and M. Putinar, eds.), Operator Theory: Advances and Applications. \textbf{156} Birkhäuser Basel, 2005, 1--25.

\bibitem{GTV14} B. Gustafsson, R. Teodorescu. and A. Y. Vasil’ev, \emph{Classical and Stochastic Laplacian Growth}, Advances in Mathematical Fluid Mechanics, Birkh\"auser/Springer, Cham, 2014.

\bibitem{GV06} B. Gustafsson and A. Y. Vasil'ev, \emph{Conformal and potential analysis in Hele-Shaw cells}, Advances in Mathematical Fluid Mechanics, Birkh\"auser, Basel, 2006.
 

\bibitem{Hed24} H. Hedenmalm, \emph{Soft Riemann-Hilbert problems and planar orthogonal polynomials}, Comm. Pure Appl. Math. \textbf{77} (2024), 2413--2451.

\bibitem{HM13} H. Hedenmalm and N. Makarov, \emph{Coulomb gas ensembles and Laplacian growth}, Proc. Lond. Math. Soc. (3) \textbf{106} (2013), 859--907.


\bibitem{HW21} H. Hedenmalm and A. Wennman, \emph{Planar orthogonal polynomials and boundary universality in the random normal matrix model}, Acta Math. \textbf{227} (2021), 309--406.

\bibitem{HW24} H. Hedenmalm and A. Wennman, \emph{Berezin density and planar orthogonal polynomials}, Trans. Amer. Math. Soc. \textbf{377} (2024), 4825--4863.

\bibitem{HW24a} H. Hedenmalm and A. Wennman, \emph{A global asymptotic expansion of the polynomial Bergman density}, preprint.

\bibitem{JMP94} B. Jancovici, G. Manificat and C. Pisani, \emph{Coulomb systems seen as critical systems: finite-size effects in two dimensions}, J. Stat. Phys. \textbf{76} (1994), 307--329.

\bibitem{Joh98} K. Johansson, \emph{On fluctuations of eigenvalues of random Hermitian matrices}, Duke Math. J. \textbf{91} (1998), 151--204.

\bibitem{Joh00} K. Johansson, \emph{Shape fluctuations and random matrices}, Comm. Math. Phys. \textbf{209} (2000), 437--476.

\bibitem{Joh22} K. Johansson, \emph{Strong Szeg\H{o} theorem on a Jordan curve}. Toeplitz Operators and Random Matrices in Memory of Harold Widom (Basor, et al., eds.), Operator Theory Advances and Applications, Birkhäuser, Basel, 2022.

\bibitem{JV26} K. Johansson and F. Viklund, \emph{Coulomb gas and the Grunsky operator on a Jordan domain with corners}, Invent. Math. (Online), arXiv:2309.00308.
 

\bibitem{KM13} N.-G. Kang and N. Makarov, \emph{Gaussian free field and conformal field theory}, Astérisque \textbf{353} (2013), viii+136.
 

\bibitem{KKL25} M. Kieburg, A. B. J. Kuijlaars and S. Lahiry, \emph{Orthogonal polynomials in the normal matrix model with two insertions}, Nonlinearity \textbf{38} (2025), no. 6, Paper No. 065013, 66pp.
 

\bibitem{KMMW17} S. Klevtsov, X. Ma, G. Marinescu and P. Wiegmann, \emph{Quantum Hall effect and Quillen metric}, Comm. Math. Phys. \textbf{349} (2017), 815--855.

\bibitem{KKMWZ01} I.~K. Kostov, I. Krichever, M. Mineev-Weinstein, P. Wiegmann, A. Zabrodin, \emph{The $\tau$-function for analytic curves}, Random matrix models and their applications, Math. Sci. Res. Inst. Publ. \textbf{40}, Cambridge Univ. Press, 2001, 285--299.

\bibitem{Kra07} I. Krasovsky, \emph{Correlations of the characteristic polynomials in the Gaussian unitary ensemble or a singular Hankel determinant}, Duke Math. J. \textbf{139} (2007), 581--619.

\bibitem{KLY25} T. Krüger, S.-Y. Lee and M. Yang, \emph{Local statistics in normal matrix models with merging singularity}, Comm. Math. Phys. \textbf{406} (2025), 130.
 

\bibitem{Lam20} G. Lambert, \emph{Maximum of the characteristic polynomial of the Ginibre ensemble}, Comm. Math. Phys. \textbf{378} (2020), 943--985.
 

\bibitem{LMS18} P. Le Doussal, S. N. Majumdar and G. Schehr, \emph{Multicritical edge statistics for the momenta of fermions in nonharmonic traps}, Phys. Rev. Lett. \textbf{121} (2018), 030603.

\bibitem{LS17} T. Leblé and S. Serfaty, \emph{Large deviation principle for empirical fields of log and Riesz gases}, Invent. Math. \textbf{210} (2017), 645--757.
 

\bibitem{LM16} S.-Y. Lee and N. Makarov, \emph{Topology of quadrature domains}, J. Amer. Math. Soc. \textbf{29} (2016), 333--369.

\bibitem{LR16} S.-Y. Lee and R. Riser, \emph{Fine asymptotic behavior for eigenvalues of random normal matrices: ellipse case}, J. Math. Phys. \textbf{57} (2016), 023302.

\bibitem{LY17} S.-Y. Lee and M. Yang, \emph{Discontinuity in the asymptotic behavior of planar orthogonal polynomials under a perturbation of the Gaussian weight}, Comm. Math. Phys. \textbf{355} (2017), 303--338.
 

\bibitem{LY23} S.-Y. Lee and M. Yang, \emph{Strong asymptotics of planar orthogonal polynomials: Gaussian weight perturbed by finite number of point charges}, Comm. Pure Appl. Math. \textbf{76} (2023), 2888--2956.
 

\bibitem{MMO25} J. Marzo, L. Molag and J. Ortega-Cerdà, \emph{Universality for fluctuations of counting statistics of random normal matrices}, J. Lond. Math. Soc. (2) \textbf{113} (2026), e70462.

\bibitem{MWZ00} M. Mineev-Weinstein, P.~Wiegmann and A. Zabrodin, \emph{Integrable structure of interface dynamics}, Phys. Rev. Lett. \textbf{84} (2000), 5106--5109.
 

\bibitem{MR25} S. Mukherjee and Rashmita, \emph{On topology and singularities of quadrature domains}, arXiv:2509.21468.
 

\bibitem{No25} K. Noda, \emph{Partition functions of two-dimensional Coulomb gases with circular root- and jump-type singularities}, arXiv:2510.00843.

\bibitem{NIST} F. W. J. Olver, D. W. Lozier, R. F. Boisvert and C. W. Clark, eds. \emph{NIST Handbook of Mathematical Functions}, Cambridge University Press, Cambridge, 2010.
 

\bibitem{Rou25} N. Rougerie, \emph{Free-energy variations for determinantal 2D plasmas with holes}, arXiv:2510.01745.

\bibitem{ST97} E. B. Saff and V. Totik, \emph{Logarithmic Potentials with External Fields}, Grundlehren der Mathematischen Wissenschaften, Springer-Verlag, Berlin, 1997.

\bibitem{Se23} S. Serfaty, \emph{Gaussian fluctuations and free energy expansion for Coulomb gases at any temperature}, Ann. Inst. Henri Poincaré Probab. Stat. \textbf{59} (2023), 1074--1142.

\bibitem{Se24} S. Serfaty, \emph{Lectures on Coulomb and Riesz Gases}, Amer. Math. Soc. Colloq. Publ. (to appear), arXiv:2407.21194.
 

\bibitem{SW24} J. Sung and Y. Wang, \emph{Quasiconformal deformation of the chordal Loewner driving function and first variation of the Loewner energy}, Math. Ann. \textbf{390} (2024), 4789--4812. 


\bibitem{SY25} S. Shen and J. Yu, \emph{Geometric Zabrodin-Wiegmann conjecture for integer Quantum Hall states}, Comm. Math. Phys. \textbf{406} (2025), 298.

\bibitem{TT03} L. A. Takhtajan and L.-P. Teo, \emph{Liouville action and Weil-Petersson metric on deformation spaces, global Kleinian reciprocity and holography}, Comm. Math. Phys. \textbf{239} (2003), 183--240.

\bibitem{TT06} L. A. Takhtajan and L.-P. Teo, \emph{Weil-Petersson metric on the universal Teichmüller space}, Mem. Amer. Math. Soc. \textbf{183} (2006), 861, viii+119 pp.

\bibitem{TF99} G. Téllez and P. J. Forrester, \emph{Exact finite-size study of the 2D OCP at $\Gamma=4$ and $\Gamma= 6$}, J. Stat. Phys. \textbf{97} (1999), 489--521.

\bibitem{Te09} L.-P. Teo, \emph{Conformal mappings and dispersionless Toda hierarchy}, Comm. Math. Phys. \textbf{292} (2009), 391--415.
 

\bibitem{Wa19b} Y. Wang, \emph{Equivalent descriptions of the Loewner energy}, Invent. Math. \textbf{218} (2019), 573--621.
 

\bibitem{WW19} C. Webb and M. D. Wong, \emph{On the moments of the characteristic polynomial of a Ginibre random matrix}, Proc. Lond. Math. Soc. (3) \textbf{118} (2019), 1017--1056.

\bibitem{WZ00} P. Wiegmann and A. Zabrodin, \emph{Conformal maps and integrable hierarchies}, Comm. Math. Phys. \textbf{213} (2000), 523--538.

\bibitem{WZ22} P. Wiegmann and A. Zabrodin, \emph{Dyson gas on a curved contour}, J. Phys. A \textbf{55} (2022), 165202.
 

\bibitem{ZW06} A. Zabrodin and P. Wiegmann, \emph{Large-$N$ expansion for the 2D Dyson gas}, J. Phys. A \textbf{39} (2006), 8933--8964.
 

\end{thebibliography}

\end{document}